\documentclass[a4paper,reqno,oneside]{amsart}

\usepackage{a4wide}
\usepackage{amsmath}
\usepackage{amssymb}
\usepackage[english]{babel}
\usepackage{color}
\usepackage{enumitem}
\usepackage[breaklinks,colorlinks=true, citecolor=blue]{hyperref}
\usepackage{dsfont}
\usepackage{mathrsfs}
\usepackage{mathtools}
\usepackage{algorithm}
\usepackage[noend]{algpseudocode}
\usepackage{graphicx}
\usepackage{stmaryrd}
\usepackage{thm-restate}
\usepackage[normalem]{ulem}
\usepackage{tikz}
\usetikzlibrary{matrix}
\usetikzlibrary{arrows}
\usetikzlibrary{patterns, shapes, positioning}
\usetikzlibrary{decorations.pathmorphing,shapes}
\usepackage{makecell, multirow} 
\usepackage{setspace}
\usepackage{physics}

\usepackage{comment}
\usepackage{bbm}

\usepackage{qcircuit}

\usepackage[foot]{amsaddr}

\newcommand{\ie}{i.e.,}

\newtheorem{definition}{Definition}

\newtheorem{lemma}{Lemma}

\newtheorem{remark}{Remark}
\newtheorem{theorem}{Theorem}
\newtheorem*{theorem*}{Theorem}

\newtheorem{corollary}{Corollary}

\newcommand*{\eqdef}{:=}

\renewcommand{\vec}[1]{\mathbf{#1}}

\newcommand{\QFT}{\vec{QFT}}

\newcommand{\lwe}{\mathsf{LWE}}

\newcommand{\qlwe}{\mathsf{C}\ket{\mathsf{LWE}}}

\newcommand{\klwe}{\mathsf{knLWE}}
\newcommand{\kmlwe}{\mathsf{kn}\mathsf{MLWE}}

\newcommand{\QLWE}[2]{\mathsf{C}\ket{\mathsf{LWE}({#2})}_{#1}}

\newcommand{\rlwe}{\mathsf{RLWE}}
\newcommand{\mlwe}{\mathsf{MLWE}}

\newcommand{\poly}{\mathsf{poly}}
\newcommand{\negl}{\mathsf{negl}}

\DeclareMathOperator*{\Motimes}{\text{\raisebox{0.25ex}{\scalebox{1}{$\bigotimes$}}}}

\renewcommand{\Pr}{\mathbb{P}}
\newcommand{\Ex}{\mathbb{E}}

\newcommand{\dtr}{D_{\textup{tr}}}
\newcommand{\dstat}{\Delta}

\newcommand{\nprel}{\mathrm{R}}

\newcommand{\NP}{\mathsf{NP}}

\newcommand{\sampler}{\mathcal{S}}

\newcommand{\qsampler}{\mathcal{S}}

\newcommand{\extractor}{\mathcal{E}}
\newcommand{\unif}{U}

\newcommand{\rot}{\mathrm{rot}}

\algnewcommand\Input{\item[\bf Input:]}
\algnewcommand\Output{\item[\bf Output:]}

  \author{Thomas Debris--Alazard $^{1}$}
  \email{thomas.debris@inria.fr}
  \author{Pouria Fallahpour $^{2}$}
  \email{pouria.fallahpour@ens-lyon.fr}
  \author{Damien Stehl\'e $^{2,3}$}
  \email{damien.stehle@cryptolab.co.kr}

  \address{$^{1}$ Inria and  Laboratoire LIX, \'Ecole Polytechnique, Palaiseau, France}
  
  \address{$^{2}$ ENS de Lyon and LIP (CNRS, U. Lyon, ENS de Lyon, Inria, UCBL), Lyon, France}

  \address{$^{3}$ Cryptolab Inc., Lyon, France}

\title[Quantum Oblivious LWE Sampling]{Quantum Oblivious LWE Sampling \\ and   Insecurity  
	of Standard Model Lattice-Based SNARKs}	

\begin{document}	

\maketitle

\vspace*{-0.75cm}
\begin{abstract}
 The Learning With Errors ($\mathsf{LWE}$) problem asks to find~$\mathbf{s}$ from an input of the form $(\mathbf{A}, \mathbf{b} = \mathbf{A}\mathbf{s}+\mathbf{e}) \in  
 (\mathbb{Z}/q\mathbb{Z})^{m \times n} \times (\mathbb{Z}/q\mathbb{Z})^{m}$, for a vector~$\mathbf{e}$ that has small-magnitude entries.  
In this work, we do not focus on solving~$\mathsf{LWE}$ but on the task of sampling instances. As these are extremely sparse in their range, it may seem plausible that the only way to proceed is to first create~$\mathbf{s}$ and~$\mathbf{e}$ and then  
set~$\mathbf{b} = \mathbf{A}\mathbf{s}+\mathbf{e}$. In particular, such an instance sampler knows the solution. This raises the question whether it is possible to obliviously sample $(\mathbf{A}, \mathbf{A}\mathbf{s}+\mathbf{e})$, namely, without knowing the underlying~$\mathbf{s}$. A variant of the assumption that oblivious~$\mathsf{LWE}$ sampling is hard has been used in a series of works to analyze the security of candidate constructions of Succinct Non-interactive Arguments of Knowledge (SNARKs). As the assumption is related to~$\mathsf{LWE}$, these 
SNARKs have been conjectured to be secure in the presence of quantum adversaries.

Our main result is a quantum polynomial-time algorithm that  
samples well-distributed~$\mathsf{LWE}$ instances while provably not knowing 
the solution, under the assumption that~$\mathsf{LWE}$ is hard. Moreover, the approach works for a vast range of $\mathsf{LWE}$ parametrizations, including those used in the above-mentioned SNARKs. This invalidates the assumptions used in their security analyses, although it does not yield attacks against the constructions themselves.

\end{abstract}

\section{Introduction}

The Learning With Errors ($\lwe$) problem~\cite{Regev09} is well-known for its conjectured intractability for quantum algorithms, inherited from 
the conjectured worst-case hardness of specific problems over Euclidean lattices. It has led to abundant cryptographic constructions that are presumably quantum resistant. For three integers~$m \geq n \geq 1$ and~$q \geq 2$ as well as a distribution~$\chi$ over~$\mathbb{Z}/q\mathbb{Z}$ concentrated on values that are small modulo~$q$, the search version of $\lwe$ with parameters~$m,n,q$ and~$\chi$
consists in recovering the secret~$\vec{s}$ from the $\lwe$ instance~$(\vec{A}, \vec{A}\vec{s}+\vec{e}) \in  (\mathbb{Z}/q\mathbb{Z})^{m \times n} \times (\mathbb{Z}/q\mathbb{Z})^m$. In the latter, the matrix~$\vec{A}$ and the vector~$\vec{s}$ are typically uniformly distributed, and each coefficient of~$\vec{e}$ is i.i.d.\ from~$\chi$. In most cases, the dimension~$m$ is polynomial in a security parameter~$\lambda$ and the modulus~$q$ ranges from polynomial to exponential in~$\lambda$; the distribution~$\chi$ is often set as an integer Gaussian of standard deviation parameter~$\sigma \in [\Omega(\sqrt{n}), O(q/\sqrt{n})]$ that is folded modulo~$q$, which will be subsequently denoted by~$\vartheta_{\sigma,q}$. We have~$\vartheta_{\sigma,q}(e) = \sum_{k \in \mathbb{Z}} \exp(- |e+qk|^2 / \sigma^2)$ for all~$e \in \mathbb{Z}$, up to a normalization factor. 
From a hardness point of view, it was shown in~\cite{Regev09} that the~$\lwe$ problem with parameters~$m,n,q,\vartheta_{\sigma,q}$ is quantumly at least as hard as some worst-case lattice problems in dimension~$n$ when~$\sigma \geq \Omega( \sqrt{n})$.
The reduction was later ``dequantized'' by~\cite{Peikert09} under the condition that~$q\geq 2^{\Omega(n)}$. Moreover, the authors of~\cite{BLPRS13} showed that the~$\lwe$ problem with parameters~$m,n^2,q,\vartheta_{\sigma,q}$ is classically at least as hard as some worst-case lattice problems in dimension~$n$ when~$\sigma$ is larger than some polynomial in~$n$. 
From an algorithmic viewpoint,  there is no known solver for~$\lwe$ with runtime lower than~$\exp(   \Omega( n\log n \log q / \log^2(q/\sigma)     )          )$, when~$m$ is polynomially large (see, e.g.,~\cite{HKM18}).

In this work, we do not focus on the hardness of~$\lwe$, 
but on the task of generating $\lwe$ samples. 
Concretely, we consider algorithms~$\mathcal{S}$, which we call $\lwe$ samplers, that take as input a uniform matrix~$\vec{A}$ and output a correctly distributed~$\vec{b} = \vec{A}\vec{s}+\vec{e}$:
\[
 \mathcal{S}_{m,n,q,\chi}: \ \vec{A} \in (\mathbb{Z}/q\mathbb{Z})^{m \times n} \ \ \ \longrightarrow \ \ \ \vec{b} = \vec{A}\vec{s}+\vec{e} \in 
 (\mathbb{Z}/q\mathbb{Z})^{m} \ . 
 \] 
 
For parameters of cryptographic interest, a correctly distributed $\lwe$ pair~$(\vec{A}, \vec{b})$ admits a unique pair~$(\vec{s}, \vec{e})$ that is much more likely than any other pair to satisfy~$\vec{b} = \vec{A}\vec{s}+\vec{e}$. This is provided by~$m$ being sufficiently large as a function of~$n, q$ and~$\chi$. 
As the correctly formed~$\vec{b}$'s are extremely sparse in their range (e.g., exponentially so as a function of~$\lambda$), the naive approach of sampling a uniform~$\vec{b}$  and keeping it  if it has the correct form is prohibitively expensive. Another major bottleneck with this naive approach is that the distinguishing version of~$\lwe$ is no easier than its search version (see~\cite{Regev09} for small values of~$q$ and~\cite{Peikert09,BLPRS13} for large values of~$q$). Given this, it could seem that the only 
way to proceed for a sampler~$\mathcal{S}$ is to first create~$\vec{s}$ and~$\vec{e}$ and then return~$\vec{A}\vec{s}+\vec{e}$. This leads us to the following question:
\begin{center}
\emph{Does there exist an efficient algorithm  that generates $\lwe$ samples \\ 
without knowing the underlying secrets?}
\end{center}

The obliviousness of the sampler can be formalized by considering an extractor algorithm that takes as inputs the sampler's input and 
sampler's random coins and outputs the $\lwe$ secret of the sampler's output: the $\lwe$ sampler is oblivious if no efficient extractor exists. The existence of an oblivious sampler hence implies the hardness of~$\lwe$. 

Variants of the assumption that no such algorithm exists have been introduced to serve as security foundation of 
several cryptographic constructions. An early occurrence was~\cite{LMSV12}, to build a homomorphic encryption scheme 
with security against chosen ciphertext attacks. The precise algebraic framework was different and led to a  quantum 
polynomial-time attack in~\cite{CDPR16}, but the usefulness of the assumption can be explained in the $\lwe$ context as follows. Assume 
a ciphertext corresponds to an $\lwe$ instance~$\vec{b} = \vec{A}\vec{s}+\vec{e}$ belonging to the ciphertext 
space~$(\mathbb{Z}/q\mathbb{Z})^m$, and that the plaintext of a well-formed ciphertext is a function of~$\vec{s}$ (the matrix~$\vec{A}$ is publicly known, and could for example be part of the public key). 
In the context of chosen-ciphertext security, the attacker is allowed to query a decryption oracle on any element in 
the ciphertexts space to extract useful information. In the scheme, if the query is not a well-formed ciphertext, the 
challenger will be able to detect it and reply with a failure symbol. The oblivious sampling hardness assumption 
ensures that if the adversary makes a decryption query on a well-formed~$\vec{b} = \vec{A}\vec{s}+\vec{e}$, 
then the reply to the query does not give it anything more than it already knows. The oblivious sampling hardness assumption 
was used more recently in a series of works building Succinct Non-interactive Arguments of Knowledge (SNARKs) from lattice 
assumptions~\cite{GMNO18,NYI20,ISW21,SSEK22,CKKK23,GNS23}
in the standard model. In this context, the assumption is typically stated for the knapsack variant of~$\lwe$, which asks to recover~$\vec{e}$ from~$\vec{B} \in (\mathbb{Z}/q\mathbb{Z})^{(m-n) \times m}$ and~$\vec{B} \vec{e} \in   (\mathbb{Z}/q\mathbb{Z})^{m-n}$ where~$\vec{B}$ is typically uniform (possibly in a computational indistinguishability sense) and  each coefficient of~$\vec{e}$ is i.i.d.\ from~$\chi$.  We refer to~\cite{MM11} for reductions between $\lwe$ in standard and knapsack forms. As before, we 
are interested  in a regime where there is a single solution. In this formulation, an instance sampler is oblivious if it can create an 
instance~$\vec{B}\vec{e}$ without knowing the solution~$\vec{e}$. In the mentioned SNARK constructions, the non-existence of an efficient
oblivious sampler is used to provide the knowledge soundness property, i.e., to extract a witness from a prover. As these constructions
rely on assumptions related to lattices, they are often conjectured secure even against quantum adversaries.  In terms of parameters, several of those SNARKs require an exponential gap between the noise and the modulus, i.e., a large~$q/\sigma$. For example, one may choose~$q/\sigma = 2^\lambda,\sigma = \poly(\lambda), q = 2^{\Theta(\lambda)}$, and~$n = \Theta(\lambda^2 / \log \lambda)$. For this parametrization, the runtime of the best known algorithm grows as~$\exp(\Omega(\lambda))$.

\medskip
\noindent
{\bf Contributions.}
Our main contribution is a polynomial-time quantum $\lwe$ sampler that we prove  oblivious under the assumption 
that~$\lwe$ is intractable, under very mild parameter restrictions. We  prove the following
result.
\begin{theorem}
\label{th:main-intro}
	Let $m \geq n\geq 1$ and~$q \geq 3$ be integers and $\sigma \geq 2$ be a real number. The parameters~$m,n,q,\sigma$ are functions of the security parameter~$\lambda$ with~$m,\log q \leq \poly(\lambda)$ and~$q$ prime. Assume that the parameters satisfy the following conditions:
	$$
 m \geq n\sigma \cdot \omega (\log \lambda) \quad \mbox{ and } \quad  2 \leq \sigma \leq \frac{q}{ \sqrt{8m\ln q}} \ .
	$$
	 Then there exists a $\poly(\lambda)$-time quantum oblivious $\lwe_{m,n,q,\vartheta_{\sigma,q}}$ instance sampler, 
	 under the assumption that $\lwe_{m,n,q,\vartheta_{\sigma,q}}$ is hard. 
\end{theorem}

The proof technique and result are quite flexible. For example, the secret~$\vec{s}$ can have any efficiently sampleable distribution. 
Also, we will show that obliviousness is preserved through randomized Karp reductions. Then, by using reductions from $\lwe$ with a 
parametrization satisfying the conditions of the statement above to $\lwe$ with a second parametrization, we obtain the existence 
of an efficient quantum oblivious $\lwe$ sampler for the second parametrization, under the assumed hardness of $\lwe$ for the first 
parametrization. We can notably throw away superfluous samples (\ie{} decrease~$m$), take an arbitrary arithmetic shape for~$q$ and
choose larger values for~$\sigma$, by using modulus-dimension switching~\cite{BLPRS13}.

Our result shows that $\lwe$ sampling can be performed quantumly in polynomial-time while the problem seems to be hard classically.
So far, only very few problems related to lattices admit a quantum polynomial-time algorithm while remaining
conjecturally hard for classical algorithms. Notable exceptions include finding a shortest non-zero vector in a lattice corresponding to a principal ideal in a family of number fields such that the ideal contains an unexpectedly short generator~\cite{CDPR16}, and finding a mildly short non-zero vector in a lattice corresponding to an (arbitrary) ideal in a family of number fields~\cite{CDW21}. These exceptions are restricted to (specific) lattices arising from algebraic number theory. 

As a first step, we discuss the notion of oblivious sampling for a quantum algorithm.  A prior definition was put forward in~\cite{LMZ23}.
We propose an alternative definition that, in our opinion, better models what an extractor should be allowed.  
For the class of quantum algorithms that first perform a unitary and then a measurement, we show that these two definitions 
are equivalent. Even though our sampler belongs to that specific class of algorithms, we believe that our new definition is valuable 
as it provides further insight on oblivious sampling.

We then propose a general approach for a quantum oblivious sampling. Namely, we consider quantum  algorithms that, given~$\vec{A}$ as input, first generate a state of the form 
\begin{equation}
\label{eq:qlwe}
  \sum_{\vec{s}, \vec{e}} \left(\prod_i f(e_i)\right) \ket{\vec{A}\vec{s}+\vec{e}} \ , 
 \end{equation}
up to normalization and with~$\vec{e} = (e_1,\ldots,e_m)^{\intercal}$, and then measure this state. 
Here~$f: \mathbb{Z}/q\mathbb{Z} \rightarrow \mathbb{C}$ is a non-zero complex-valued function, and if there are 
auxiliary registers, then they are all set to zero.  The output is indeed an $\lwe$ instance $\vec{A}\vec{s}+\vec{e}$ for
the input matrix~$\vec{A}$. We show that any such quantum algorithm is an oblivious $\lwe$ sampler for the error 
distribution~$\chi$ proportional to~$|f|^2$ (under the assumption that~$\lwe$ is hard).

Next, we modify an algorithm from~\cite{CLZ22} to obtain a quantum algorithm for generating a state as above,  in time polynomial 
in~$m$ and~$\log q$, for a uniformly distributed~$\vec{A}$ and for the folded integer Gaussian distribution~$\vartheta_{\sigma,q}$, \ie{} 
with  $|f|^2 = \vartheta_{\sigma,q}$.

Finally, we consider the application of our result to the SNARK constructions mentioned above. In particular, this requires to adapt 
our analysis of the oblivious sampler to matrices~$\vec{A}$ corresponding to the module version of $\lwe$~\cite{BGV12,LS15}.  
We obtain that the underlying hardness assumption of Linear-Only Vector Encryption does not hold against quantum algorithms. This invalidates the security analyses of several standard model lattice-based SNARKs~\cite{GMNO18,NYI20,ISW21,SSEK22,CKKK23,GNS23}. We stress that this does not break the constructions themselves. For instance, the authors of~\cite{BISW17} mention a different route to analyze their SNARK construction in their Remark~4.9.
Their approach is inspired by~\cite[Lem.~6.3]{BCI13} and can be applied to some of the constructions mentioned above.

\subsection{Technical overview}
We now go into further detail for each one of the contributions. 

\subsubsection{Defining oblivious sampling for quantum algorithms.}
Let us first recall the classical notion of oblivious $\lwe$ sampling. Note that the discussion below could 
be generalized to more problems than~$\lwe$, but we focus on~$\lwe$ for the sake of simplicity.  Let~$\mathcal{S}$ be an $\lwe$ sampler,
taking as input a matrix~$\vec{A}$ and returning a vector~$\vec{b} = \vec{A}\vec{s}+\vec{e}$. To capture the notion of 
obliviousness, we consider extractor algorithms that can observe the behaviour of~$\mathcal{S}$. More concretely, an extractor~$\mathcal{E}$ is an algorithm that has access to the description of~$\mathcal{S}$, its input~$\vec{A}$ and its internal 
randomness~$\rho_\mathcal{S}$ (which implies that~$\mathcal{E}$ also knows the output~$\vec{b}$). The extractor can also use random coins~$\rho_{\mathcal{E}}$ of its own. Finally, it is requested to output~$\vec{s}$. We say that~$\mathcal{S}$ is an oblivious $\lwe$ sampler if no efficient extractor~$\mathcal{E}$ succeeds with non-negligible probability over the choice of~$\vec{A}$, $\rho_\mathcal{S}$ and~$\rho_{\mathcal{E}}$.

The main difficulty that emerges in the quantum setting stems from measurements. They add inherent randomness to the computation 
that is not  extractable, while classically, the randomness comes from an a priori given random string. In~\cite{LMZ23}, the authors proposed an adaptation of extractability to the quantum setting that aims at handling this issue. 
By arguing that any quantum algorithm can be generically transformed into another one that first starts by a unitary transformation and then performs a measurement (possibly not on all its registers), the authors of~\cite{LMZ23} consider only such quantum samplers to define extractability.
In their definition of extraction, the sampler is first executed until it performs its  measurement,
and then the measurement outcome and remaining registers are handed over to the extractor. More formally, the extractor~$\mathcal{E}$ is a quantum algorithm that is given as inputs the description of the quantum sampler~$\mathcal{S}$, the input matrix~$\vec{A}$,  the output~$\vec{b} = \vec{A}\vec{s}+\vec{e}$ and  the auxiliary registers of~$\mathcal{S}$
(the extractor may also have auxiliary ancillas of its own). Again, we are interested in the existence of efficient samplers of that form that 
output~$\vec{s}$ with non-negligible probability. The authors justify  as follows that it is a definition that is consistent with the classical setting. It is first observed that unitary algorithms are reversible. In the classical setting, every algorithm can be turned into a reversible one. By having the output of the reversible algorithm, one can find the input which contains the randomness. Therefore, giving the output of the reversible sampler to the extractor is equivalent to giving its input and the randomness to the extractor. 

We propose an alternative definition, to allow the extractor to more closely look at the behaviour of the sampler. Indeed, it may seem overly restrictive to forbid the extractor from looking at the sampler's execution itself. Furthermore, the fact that any quantum computation can be converted  into a unitary-then-measurement sampler cannot be applied in the extractability context, as it is not a priori
excluded that one would be able to extract the secret from the complex form of the algorithm and not from the compiled form and vice-versa. Our definition
aims at handling these two limitations of the~\cite{LMZ23} definition.  
The main principle we use  is that observing or measuring the execution of a machine (classical or quantum) must not change too much the view that the sampler has of itself. Assume that an extractor is observing a sampler. Let~$\rho_{\mathsf{Q} \otimes \mathsf{E}}$ represent the joint state of the sampler~$\qsampler$ and the extractor~$\extractor$ at some step of the execution. The extractor might have carried out particular inspections that ended up in entangling its register with that of the sampler, so the state~$\rho_{\mathsf{Q} \otimes \mathsf{E}}$ might not be separable. We intuitively expect from a valid extractor that if we trace out its register, the remaining state must be as if the extractor
was not inspecting the sampler at all, or as if it was modifying  the behaviour of the sampler in a negligible manner. Namely, if~$\rho_\mathsf{Q}$ was the state of an isolated sampler at a specific step, and $\rho_{\mathsf{Q} \otimes \mathsf{E}}$ is the joint state of the sampler and extractor at the same step, we require that
$\tr_\mathsf{E}(\rho_{\mathsf{Q} \otimes \mathsf{E}})$ is close  to~$\rho_\mathsf{Q}$ for the trace distance.

We show that these two definitions are equivalent in the case of unitary-then-measurement samplers. However, our definition handles 
more general samplers, making it easier for the adversary to design an oblivious sampler, and hence providing a  stronger notion oblivious sampling hardness. It turns out that our oblivious sampler is of the unitary-then-measurement type, so the definitional discrepancy is not critical to our result. We however believe that our definition provides further insight into the set of operations that  an extractor should be allowed to perform.

As an additional contribution on oblivious sampling, we show that obliviousness is preserved under black-box Karp reductions between distributional problems. Let us consider the following scenario in our $\lwe$ setting. Assume that there is a reduction~$\mathcal{A}$  from an $\lwe$ variant $\lwe_1$ to another $\lwe$ variant $\lwe_2$ such that $(i)$~an instance of~$\lwe_1$ is mapped to an instance of $\lwe_2$ (with appropriate distribution over the randomness of the $\lwe_1$ instance and the random coins of $\lwe_2$) and~$(ii)$ a solution to the $\lwe_1$ instance can be obtained from a solution to the $\lwe_2$ instance. Then applying the reduction~$\mathcal{A}$ to an oblivious $\lwe_1$ sampler gives an oblivious~$\lwe_2$ sampler. In our case, this observation will prove useful to weaken parameter constraints on $\lwe$ for oblivious sampling: we will first obtain an oblivious sampler for some restricted parametrization of $\lwe$ and extend it to more general setups thanks to existing such reductions.

\subsubsection{Reducing oblivious $\lwe$ sampling to $\qlwe$.}

Assume we have a (classically) known matrix~$\vec{A} \in (\mathbb{Z}/q\mathbb{Z})^{m \times n}$, and that we manage to build
the quantum state from Equation~\eqref{eq:qlwe} using a unitary transformation (with possibly auxiliary registers equal to zero).
Creating such a state was studied in~\cite{SSTX09} and referred to as the $\qlwe$ problem in~\cite{CLZ22}. 
We show that oblivious $\lwe$ sampling reduces to $\qlwe$. Intuitively,
a measurement of the state above provides an $\lwe$ sample~$\vec{A}\vec{s}+\vec{e}$ for a uniformly 
distributed~$\vec{s}$ and a vector~$\vec{e}$ with distribution~$\chi$ proportional to~$|f|^2$: there is no reason for a specific~$\vec{s}$ to 
be privileged, and this algorithm does not seem to have any additional knowledge about the $\lwe$ solution.  
We formalize this intuition using the obliviousness sampling definitions discussed above (which coincide here, as we have a unitary-then-measure algorithm). The result also holds if we add non-constant phases for~$\vec{s}$ (for example to obtain~$\vec{s}$ that is uniform among those with binary coordinates). It also allows the parameter~$m$ from $\qlwe$  to be larger than the one we want for oblivious $\lwe$ sampling, as we may throw away the superfluous coordinates without compromising the obliviousness. 

We now discuss two existing approaches for solving $\qlwe$. The first one, derived from the $\lwe$ hardness proof from~\cite{Regev09}, is to generate
the following quantum state
\begin{equation}
\label{eq:regev}
  \sum_{\vec{s}, \vec{e}} \left(\prod_i f(e_i)\right)  \ket{\vec{s}}\ket{\vec{A}\vec{s}+\vec{e}} \ ,
\end{equation}
up to normalization, and then to uncompute~$\ket{\vec{s}}$ from $\ket{\vec{A}\vec{s}+\vec{e}}$. Generating this state can 
be done efficiently (possibly under some conditions on~$f$) by first creating the superposition over all~$\vec{s}$ and~$\vec{e}$ of~$\ket{\vec{s}}\ket{\vec{e}}$ with proper amplitudes, and then multiplying the first register by~$\vec{A}$ to add it to the second one. To remove~$\vec{s}$, i.e., to replace~$\vec{s}$ by~$\vec{0}$ in the first register,
the approach from~\cite{Regev09} is to recover~$\vec{s}$ from the second register and subtract it to the first one, by using 
a quantized $\lwe$ solver. This leads to a reduction from~$\qlwe$ to~$\lwe$. Unfortunately, in our context, this is not satisfactory, 
as $\lwe$ must be assumed difficult for oblivious $\lwe$ sampling to be feasible.  

Another approach for solving~$\qlwe$ was recently proposed in~\cite[Sec.~5]{CLZ22}. The proposed algorithm does not require
any oracle for a presumably hard problem, but seems restricted to specific parametrizations of~$\qlwe$, as we discuss below. 
\begin{itemize}\setlength{\itemsep}{5pt}
\item[$\bullet$] First, it builds the quantum state from Equation~\eqref{eq:regev}. It can be rewritten as follows:
\begin{eqnarray*}
  \sum_{\vec{s}, \vec{e}} \Motimes\limits_{i \leq m}  \ket{\vec{s}} f(e_i) \ket{\langle \vec{a}_i,\vec{s} \rangle + e_i}  &   =  &
  \sum_{\vec{s}} \ket{\vec{s}}\left( \Motimes\limits_{i \leq m} \sum_{e_i} f(e_i) \ket{\langle \vec{a}_i,\vec{s} \rangle + e_i}\right)  \\ 
  & = &  \sum_{\vec{s}}\ket{\vec{s}} \left(\Motimes\limits_{i\leq m} \ket{\psi_{\langle \vec{a}_i,\vec{s} \rangle}}\right)  \ ,
\end{eqnarray*}
where~$\ket{\psi_k} = \sum_{e} f(e) \ket{k+e}$ for all~$k \in \mathbb{Z}/q\mathbb{Z}$ (up to normalization). 
\item[$\bullet$] Second, it individually considers 
all sub-registers $\ket{\psi_{\langle \vec{a}_i,\vec{s} \rangle}}$ of the $\ket{\vec{A}\vec{s}+\vec{e}}$ register, and performs a measurement for each one of them. For each~$i$, the measurement consists in sampling a uniform~$k_i \in \mathbb{Z}/q\mathbb{Z}$ and applying a projective measurement with respect to the (normalized) Gram-Schmidt orthogonalization of~$\ket{\psi_{k_i+1}}, \ket{\psi_{{k_i}+2}}, \ldots, \ket{\psi_{k_i-1}}, \ket{\psi_{k_i}}$ (we assume that the $\ket{\psi_j}$'s are linearly independent, and the indices are taken modulo~$q$). If the measurement is on the last direction, then $\ket{\psi_{\langle \vec{a}_i,\vec{s} \rangle}}$ cannot have a component on the span of~$\ket{\psi_{k_i+1}}, \ket{\psi_{{k_i}+2}}, \ldots, \ket{\psi_{k_i-1}}$ and must be equal to $\ket{\psi_{k_i}}$. When successful, the measurement indicates that~$\langle \vec{a}_i,\vec{s} \rangle = k_i \bmod q$.

\item[$\bullet$] Third, sufficiently many successful measurements are collected through the different values of~$i$, to obtain many equations of the type~$\langle \vec{a}_i,\vec{s} \rangle = k_i$, where the~$\vec{a}_i$'s and~$k_i$'s are known. This is then fed to a quantized Gaussian elimination algorithm (recall that we are working with a superposition over all~$\vec{s}$'s). The latter
outputs~$\vec{s}$, which is then subtracted from the first register.
\end{itemize}
It was proved in~\cite{CLZ22} that this algorithm solves $\qlwe$ in polynomial time,  if~$m$ and $q$ are polynomial in the security parameter~$\lambda$, if a state proportional to 
$\sum_e f(e) \ket{e}$ can be efficiently computed, and if
\[
m = \frac{n}{p^{\mathrm{CLZ}}} \cdot \omega( \log \lambda) \ \ \mbox{ with } 
\ \ p^\mathrm{CLZ} = \frac{ \min_x |\widehat{f}(x)|^2}{q} \ . 
\] 
Here, the notation~$\widehat{f}$ refers to the Fourier transform over~$\mathbb{Z}/q\mathbb{Z}$ of~$f$. The 
quantity~$p^\mathrm{CLZ}$ corresponds to the probability that an individual measurement of the second step succeeds. 
This result notably allows to solve $\qlwe$ (and hence oblivious $\lwe$ sampling) 
for~$q$ polynomial and~$\chi$ set as the uniform distribution in an interval~$[-B,B]\cap \mathbb{Z}$, for 
any~$B \in \mathbb{Z}$ such that~$0<2B+1<q$ and~$\mathrm{gcd}(2B+1,q) =1$. Interestingly, by taking~$q=2$, the result also allows to solve the adaptation of~$\qlwe$ to the decoding problem for uniform binary codes (also known as Learning Parity with Noise), and to obliviously sample points near codewords for the Bernoulli distribution with an arbitrary Bernoulli parameter in~$(0,1/2)$.

This result has two limitations. First, the lower bound on~$m$ and the run-time both grow at least polynomially with~$q$ 
(note that~$\min |\widehat{f}| \leq 1$), which prevents
us from choosing an exponential~$q$. This is in part due to the uniform guess of $\langle \vec{a}_i,\vec{s} \rangle$ in the 
measurement, which directly incurs a loss by a factor~$q$ in the success probability of each individual measurement. Note that for a fixed standard deviation~$\sigma$, $\lwe$ becomes no harder as~$q$ increases, so that one can expect that it is indeed no easier to obliviously sample~$\lwe$ instances (intuitively, the easier is the considered problem, the harder it is to obliviously sample instances).
Most SNARKs that we consider use an exponential~$q$. 
Second, the quantity~$\min |\widehat{f}|$ is extremely small for a wide range of amplitude functions, notably~$f =  \sqrt{\vartheta_{\sigma,q}}$ up to a normalization factor (recall that~$\vartheta_{\sigma,q}$ denotes the folded discrete Gaussian distribution). Indeed, we prove in Lemma~\ref{lemma:noPhaseDG} that in that case and if~$\sigma \geq 1$, we 
have
\[
q \cdot \min |\widehat{f}|^2 \leq 32 \sigma \cdot \exp\left(-\min\left(\frac{\pi \sigma^2}{4},\frac{q^2}{4\sigma^2}\right)\right) \ . 
\]
This expression is most often extremely small. For example, for~$\sigma = \Omega(\sqrt{n})$ and~$q = \Omega(\sqrt{n} \sigma)$, the expression is~$2^{-\Omega(n)}$.  
This prevents from meaningfully using the result for the discrete Gaussian distribution, which is the most common choice of error distribution for~$\lwe$.

As a remark, we would like to highlight that in the reduction from oblivous~$\lwe$ sampling to the~$\qlwe$ problem, a crucial step consists in erasing the memory (up to some negligible error) that has been used during the course of computation. Constructing a superposition of the possible values for the secret~$\vec{s}$ in a separate register is neccessary for the algorithm to succeed, while the last step to revert it back to~$\ket{\vec{0}}$ is pivotal for obtaining obliviousness. We note that in designing classical oblivious samplers, forgetting the history of the computation remains a challenging obstacle.

\subsubsection{Measuring with increased success probability.}

The second step of the algorithm from~\cite{CLZ22} consists in taking~$\ket{\psi_k}$ for an unknown~$k \in \mathbb{Z}/q\mathbb{Z}$ as
input and returning~$k$, i.e., it aims at distinguishing the quantum states~$\ket{\psi_0},\ldots,\ket{\psi_{q-1}}$. 
One could proceed as follows if the states were orthogonal. Consider the well-defined projective measurement $(\vec{E}_{i})_i$ defined 
by~$\vec{E}_{i} = \ketbra{\psi_{i}}{\psi_{i}}$ for $0 \leq i <q$. Then, if the state $\ket{\psi_{k}}$ is given, the probability 
to see~$k$ as the outcome is~$\bra{\psi_{k}}\vec{E}_{k}^{\dagger} \vec{E}_{k}^{\phantom{\dagger}}\ket{\psi_{k}} = 1$. 
In other words, this quantum measurement perfectly distinguishes the quantum states. 
However, when the $\ket{\psi_{k}}$'s are not orthogonal (which is our case except for particular amplitudes like~$f$ being~$1$ in~$0$ and $0$ elsewhere), 
it is known that there exists no quantum measurement to perfectly distinguish them (see \cite[Box~2.3]{ChuangNielsen}). 
The measurement from~\cite{CLZ22} may output a special symbol~$\bot$ representing the ``unknown'' answer, but it does not 
make any mistake, in the sense that it never outputs some~$\ell \in \mathbb{Z}/q\mathbb{Z}$ different from~$k$.  
Such a process is referred to as unambiguous. This property is important for the subsequent Gaussian elimination step, as it requires 
all linear equations to be correct.  We define the error parameter of the unambiguous measurement as 
the maximal probability that the measurement outputs~$\bot$ over all possible input states:
\[
p_\bot = \max\limits_{k} \bra{\psi_k}\vec{E}_\bot\ket{\psi_k} \ ,
\]
where~$\vec{E}_{\bot}$ corresponds to the outcome~$\bot$. The measurement from~\cite{CLZ22} satisfies 
\[
1-p_\bot^{\mathrm{CLZ}} = p^{\mathrm{CLZ}} = \frac{\min |\widehat{f}|^2 }{ q} \ .
\]

We propose to change the unambiguous measurement by the positive operator-valued measure (POVM) from~\cite{CB98}. It is 
known to be ``optimal'' when the $\ket{\psi_{i}}$'s are 
symmetric and linearly independent, in the sense that it minimizes the error parameter~$p_\bot$ over all possible choice of POVMs. Here, symmetric means that there exists a unitary~$\vec{U}$ such that~$\ket{\psi_i} = \vec{U} \cdot \ket{\psi_{i-1 \bmod q}}$ for all~$0 \leq i <q$: our states indeed satisfy this property with $\vec{U}$ being the mod-$q$ translation operator. The linear independence property may or may not be satisfied, depending on the choice of~$f$ (this is a difficulty encountered in other sections of~\cite{CLZ22}).
The measurement from~\cite{CB98} is defined as follows: 
\[
\forall 0 \leq i < q: \vec{E}_i  =  \alpha \cdot  \ket{\psi_i^\perp} \bra{\psi_i^\perp} \ \ \mbox{ and  }  \ \  
\vec{E}_\bot  =  \vec{I} - \sum_i \vec{E}_i \ ,
\]
where~$\ket{\psi_i^\perp}$ is a unit vector orthogonal to all~$\ket{\psi_j}$'s for~$j \neq i$, for all~$0 \leq i < q$. The scalar~$\alpha$ is chosen 
maximal such that the POVM is well-defined, i.e., such that~$\vec{E}_\bot$ is non-negative: it is the inverse of the largest eigenvalue 
of~$\sum_i \vec{E}_i$. We compute that this measurement leads to:
\[
1-p_{\bot}^{\mathrm{CB}} = \frac{q^{2} \; \alpha}{\sum_{x \in \mathbb{Z}/q\mathbb{Z}} \left| \widehat{f}(x) \right|^{-2}} =  q \cdot \min \left| \widehat{f} \right|^{2} \ .
\]
Note that the success probability of the measurement is a factor~$q^2$ higher than the one from~\cite{CLZ22}. By the union bound (and still assuming~$q$ prime), with the optimal unambiguous measurement, it suffices to set
\begin{equation}
\label{eq:m_lower_bound}
m =  \frac{n}{q\cdot \min |\widehat{f}|^2} \cdot \omega(\log \lambda) \ . 
\end{equation}

\subsubsection{How to implement the measurement in time polynomial in~$\log q$} Compared to the~\cite{CLZ22} approach, the quantum distinguishing measurement from~\cite{CB98} allows one to choose~$m$ smaller by a factor~$q^{2}$ and also 
to gain a factor~$q^{2}$ in the run-time. But beyond these considerations over the parameter~$m$, that we discuss more deeply in Subsection~\ref{subsect:FT}, recall that we are looking for a sampler whose run-time is polynomial in~$m$ and~$\log q$. We therefore have to efficiently implement the above POVM. Although the measurement was introduced in~\cite{CB98}, this work does not specify how to efficiently compute it. The POVM components $\left( \vec{E}_{j} \right)_{0 \leq j < q}$ turn out to be some projections~$\left( \ketbra{\psi_j^{\perp}}{\psi_i^{\perp}}\right)_{0 \leq j < q}$. A first approach would be to compute the quantum states~$\ket{\psi_{i}^{\perp}}$'s, which are given by (here~$\omega_q$ refers to a primitive $q$-th root of unity.)
$$
\forall j: \ket{\psi_j^{\perp}} = \frac{1}{\sqrt{N}} \sum_{x \in \mathbb{Z}/q\mathbb{Z}} \omega_q^{-jx} \cdot  \overline{\widehat{f}(-x)^{-1}}\; \ket{\chi_x}
$$
where $N = \sum_{x \in \mathbb{Z}/q\mathbb{Z}} |\widehat{f}(-x)|^{-2}$ and $\left(\ket{\chi_x}\right)_{x \in \mathbb{Z}/q\mathbb{Z}}$ denotes the Fourier basis, namely 
$$
\forall x: \ket{\chi_x} = \frac{1}{\sqrt{q}} \sum_{y \in \mathbb{Z}/q\mathbb{Z}} \omega_q^{xy} \ket{y} \ .
$$ 
However, even if one were able to compute these quantum states in polynomial time, there are~$q$ of them, making it difficult to obtain 
a run-time polynomial in~$\log q$. One may therefore try to find a way to efficiently compute a unitary sending~$\ket{j}\ket{0}$ 
to~$\ket{j}\ket{\psi_{j}^{\perp}}$ for all~$j$. This seems to be a challenging path, as such a unitary would need to implement the POVM. 
Let us backtrack a little, and try to  see how the POVM given by $(\vec{E}_{j})_{0 \leq j < q}$ and $\vec{E}_{\bot}$ acts on the $\ket{\psi_{j}}$'s. 
First, let us decompose the $\ket{\psi_j}$'s in the Fourier basis:
$$
\forall j : \ket{\psi_j} = \sum_{x \in \mathbb{Z}/q\mathbb{Z}} \widehat{f}(-x)\cdot\omega_q^{- jx}\;\ket{\chi_{x}} \ .
$$ 
To correctly identify~$\ket{\psi_j}$, we project it according to $\vec{E}_{j} = \ketbra{\psi_j^{\perp}}{\psi_{j}^{\perp}}$. This leads to considering the following Hermitian product:
\begin{align*}
	\braket{\psi_j^{\perp}}{\psi_j} &= \frac{1}{\sqrt{N q^{n}}} \sum_{x \in \mathbb{Z}/q\mathbb{Z}} \omega_q^{ j x}\cdot \widehat{f}(-x)^{-1} \cdot \widehat{f}(-x) \cdot \omega_q^{- jx} \\
	&= \frac{1}{\sqrt{N q^{n}}} \sum_{x \in \mathbb{Z}/q\mathbb{Z}} \widehat{f}(-x)^{-1} \cdot \widehat{f}(-x)
\end{align*}
In other words, when projecting $\ket{\psi_j}$ on~$\ket{\psi_j^{\perp}}$, we want to ``remove'' $\widehat{f}(-x)$ in the amplitudes of $\ket{\psi_{j}}$ in its Fourier basis decomposition. Therefore, simulating the POVM of~\cite{CB98} leads to considering the unitary performing this task, \ie{} a unitary~$\vec{V}$ such that
$$
\forall x:  \ket{\chi_x}\ket{0}  \longmapsto \frac{\min |\widehat{f}|}{\widehat{f}(-x)} \ket{\chi_x}\ket{0} + \sqrt{1- \left|\frac{\min |\widehat{f}|}{\widehat{f}(-x)}\right|^{2}}\ket{\chi_x}\ket{1} 
$$
(We note that a similar approach was considered in~\cite{CT23}.)
Such a unitary is efficiently computable under the conditions that both~$\min |\widehat{f}|$ and~$\widehat{f}(-x)$ can be efficiently approximated. Let us check that~$\vec{V}$ indeed ``simulates'' the measurement from~\cite{CB98}, by  computing 
how it acts on the $\ket{\psi_j}$'s:
\begin{align*}
	\vec{V} \left(\ket{\psi_j}\ket{0} \right)&= \vec{V} \left( \sum_{x \in \mathbb{Z}/q\mathbb{Z}} \widehat{f}(-x) \cdot \omega_q^{- jx} \;  \ket{\chi_x}\ket{0} \right) \\
	&= \sum_{x \in \mathbb{Z}/q\mathbb{Z}} \left( \min |\widehat{f}| \cdot \omega_q^{- jx} \; \ket{\chi_x}\ket{0} + \widehat{f}(-x)\cdot\omega_q^{-jx} \cdot \sqrt{1- \left|\frac{\min |\widehat{f}|}{\widehat{f}(-x)}\right|^{2}}\;\ket{\chi_x}\ket{1} \right)  \\
	&= \sqrt{q} \cdot \min |\widehat{f}|\; \ket{j}\ket{0} + \sum_{x \in \mathbb{Z}/q\mathbb{Z}} \widehat{f}(-x)\cdot\omega_q^{- jx}\cdot \sqrt{1- \left|\frac{\min |\widehat{f}|}{\widehat{f}(-x)}\right|^{2}}\;\ket{\chi_x}\ket{1} 
\end{align*}
In other words, we have (for some quantum state $\ket{\eta_{j}}$):
$$
\vec{V} \ket{\psi_j}\ket{0} = \sqrt{p^{\mathrm{CB}}} \ket{j}\ket{0} +  \sqrt{1-p^{\mathrm{CB}}}\ket{\eta_{j}}\ket{1}  \ , 
$$
 where $p^{\mathrm{CB}}$ turns out to be equal to the success probability of the POVM given in \cite{CB98}, \ie{} ~$p^{\mathrm{CB}} = 1-p_{\bot}^{\mathrm{CB}}$. Therefore, by interpreting any quantum state whose last qubit is~$\ket{1}$ as~$\bot$, applying~$\vec{V}$ amounts to quantumly  recovering~$j$ from~$\ket{\psi_j}$ with probability $1-p_{\bot}^{\textup{CB}}$.

\subsubsection{Increasing the Fourier coefficients}\label{subsect:FT}
At this stage, we have that the modified $\qlwe$ algorithm is polynomial in~$m$ and~$\log q$. We also have decreased the feasibility threshold on~$m$ from~$n  q / \min |\widehat{f}|^2 \cdot \omega(\log \lambda)$ to~$n  / (q \cdot \min |\widehat{f}|^2) \cdot \omega(\log \lambda) $.
However, as observed in~\cite{CLZ22}, the quantity $\min |\widehat{f}|^2$ can be extremely low for distributions of interest. 

Our last technical ingredient stems from the observation that for the purpose of oblivious $\lwe$ sampling for a distribution~$\chi$,
we do not need to set~$f = \sqrt{\chi}$ but can set~$f = \sqrt{\chi} \cdot u$ for any 
function~$u: \mathbb{Z}/q\mathbb{Z} \rightarrow \mathbb{C}$ taking values on the unit circle.
  Indeed, the new phases disappear 
when we measure the $\qlwe$ state to obtain the $\lwe$ sample. 
Interestingly, the phases can greatly help to 
increase~$\min |\widehat{f}|^2$. The astute reader will note that the circuit described above then 
needs to be updated to account for the phases, but we show that efficiency can be preserved, notably for the function~$u$ that we choose.  We  propose to set~$u$ as the sign function:
\[
\forall x \in \mathbb{Z}\cap[0,q/2]:  u(x) = 1 \ \ \mbox{ and } \ \
\forall x \in \mathbb{Z}\cap(-q/2,0):  u(x) = -1 \ .
\]
Then the following relations hold, for~$q$ odd and for all~$x \in \mathbb{Z}/q\mathbb{Z}$ viewed as an integer 
in~$(-q/2,q/2]$:
\begin{eqnarray*}
\widehat{f}(x) & = & \frac{1}{\sqrt{q}} \sum_{y \in \mathbb{Z}/q\mathbb{Z}} f(y) \cdot \omega_q^{xy} \\
&=  &  \frac{1}{\sqrt{q}} \sum_{\mathbb{Z}\cap[0,q/2 ]} \sqrt{\chi(y)} \cdot \omega_q^{xy} -  
\frac{1}{\sqrt{q}} \sum_{\mathbb{Z}\cap(- q/2,0)} \sqrt{\chi(y)} \cdot \omega_q^{xy}  \\ 
&=  & \frac{\sqrt{\chi(0)}}{\sqrt{q}}  
+ \frac{1}{\sqrt{q}} \sum_{\mathbb{Z}\cap(0, q/2  ]} \sqrt{\chi(y)} \cdot (\omega_q^{xy} - \omega_q^{-xy}) \ .
\end{eqnarray*}
Note that the summand is an imaginary number and hence that $\sqrt{\chi(0)/q}$ is the real part of $\widehat{f}(x)$. As a result, we obtain that~$\min |\widehat{f}| \geq \sqrt{\chi(0)/q}$. By combining with Equation~\eqref{eq:m_lower_bound}, it suffices to set~$m = n / \chi(0) \cdot \omega(\log\lambda)$. 
For the specific case of the folded 
integer Gaussian distribution, we have that~$\chi(0) \approx 1/\sigma$, leading to an efficient algorithm when~$\sigma$ is polynomial in~$\lambda$.

We stress that we use both the phases and the improved unambiguous measurement  to obtain an efficient algorithm. We already saw  
that the the improved measurement alone is insufficient. Conversely, if we use the phases and the measurement 
from~\cite{CLZ22}, then it seems that we need~$m$ to grow as~$nq^2/\sigma \cdot \omega(\log \lambda)$, which forbids a run-time 
polynomial in~$\log q$.

\subsubsection{Application to  standard model lattice-based SNARKs.}
Succinct Non-Interactive Arguments of Knowledge (SNARKs) are cryptographic schemes whose purpose is to 
prove~$\NP$ statements with a succinct proof and fast verification, as a function of the statement size.  
They must satisfy the property of knowledge soundness: informally speaking, if a malicious prover manages to build a proof that passes verification,  then one can extract from its description and execution a valid witness for the proved statement. 
Several candidate SNARKs based on lattices in the standard 
model~\cite{GMNO18,NYI20,ISW21,SSEK22,CKKK23,GNS23} assume the hardness of some type of knowledge assumption, \ie{} 
an assumption that formalizes the intuition that an algorithm cannot achieve a given task without knowing a specific information. This intuition is formalized using extractor algorithms. The specific knowledge assumptions used in those schemes are typically defined
in terms of $\lwe$-based ciphertexts (also sometimes called encodings) of a symmetric encryption scheme. 

To simplify the discussion, 
we now focus on the constructions from~\cite{ISW21,SSEK22,CKKK23}. The discussion can be adapted to the other schemes (see Section~\ref{sse:othersnarks}). The corresponding encryption scheme handles 
plaintexts defined modulo an integer~$p$, with ciphertexts that are vectors modulo a much larger integer~$q$, such that the scheme
enjoys a linear homomorphism property: given~$y_1, \ldots, y_m \in \mathbb{Z}/p\mathbb{Z}$ and ciphertexts $\vec{ct}_1, \ldots, \vec{ct}_m$ decrypting to~$a_1, \ldots, a_m$, the vector~$\sum_i y_i \vec{ct}_i$ decrypts to~$\sum y_i a_i$.  It is then assumed that the
only way to compute a valid ciphertext is to take a linear combination of the available ciphertexts (variants may be used in different schemes). To obtain SNARKs, this is formalized in terms of the existence of an efficient extractor: given the~$\vec{ct}_i$'s, the description 
of the algorithm producing a new ciphertext and its internal randomness, some efficient extractor recovers scalars~$y_i$'s modulo~$p$ such that~$\vec{ct} = \sum_i y_i \vec{ct}_i$. 

We observe that the knowledge assumptions involved in those schemes can be expressed in terms of the knapsack version of~$\lwe$. 
The $\klwe$  problem asks to recover~$\vec{e}$ from the 
input~$(\vec{B},\vec{Be})$ where~$\vec{B}$ is a uniformly chosen matrix from~$(\mathbb{Z}/q\mathbb{Z})^{(m-n) \times m}$, for some integers~$m > n \geq 1$ and~$q \geq 2$. We are in a regime of parameters where~$\vec{e}$ is uniquely determined from~$\vec{Be}$,
with overwhelming probability over the uniform choice of~$\vec{B}$. We identify the matrix~$\vec{B}$ with the 
matrix~$(\vec{ct}_1,\ldots,\vec{ct}_m)$. Note that it is not uniform, we can pretend it is as it is computationally indistinguishable from uniform under some~$\lwe$ parametrization. We then argue that the knowledge assumption is quantumly broken, by observing that our 
witness-oblivious quantum~$\lwe$ sampler can be turned into a witness-oblivious $\klwe$ sampler by relying on the randomized Karp reduction from~$\lwe$ to~$\klwe$ from~\cite{MM11}. 
As some of the considered schemes rely on algebraic variants of~$\lwe$, such as Ring-$\lwe$~\cite{SSTX09,LPR10} or 
Module-$\lwe$~\cite{BGV12,LS15}, we extend the witness-oblivious $\lwe$ sampler to those settings. The analysis extends without
difficulty, except for difficulties arising from the fact that the considered rings are not fields.

\subsection{Related works}
Positive and negative results on the possibility of obliviously sampling in the image of a one-way function have been given 
in~\cite{BCPR16}. Even though $\lwe$ is a (conjectured) one-way function, this work is incomparable to ours as it 
considers the situation where the sampler and the extractor are both given some auxiliary input.  

From a technical perspective, one of our main ingredients is to replace an unambiguous discrimination measurement used in~\cite{CLZ22} by the optimal  unambiguous discrimination measurement from~\cite{CB98}. This change was also considered recently in~\cite{CT23} for linear codes and the Hamming metric, with the objective of  obtaining an efficient quantum algorithm to find short non-zero elements in  linear codes via the framework given by \cite{CLZ22}.

The obliviousness sampling hardness assumption belongs to the family of non-falsifiable assumptions~\cite{Naor03,GW11}, similar to the knowledge of exponent assumption in the discrete logarithm setting~\cite{Dam91}. In the context of succinct non-interactive arguments, such assumptions seem difficult to avoid, as it was proved in~\cite{GW11} that non-falsifiable  assumptions must be used to prove security if one resorts to black-box reductions. Using non-falsifiable assumptions makes the cryptanalyst's task more difficult: as the attack cannot be efficiently tested by a 
challenger, the cryptanalyst is required to prove (or convincingly argue) that the attack works. One way to circumvent this impossibility result is to rely on the random oracle model, which indeed leads to efficient  succinct non-interactive arguments~\cite{BCS16}. We stress that our
algorithm has an impact on the quantum security of standard model lattice-based SNARKs, but not on~\cite{AFLN23} which relies on the random oracle model. We also stress that our algorithm does not seem applicable to known standard model lattice-based SNARGs, 
which differ from SNARKs in their definition of soundness.

Beyond those studied in this work, another candidate construction of a standard model lattice-based SNARK was proposed in~\cite{ACLMT22}, but the underlying assumption was broken classically in~\cite{WW23}.

\section{Preliminaries}

{\bf \noindent Notation.}  
The function~$\ln$ refers to the logarithm in base~$\mathrm{e}$. When the base of the logarithm does not matter, we write~$\log$.

We will  consider the additive group $\mathbb{Z}/q\mathbb{Z}$ for~$q \geq 2$ and may write its elements as
$$
\mathbb{Z}/q\mathbb{Z} = \left\{ j \in \mathbb{Z} \; : \; -\frac{q}{2} <  j \leq \frac{q}{2} \right\}\ .
$$
We define~$\omega_{q}$ as~$\mathrm{exp}(2\pi i/ q)$. 
Recall that the discrete Fourier transform of every function~$f : \mathbb{Z}/q\mathbb{Z} \rightarrow \mathbb{C}$ is defined as follows:
$$
\forall x \in \mathbb{Z}/q\mathbb{Z}, \quad \widehat{f}(x) \eqdef \frac{1}{\sqrt{q}} \sum_{y \in \mathbb{Z}/q\mathbb{Z}} f(y) \cdot \omega_q^{-xy} \ .
$$

For an integer~$m$ and a real number~$r \geq 0$, we let~$\mathrm{B}_m(r)$ denote the ball of~$\mathbb{R}^m$ with radius~$r$. 
Vectors are in column notation and are written with bold letters (such as~$\vec{x}$). Uppercase bold letters are used to denote matrices (such as $\vec{A}$). 
For vectors~$\vec{a}_1, \ldots,\vec{a}_n$, we let $(\vec{a}_1 | \ldots |  \vec{a}_n)$ denote the matrix whose columns are the~$\vec{a}_i$'s. 
For any two vectors~$\vec{x},\vec{y} \in (\mathbb{Z}/q\mathbb{Z})^d$, we define their inner product as 
$$
\langle \vec{x}, \vec{y}\rangle \eqdef \ \sum\limits_{i=1}^{d} x_i y_i \bmod{q} \ .
$$

We define $\omega(\cdot), O(\cdot), \Theta(\cdot)$, and~$\Omega(\cdot)$  in the usual way. When it is not clear from the context, we use subscripts to clarify the input parameter, for instance~$\Omega_\lambda(\cdot)$.
We let $\poly(\lambda)$ denote any function which is of order~$O(\lambda^{a})$ for some constant~$a$. Furthermore, the notation~$\negl(\lambda)$ refers to a function that is~$O(1/\lambda^{b})$ for every constant~$b>0$.  A function of~$\lambda$ is called overwhelming if it is equal to~$1-\negl(\lambda)$. 

We use PPT to denote the usual class of classical Probabilistic Polynomial-Time algorithms.

Sometimes, we will use a subscript to stress the random variable specifying the associated probability space over which the probabilities or expectations are taken. For instance the probability~$\mathbb{P}_{X}(E)$ of the event~$E$ is taken over the probability space~$\Omega$ with respect to the induced measure by~$X$. 
We let~$\unif(S)$ denote the uniform distribution over~$S$.
Given any distribution $X$, the distribution $X^{\otimes m}$ is defined as $(X_{1}, \dots, X_{m})$ where $X_{i}$'s are independently distributed as $X$. 
For any two discrete probability distributions~$X$ and~$Y$ over a set~$S$, their statistical distance (also called the total variation distance) is defined as:
	\begin{equation*}
		\dstat(X,Y) \eqdef \frac{1}{2} \sum_{s \in S} \abs{\Pr_X(s)-\Pr_Y(s)} \ .
	\end{equation*}

	Let~$f:S\rightarrow \mathbb{C}$ be a function. We define the function~$f^{\otimes d}:S^d \rightarrow \mathbb{C}$ 
as $f^{\otimes d}(x_1,\dots,x_d) \eqdef f(x_1) \cdots f(x_d)$, for all~$(x_1,\ldots,x_d)\in S^d$.
When the dimension~$d$ is clear from the context, we abuse the notation and write~$f$ instead of~$f^{\otimes d}$.
Let~$S$ be a finite set and~$f:S \rightarrow \mathbb{C}$. We say that~$f$ is an \textit{amplitude function} if
$$
\sum\limits_{x \in S} |f(x)|^2 =1\ .
$$
We note that~$f^{\otimes d}$ is an amplitude function whenever~$f$ is an amplitude function.

\subsection{Quantum computations}
We refer the reader to~\cite{ChuangNielsen,watrous18} for introductions on quantum algorithms.

We use the quantum circuit model of computation. A quantum circuit operates on some number of qubits, using one-qubit or two-qubit unitary gates and projective measurements. The measurements are performed in a priori fixed computational basis. The outcome of the last measurement is typically considered as the output of the algorithm. 
An algorithm may use ancilla qubits, \ie{} extra quantum registers initialized to~$\ket{0}$.
 We say that a sequence of quantum circuits~$(Q_i)_i$ is QPT if there exists a deterministic
polynomial-time algorithm that takes~$i$ in unary as input and outputs the description of~$Q_i$ with gates and measurements.

		\medskip

{\bf \noindent Partial trace.} For our purposes, we need to describe sub-systems of a given ``composite'' quantum system. This description involves the {partial trace.} 
Let $\mathcal{A}$ and $\mathcal{B}$ be two Hilbert spaces with~$\{\ket{a}\}_{a \in \mathcal{I}}$ and~$\{\ket{b}\}_{b \in \mathcal{J}}$ as their orthonormal bases, respectively. For all~$a_1,a_2 \in \mathcal{I}$ and~$b_1,b_2 \in \mathcal{J}$, tracing out the register of~$\mathcal{B}$ is defined as follows:
$$
\tr_{\mathcal{B}} \left( \ketbra{a_1}{a_2} \otimes \ketbra{b_1}{b_2} \right) \eqdef 
\braket{b_1}{b_2} \ketbra{a_1}{a_2} \ .
$$ 
It is extended by linearity. 

\medskip

{\bf \noindent Trace distance.} We will also use the {\em trace distance} which is defined over two quantum states~$\rho,\sigma$ as follows:
\begin{align*}
\dtr\left(\rho,\sigma\right) \eqdef \frac{1}{2} \tr\left( \sqrt{(\rho-\sigma)^\dagger (\rho-\sigma)}\right)\ .
\end{align*}
For pure quantum states~$\ket{\psi}$ and~$\ket{\varphi}$, it can be simplified to~$\sqrt{1-|\braket{\varphi}{\psi}|^{2}}$.
The trace distance has the following properties (see~\cite[Th.~9.2]{ChuangNielsen}):
\begin{itemize}\setlength{\itemsep}{5pt}
\item for any joint states~$\rho,\sigma$ over~$\mathcal{A}\otimes \mathcal{B}$, it holds that~$\dtr(\tr_{\mathcal{B}}(\rho), \tr_{\mathcal{B}}(\sigma)) \leq \dtr(\rho, \sigma)$;
\item for any quantum states~$\rho,\sigma,\tau$, it holds that~$\dtr(\rho,\sigma) \leq \dtr(\rho, \tau) + \dtr(\tau, \sigma)$;
\item for any quantum states~$\rho,\sigma,\tau$, it holds that~$\dtr(\rho \otimes \tau, \sigma \otimes \tau) = \dtr(\rho,\sigma)$;
\item for any quantum algorithm~$\mathcal{Q}$ and any quantum states~$\rho, \sigma$, it holds that 
$\dtr(\mathcal{Q}(\rho), \mathcal{Q}(\sigma)) \leq \dtr(\rho, \sigma)$.
\end{itemize}

Let~$M$ be the set of possible outcomes of a measurement on the above states. Let~$X$ and~$Y$ be the distributions over~$M$ induced by measuring~$\rho$ and~$\sigma$, respectively. We have:
\begin{equation}\label{eq:stat-leq-tr}
\dstat(X,Y) \  \leq \  \dtr(\rho, \sigma) \ .
\end{equation}

{\bf \noindent Quantum Fourier transform (QFT).} 
The $\QFT$ over the additive group $\mathbb{Z}/q\mathbb{Z}$, whose characters are $\chi_{x} : y \mapsto \omega_{q}^{xy}$ for~$x \in \mathbb{Z}/q\mathbb{Z}$, is defined as follows:
$$
\forall x \in \mathbb{Z}/q\mathbb{Z}, \quad \QFT \ket{x} \eqdef \frac{1}{\sqrt{q}} \sum_{x \in \mathbb{Z}/q\mathbb{Z}} \omega_{q}^{xy} \ket{y} \ . 
$$ 
The quantum states~$\ket{\chi_{x}} \eqdef \QFT \ket{x}$ for~$x \in \mathbb{Z}/q\mathbb{Z}$ 
are called the {\em Fourier basis}, whereas the states~$\ket{x}$ form the {\em computational basis}. 
The following lemma recalls how the computational basis decomposes in the Fourier basis.
\begin{lemma}\label{lemma:qftm1} For any~$q \geq 2$, it holds that
	$$
	\forall y \in \mathbb{Z}/q\mathbb{Z}, \quad \ket{y} = \frac{1}{\sqrt{q}} \sum_{x \in \mathbb{Z}/q\mathbb{Z}} \omega_{q}^{-y x } \ket{\chi_{x}} \ .
	$$
\end{lemma}
\begin{proof}We have the following equalities:
	\begin{align*}
		\vec{QFT}\sum_{x \in \mathbb{Z}/q\mathbb{Z}} \omega_{q}^{-y x } \ket{\chi_{x}} &= \sum_{x \in \mathbb{Z}/q\mathbb{Z}} \omega_{q}^{-y x}\; \vec{QFT} \ket{\chi_{x}} \\
		&= \sum_{x \in \mathbb{Z}/q\mathbb{Z}} \omega_{q}^{-y x} \frac{1}{\sqrt{q}} \sum_{m\in \mathbb{Z}/q\mathbb{Z}} \omega_{q}^{x m} \; \vec{QFT}\ket{m} \\
		&=  \sum_{m \in \mathbb{Z}/q\mathbb{Z} } \left( \frac{1}{\sqrt{q}}\sum_{x \in \mathbb{Z}/q\mathbb{Z}} \omega_{q}^{x(m-y)} \right) \vec{QFT} \ket{m} \\
		&= \sqrt{q}\; \vec{QFT} \ket{y} \ .
	\end{align*}
	The result is obtained by applying $\vec{QFT}^{-1}$. 
\end{proof}

\subsection{Gaussian distributions}
The Gaussian function centered around $\vec{0}$ with the standard deviation parameter~$\sigma > 0$ is defined as:
$$ 
\forall \vec{x} \in \mathbb{R}^m: \ \rho_{\sigma}(\vec{x}) \eqdef \mathrm{e}^{-\pi \frac{\|\vec{x}\|^{2}}{\sigma^{2}}} \ .
$$
where $\| \cdot \|$ denotes the Euclidean norm of~$\vec{x}$. 
The following lemma shows the concentration behaviour of~$\rho_\sigma$ over~$\mathbb{Z}^m$.
\begin{lemma}[Adapted from~{\cite[Le.~1.5]{Ban93}}] \label{lemma:Ban-bound}
For any positive integer~$m$ and any real numbers~$\sigma>0$ and~$\sigma' \geq \sigma / \sqrt{2 \pi}$, it holds that
$$
\rho_\sigma \left( \mathbb{Z}^m \setminus \mathrm{B}_m(\sigma'\sqrt{m}) \right) \leq \left(\frac{\sigma'}{\sigma} \ \sqrt{2\pi \mathrm{e}} \ \mathrm{e}^{- \pi \frac{\sigma'^2}{\sigma^2}} \right)^m \rho_\sigma(\mathbb{Z}^m) \ .
$$
\end{lemma}

We have the following inequality.
\begin{lemma}\label{lemma:rho-Z-bound}
For every~$\sigma > 0$, we have~$\sigma \leq \rho_\sigma(\mathbb{Z}) \leq 1+\sigma$.
\end{lemma}
\begin{proof}
We have $\rho_{\sigma}(\mathbb{Z})\nonumber \leq 1+2\int_{0}^{+\infty}\rho_{\sigma}(x) \mathrm{d}x = 1 + \sigma$,
by comparing the sum and the integral. Moreover, using the Poisson summation formula, one obtains
 $\rho_\sigma(\mathbb{Z}) = \sigma  \cdot \rho_{1/\sigma}(\mathbb{Z}) \geq \sigma$.
\end{proof}

The discrete Gaussian distribution over~$\mathbb{Z}^m$ centered around~$\vec{0}$ with the standard deviation~$\sigma$ is defined as follows:
$$
\forall \vec{k} \in \mathbb{Z}^m: \ D_{\mathbb{Z}^m,\sigma}(\vec{k}) \eqdef \frac{\rho_\sigma(\vec{k})}{\rho_\sigma(\mathbb{Z}^m)} \ ,
$$
where~$\rho_\sigma(\mathbb{Z}^m) \eqdef \sum_{\vec{k} \in \mathbb{Z}^m} \rho_\sigma(\vec{k})$.
Folding~$D_{\mathbb{Z}^m,\sigma}$ modulo an integer~$q$ yields the distribution~$\vartheta_{\sigma,q}$.
\begin{definition}[Folded Discrete Gaussian Distribution]\label{def:vartheta}
Let~$q\geq 2$ an integer and~$\sigma > 0$ a real number. We define the folded discrete Gaussian distribution over~$(\mathbb{Z}/q\mathbb{Z})^m$ with standard deviation~$\sigma$ and folding parameter~$q$ by its probability mass function~$\vartheta_{\sigma,q}$:
$$\forall \vec{x} \in (\mathbb{Z}/q\mathbb{Z})^m: \ \vartheta_{\sigma,q} (\vec{x}) \eqdef \frac{\sum_{\vec{k} \in \mathbb{Z}^m} \rho_{\sigma} (\vec{x} + \vec{k}q)}{\rho_s(\mathbb{Z}^m)} \enspace .$$
\end{definition}

We note that in all distributions above, the dimension of the input is implicit and can be derived from the context.
The distribution~$\vartheta_{\sigma,q}$ behaves very closely to~$D_{\mathbb{Z}^m,\sigma}$.

\begin{lemma} \label{lemma:vartheta-rho}
Let~$m\geq 1$ and~$q\geq 2$ integers and~$ \sigma > 0$ a real number. Assume that~$q \geq 2\sigma\sqrt{m}$. Then for every~$\vec{x} \in \mathbb{Z}^m \cap (-q/2,q/2]^m$, it holds that
$$
D_{\mathbb{Z}^m,\sigma}(\vec{x}) \leq \vartheta_{\sigma,q}(\vec{x}) \leq D_{\mathbb{Z}^{m},\sigma}(\vec{x}) + \mathrm{e}^{-\frac{q^2}{(2\sigma)^2}} \ \text{ and } \  \sqrt{D_{\mathbb{Z}^m,\sigma}(\vec{x})} \leq  \sqrt{\vartheta_{\sigma,q}(\vec{x})} \leq  \sqrt{D_{\mathbb{Z}^m,\sigma}(\vec{x})} + \mathrm{e}^{-\frac{q^2}{8\sigma^2}} \ .
$$
\end{lemma}
\begin{proof}
For $\vec{x} \in \mathbb{Z}^m \cap (-q/2,q/2]^m$, we have
\begin{align*}
\sum_{\vec{k} \in \mathbb{Z}^m} \rho_{\sigma} (\vec{x} + \vec{k}q) &= \rho_{\sigma}(\vec{x}) + \sum_{\vec{k} \in \mathbb{Z}^m \setminus \{\vec{0}\}}  \rho_{\sigma} (\vec{x} + \vec{k}q)\\
&\leq \rho_{\sigma}(\vec{x}) + \sum_{\vec{k} \in \mathbb{Z}^m: \|\vec{k}\| \geq  \frac{q}{2}}  \rho_{\sigma} (\vec{k})\\
&\leq \rho_{\sigma}(\vec{x}) + \rho_{\sigma}\left(\mathbb{Z}^m \setminus \mathrm{B}_m\left(\frac{q}{2}\right)\right)\\
&\leq \rho_{\sigma}(\vec{x}) + \left(\frac{q}{2\sigma \sqrt{m}} \ \sqrt{2\pi \mathrm{e}} \ \mathrm{e}^{- \pi \frac{q^2}{m(2\sigma)^2}} \right)^m \rho_{\sigma}(\mathbb{Z}^m) \quad (\text{by Lemma~\ref{lemma:Ban-bound} with }\sigma' = \frac{q}{2\sqrt{m}})\\
&\leq \rho_{\sigma}(\vec{x}) + \left( \mathrm{e}^{-\frac{q^2}{m(2\sigma)^2}} \right)^m \rho_{\sigma}(\mathbb{Z}^m) \ ,
\end{align*}
where the last inequality follows since for every~$x \geq 1$, we have~$(1-\pi)x^2 \geq \ln x + \ln \sqrt{2 \pi \mathrm{e}}$. This gives the first statement. For the other one, we take the square-root of the last inequality above to obtain:
\[
\sqrt{\vartheta_{\sigma,q} (\vec{x})} \leq \sqrt{\frac{\rho_{\sigma}(\vec{x}) + \mathrm{e}^{-\frac{q^2}{(2\sigma)^2}}\rho_{\sigma}(\mathbb{Z}^m)}{\rho_\sigma(\mathbb{Z}^m)}} \leq \sqrt{\frac{\rho_{\sigma}(\vec{x})}{\rho_s(\mathbb{Z}^m)}} + \mathrm{e}^{-\frac{q^2}{8\sigma^2}} \ .
\]
This completes the proof.
\end{proof}

\medskip
\subsection{Learning With Errors} \label{subsec:prelim-lwe}

The $\lwe$ problem was introduced by~\cite{Regev09}. It can be viewed as a distributional variant of the bounded distance decoding problem over Euclidean lattices  (see, e.g., the discussions in~\cite{GPV08,SSTX09}).

\begin{definition}[$\lwe$]
\label{def:lwe} Let~$m\geq n \geq 1$ and~$q\geq 2$ be integers, and~$\chi$ be a distribution over~$\mathbb{Z}/q\mathbb{Z}$. 
The parameters~$m,n,q$ and~$\chi$ are functions of some security parameter~$\lambda$. 
Let~$\vec{A} \in (\mathbb{Z}/q\mathbb{Z})^{m \times n}$, $\vec{s} \in (\mathbb{Z}/q\mathbb{Z})^{n}$ be sampled uniformly and~$\vec{e} \in (\mathbb{Z}/q\mathbb{Z})^{m}$ be sampled from~$\chi^{\otimes m}$.  The search~$\lwe_{m,n,q,\chi}$ problem is to find~$\vec{s}$ and~$\vec{e}$ given the pair~$(\vec{A},\vec{A}\vec{s}+\vec{e})$. The vectors~$\vec{s}$ and~$\vec{e}$ are respectively called the secret and the noise.

Whenever~$\chi$ is equal to the folded discrete Gaussian distribution~$\vartheta_{\sigma,q}$  for some~$\sigma>0$, we overwrite the notation as~$\lwe_{m,n,q,\sigma}$.
\end{definition}

For sufficiently small values of~$\sigma$, for example~$\sigma = O( q^{(m-n)/m}/\sqrt{\lambda})$, one can show that the valid~$\lwe$ instances are sparse in~$(\mathbb{Z}/q\mathbb{Z})^m$: a uniformly sampled vector~$\vec{b}$ is unlikely a valid~$\lwe$ instance.

The quantum hardness of the $\lwe$ problem for various distributions of the noise and the secret has been extensively studied \color{blue} (see, e.g.,~\cite{Regev09,Peikert09,GKPV10,MM11,BLPRS13,BD20}) \color{black} and it is known that $\lwe$ is no easier than some conjecturally hard worst-case problems~\cite{Regev09}.

\section{Witness Obliviousness} \label{sec:aw-ob}

In this section, we are interested in the task of sampling an~$\lwe$ instance~$(\vec{A},\vec{b})$, given a matrix~$\vec{A}$. A direct approach 
(which follows the definition of the $\lwe$ problem) is, using a source of randomness, to produce a secret vector~$\vec{s}$ and a noise vector~$\vec{e}$ with appropriate distributions, and then to output~$\vec{b} = \vec{As} + \vec{e}$. This sampler has a 
particular property:  it itself knows the secret~$\vec{s}$. In a sense, the~$\lwe$ problem with the vector~$\vec{b}$ is not hard for the sampler. In that case, we say that an~$\lwe$ sampler is \textit{witness-aware}. We are interested in samplers that are not witness-aware, i.e., that are \textit{witness-oblivious}.

Below, we discuss instance samplers and knowledge assumptions with a focus on the~$\lwe$ problem. We start by splitting our discussion about obliviousness between the classical and quantum settings in Subsections \ref{subsec:class} and \ref{subsec:quant}.
Furthermore, we show in Subsection \ref{subsec:OSred} how to deduce from a given oblivious sampler another one via reductions. 
Finally, we show in Subsection \ref{subsec:OS} how to design a quantum oblivious sampler for $\lwe$.

\subsection{Classical Setting}\label{subsec:class}

We begin by the definition of a classical~$\lwe$ sampler. 
\begin{definition}[Classical~$\lwe$ Samplers] \label{def:lwe-sampler}
Let~$m \geq n \geq 1$ and~$q\geq 2$ be integers, and~$\chi$ be a distribution over~$\mathbb{Z}/q\mathbb{Z}$. The parameters~$m,n,q,\chi$ are functions of some security parameter~$\lambda$. 
Let~$\sampler$ be a PPT algorithm that has the following specification:
\begin{itemize}
\item[] $\sampler(1^{\lambda},\vec{A};r)$: Given as input the security parameter $1^{\lambda}$ (in unary), the matrix~$\vec{A} \in (\mathbb{Z}/q\mathbb{Z})^{m \times n}$ and an auxiliary bit string~$r$ of size $\poly(\lambda)$, it returns a pair~$(\vec{A},\vec{b})$ with~$\vec{b} \in \left( \mathbb{Z}/q\mathbb{Z}\right)^{m}$.
\end{itemize}
We say that~$\sampler$ is a classical~$\lwe_{m,n,q,\chi}$ sampler if, for a uniformly distributed input matrix~$\vec{A}$ and a statistically independent random string~$r$, 
the distribution of~$\sampler(1^{\lambda},\vec{A};r)$ 
is within statistical distance~$\negl(\lambda)$ from the distribution of~$\lwe_{m,n,q,\chi}$ as given in Definition~\ref{def:lwe}.
\end{definition}

As discussed above, some samplers, during their course of execution, might need to produce the witness in order to be successful, namely they are aware of the witness. Assume that we are given the concrete machine that implements the sampler. If we carefully inspect all steps of the machine, the witness must show up at some point (in an easily recoverable way), which allows us to extract it. We grasp this intuition in the following definition.

\begin{definition}[Witness-Oblivious $\lwe$ Samplers] \label{def:r-obl}
Let~$m,n,q,\chi,\lambda$ as above.
We say that a classical~$\lwe_{m,n,q,\chi}$ sampler~$\sampler$ is witness-oblivious if for every PPT extractor~$\extractor$, we have
\begin{align*}
	\Pr\left(\vec{s}' = \vec{s} \mbox{ and } \vec{e}' = \vec{e} \ \Bigg| \ 
	\begin{array}{l}
		\vec{A} \leftarrow\unif\left((\mathbb{Z}/q\mathbb{Z})^{m \times n}\right)\\
		r \leftarrow\unif\left(\{0,1\}^{\poly(\lambda)}\right)\\
		(\vec{A},\vec{b}\eqdef \vec{A}\vec{s} + \vec{e})  \leftarrow \sampler(1^{\lambda},\vec{A};r)\\
		(\vec{s}',\vec{e}') \leftarrow \extractor(1^{\lambda},\vec{A},\vec{b},r)\\
	\end{array} 
	\right)
	\leq \negl(\lambda) \ ,
\end{align*}
where the probability is also taken over the randomness of~$\extractor$.
\end{definition}

This definition implies that given~$(\vec{A},\vec{b})$, finding a witness is hard for all PPT algorithms.

\begin{lemma} \label{prop:hard-sample-obliviousness}
Let~$m,n,q,\chi,\lambda$ as above.
Suppose that there exists a classical witness-oblivious $\lwe_{m,n,q,\chi}$ sampler.
Then the~$\lwe_{m,n,q,\chi}$ problem is hard for every PPT algorithm; for all PPT algorithm $\mathcal{B}$, we have 
	  	$$
	  	 \mathbb{P}\left(\vec{s}' = \vec{s} \mbox{ and } \vec{e}' = \vec{e}  \ \Bigg| \ 
	  	 \begin{array}{l}
	  	 	(\vec{A},\vec{b}\eqdef \vec{A}\vec{s}+\vec{e}) \leftarrow \lwe_{m,n,q,\chi}\\
	  	  (\vec{s}',\vec{e}') \leftarrow \mathcal{B}(1^{\lambda},\vec{A},\vec{b})
	  	  \end{array} \right) \leq \frac{1}{\negl(\lambda)} \ ,
	  	$$ 
	where the probability is also taken over the randomness of~$\mathcal{B}$. 
\end{lemma}

\begin{proof}
Let~$\sampler$ denote the witness-oblivious sampler and~$\mathcal{B}$ be an arbitrary PPT algorithm.
By assumption, if given as input a uniformly distributed matrix~$\vec{A}$,  the output distribution of algorithm~$\sampler$ is
within statistical distance~$\negl(\lambda)$  from the instance distribution of~$\lwe_{m,n,q,\chi}$. Therefore, by 
properties of the statistical distance,
we have:
\begin{multline*} 
\mathbb{P}\left(\vec{s}' = \vec{s} \mbox{ and } \vec{e}' = \vec{e}  \ \Bigg| \ 
\begin{array}{l}
	(\vec{A},\vec{b}\eqdef \vec{A}\vec{s}+\vec{e}) \leftarrow \lwe_{m,n,q,\chi} \\
	(\vec{s}',\vec{e}') \leftarrow \mathcal{B}(1^{\lambda},\vec{A},\vec{b})
\end{array} \right) \leq \\  
\Pr\left(\vec{s}' = \vec{s} \mbox{ and } \vec{e}' = \vec{e}\ \Bigg| \ 
\begin{array}{l}
\vec{A} \leftarrow\unif\left((\mathbb{Z}/q\mathbb{Z})^{m \times n}\right) \\
	r \leftarrow\unif\left(\{0,1\}^{\poly(\lambda)}\right)\\
	(\vec{A},\vec{b}\eqdef \vec{A}\vec{s}+\vec{e})  \leftarrow \sampler(1^{\lambda},\vec{A};r)\\
	(\vec{s}',\vec{e}') \leftarrow \mathcal{B}(1^{\lambda},\vec{A},\vec{b})\\
\end{array} 
\right)
+ \negl(\lambda) \ .
\end{multline*}
Define the following PPT algorithm $\mathcal{E}$:
$$
\extractor(1^{\lambda},\vec{A},\vec{b},r) \eqdef \mathcal{B}(1^{\lambda},\vec{A},\vec{b}) \ .
$$
 Therefore, as $\mathcal{S}$ is a classical witness-oblivious sampler, we have 
 $$
 \mathbb{P}\left(\vec{s}' = \vec{s} \mbox{ and } \vec{e}' = \vec{e}  \ \Bigg| \ 
\begin{array}{l}
\vec{A} \leftarrow\unif\left((\mathbb{Z}/q\mathbb{Z})^{m \times n}\right) \\
	r \leftarrow U(\{0,1\}^{\poly(\lambda)})\\
	(\vec{A},\vec{b}\eqdef \vec{A}\vec{s}+\vec{e})  \leftarrow \sampler(1^{\lambda},\vec{A};r)\\
	(\vec{s}',\vec{e}') \leftarrow \extractor(1^{\lambda},\vec{A},\vec{b},r)
\end{array} \right)
 \leq \negl(\lambda) \ ,
 $$
 which completes the proof. 
\end{proof}

Lemma~\ref{prop:hard-sample-obliviousness} states that the existence of a witness-oblivious~$\lwe$ sampler 
implies the hardness of the~$\lwe$ problem. 
We are interested in the converse, i.e., in obtaining an oblivious sampler, under the assumption that~$\lwe$ is hard.

\subsection{Quantum Setting}\label{subsec:quant}

To discuss the post-quantum security of cryptographic schemes, we must migrate to quantum algorithms with appropriate extension of obliviousness. 
Before going into the details, we need an appropriate definition of quantum samplers.

\begin{definition}[Quantum~$\lwe$ Samplers] \label{def:lwe-qsampler}
Let~$m \geq n \geq 1$ and~$q\geq 2$ be integers, and~$\chi$ be a distribution over~$\mathbb{Z}/q\mathbb{Z}$. The parameters~$m,n,q,\chi$ are functions of some security parameter~$\lambda$. Let~$\qsampler$ be a QPT algorithm that has the following specification:
\begin{itemize}
\item[] $\qsampler\left(1^{\lambda}, {\vec{A}},\ket{0}^{\poly(\lambda)}\right)$: Given as input the security parameter~$1^{\lambda}$ (in unary), the matrix~$\vec{A}\in (\mathbb{Z}/q\mathbb{Z})^{m \times n}$, and a polynomial number of ancillas each initialized to~$\ket{0}$ as inputs, it returns a pair~$(\vec{A},\vec{b})$ with~$\vec{b} \in \left( \mathbb{Z}/q\mathbb{Z}\right)^{m}$.
\end{itemize}

We say that~$\qsampler$ is a quantum~$\lwe_{m,n,q,\chi}$ sampler if, for a uniformly distributed input matrix~$\vec{A}$, 
the distribution of~$\sampler(1^{\lambda},{\vec{A}},\ket{0}^{\poly(\lambda)})$ 
is within statistical distance~$\negl(\lambda)$ from the distribution of~$\lwe_{m,n,q,\chi}$ as given in Definition~\ref{def:lwe}.
\end{definition}

The main principle we use to base our obliviousness definition on is that observing or measuring the execution of a machine 
(be it classical or quantum) must not change the view that the sampler has of itself. Assume that an extractor is observing
a sampler. Let~$\rho_{\mathsf{Q} \otimes \mathsf{E}}$ represent the joint state of the sampler~$\qsampler$ and the extractor~$\extractor$ at some step. The extractor might have carried out particular inspections that ended up in entangling its register with that of the sampler, so the state~$\rho_{\mathsf{Q} \otimes \mathsf{E}}$ might not be separable. We intuitively expect from a valid extractor that if we trace out its register, the remaining state must be as if no extractor was inspecting the sampler at all. Namely, if~$\rho_\mathsf{Q}$ was the state of an isolated sampler at the same step, we require that~$\tr_\mathsf{E}(\rho_{\mathsf{Q} \otimes \mathsf{E}}) = \rho_\mathsf{Q}$.  
We define valid extractors as follows, based on the above discussion, except that we only require that~$\tr_\mathsf{E}(\rho_{\mathsf{Q} \otimes \mathsf{E}})$ and~$\rho_\mathsf{Q}$ are close for the trace distance.

\begin{definition} \label{def:valid-extraction}
Let~$\mathsf{Q}$ and~$\mathsf{E}$ be two quantum registers initialized to~$\tau_\mathsf{Q}$ and~$\tau_\mathsf{E}$ where~$\tau_\mathsf{Q}$ consists of classical information and ancillas while~$\tau_\mathsf{E}$ consists of only ancillas. Let~$\mathcal{G}$ be the set of one-qubit and two-qubit unitary gates. Let~$\qsampler$ be a quantum algorithm operating on register~$\mathsf{Q}$ with gates~$\qsampler_1,\dots,\qsampler_\ell$ each of which either belongs to~$\mathcal{G}$ or is a measurement in the computational basis. Let~$\extractor$ be a quantum algorithm operating on the joint register~$\mathsf{Q} \otimes \mathsf{E}$ with the gates~$(\extractor_{0,j})_{j \leq k(0)},(\extractor_{1,j})_{j \leq k(1)},\dots,(\extractor_{\ell+1,j})_{j \leq k(\ell+1)}$ each of which either belongs to~$\mathcal{G}$ or is a measurement in the computational basis. 

In the first scenario, suppose that~$\qsampler$ is operating alone on~$\mathsf{Q}$. Let~$\sigma_\mathsf{Q}(i)$ be the density matrix representing the state of~$\mathsf{Q}$ just after the~$i$-th step of~$\qsampler$ for~$i\geq 1$ and just before the first step of~$\qsampler$ for~$i=0$, as depicted in Figure~\ref{fig:sampler}.

In the second scenario, suppose that~$\qsampler$ and~$\extractor$ are operating jointly on the registers~$\mathsf{Q}$ and~$\mathsf{E}$  as follows. After the~$i$-th step of~$\qsampler$ for~$i \geq 1$ and before the first step of~$\qsampler$ for~$i=0$, algorithm~$\extractor$ is given both registers to perform its operations~$(\extractor_{i,j})_{j \leq k(i)}$ and sends register~$\mathsf{Q}$ to~$\qsampler$. For every~$1 \leq j \leq k(i)$, let~$\rho_{\mathsf{Q} \otimes \mathsf{E}}(i,j)$ denote the joint state of the registers after applying~$\extractor_{i,j}$, and let~$\rho_{\mathsf{Q} \otimes \mathsf{E}}(i,0)$ denote the state just before applying~$\extractor_i$, as depicted in Figure~\ref{fig:interaction}.

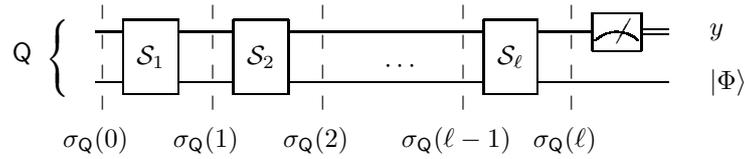
\begin{figure}
\[
\Qcircuit @C=1em @R=.7em {
 \barrier[0.25em]{1} & & \multigate{1} {\qsampler_1}{2} \barrier[0.25em]{1} & \qw & \multigate{1}{\qsampler_2}{2}  \barrier[0.25em]{1} & \qw & \qw & \cds{2}{\cdots}  \barrier[0.25em]{1} & \qw & \multigate{1}{\qsampler_{\ell}}{2} \barrier[0.25em]{1} & \qw & \meter & \cw & \rstick{y}\\
& & \ghost{\qsampler_1} & \qw & \ghost{\qsampler_2} & \qw & \qw & \qw & \qw & \ghost{\qsampler_{\ell}} & \qw & \qw & \qw & \rstick{\ket{\Phi}}\\
 &\dstick{\sigma_{\mathsf{Q}}(0)} & & \dstick{\sigma_{\mathsf{Q}}(1)}  & & \dstick{\sigma_{\mathsf{Q}}(2)} & & & \dstick{\sigma_{\mathsf{Q}}(\ell-1)} & & \dstick{\sigma_{\mathsf{Q}}(\ell)} \inputgroupv{1}{2}{.8em}{.8em}{\mathsf{Q}} \\
}
\]
\caption{The execution of the sampler.}
\label{fig:sampler}
\end{figure}

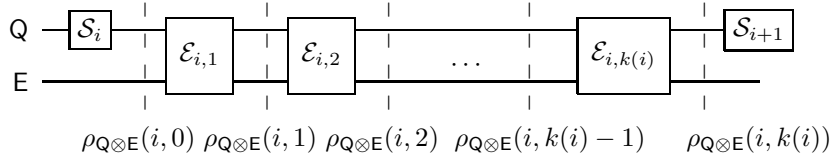
\begin{figure}
\[
\Qcircuit @C=1em @R=.7em {
& \lstick{\mathsf{Q}} & \gate{\qsampler_i} \barrier[0.25em]{1} & \qw & \multigate{1}{\extractor_{i,1}}{2} \barrier[0.25em]{1} & \qw & \multigate{1}{\extractor_{i,2}}{2}  \barrier[0.25em]{1} & \qw & \qw & \cds{2}{\cdots}  \barrier[0.25em]{1} & \qw & \qw & \multigate{1}{\extractor_{i,k(i)}}{2} \barrier[0.25em]{1} & \qw & \gate{\qsampler_{i+1}} &\\
& \lstick{\mathsf{E}} & \qw & \qw & \ghost{\extractor_{i,1}} & \qw & \ghost{\extractor_{i,2}} & \qw & \qw & \qw & \qw & \qw & \ghost{\extractor_{i,k(i)}} & \qw & \qw & \\
& & & \dstick{\rho_{\mathsf{Q}\otimes \mathsf{E}}(i,0)} & & \dstick{\rho_{\mathsf{Q}\otimes \mathsf{E}}(i,1)} & & \dstick{\rho_{\mathsf{Q}\otimes \mathsf{E}}(i,2)} & & & & \dstick{\rho_{\mathsf{Q}\otimes \mathsf{E}}(i,k(i)-1)} & & & \dstick{\rho_{\mathsf{Q}\otimes \mathsf{E}}(i,k(i))}\\
}
\]
\caption{The execution of the extractor between Steps~$i$ and~$i+1$ of the sampler.}
\label{fig:interaction}
\end{figure}

Let~$\varepsilon \geq 0$ be a real number. We say that~$\extractor$ is an $\varepsilon$-valid extractor if for every~$0 \leq i \leq \ell$ and every~$0 \leq j \leq k(i)$, it holds that:
\begin{equation}\label{eq:valid-ext}
\dtr\left( \tr_\mathsf{E}(\rho_{\mathsf{Q} \otimes \mathsf{E}}(i,j)) , \sigma_\mathsf{Q}(i) \right) \leq \varepsilon \ .
\end{equation}

We say that an extractor is perfect if it is $\varepsilon$-valid for~$\varepsilon = 0$. Furthermore, we let~$\langle \qsampler, \extractor \rangle (\tau_\mathsf{Q}, \tau_\mathsf{E})$ denote the joint output.
\end{definition}
We note that this definition does not assume that~$\qsampler$ is a sampler, nor that~$\qsampler$ and~$\extractor$ are efficient.

This definition covers all valid extractors in the classical setting. A classical sampler only exploits classical registers. Observing and copying the internal states and the randomness encoded in classical registers is perfectly doable. This translates to the extractor having all the information of internal states and randomness of the sampler. This gives exactly the same information to the extractor as  in Definition~\ref{def:r-obl}.

\begin{definition}[Witness-Oblivious Quantum Samplers] \label{def:qr-obl}
Let~$m,n,q,\chi,\lambda$ as in Definition~\ref{def:lwe-qsampler}. 
We say that a quantum~$\lwe_{m,n,q,\chi}$ sampler~$\qsampler$ is witness-oblivious if for every~$\negl(\lambda)$-valid 
QPT extractor~$\extractor$, we have
\begin{multline*}
\Pr\left(\vec{s}' = \vec{s} \mbox{ and } \vec{e}' = \vec{e} \ \Bigg| \ 
\begin{array}{l}
\vec{A} \leftarrow\unif\left((\mathbb{Z}/q\mathbb{Z})^{m \times n}\right) \\
\left((\vec{A},\vec{b}\eqdef\vec{A}\vec{s}+\vec{e}),(\vec{s}',\vec{e}') \right) \leftarrow \langle \qsampler, \extractor \rangle \left((1^{\lambda}, \vec{A},\ket{0^{\poly(\lambda)}}), \ket{0^{\poly(\lambda)}}\right)
\end{array} 
\right) \\
\leq \negl(\lambda) \ ,
\end{multline*}
where the probability is also taken over  
the measurements of~$\qsampler$ and~$\extractor$. 
\end{definition}

We note that a statement similar to Lemma~\ref{prop:hard-sample-obliviousness} holds for quantum witness-oblivious samplers.

\medskip
{\bf \noindent Relation to the definition from~\cite{LMZ23}.}
The authors of~\cite{LMZ23} adopted a different approach to define valid extractors. Their definition only deals with unitary algorithms followed by a single final measurement. The sampler is first executed until it performs its final measurement, and then the remaining working register and the measurement outcome are handed over to the extractor. The extractor is not allowed to inspect or observe the sampler during its execution. Using the notations of Figure~\ref{fig:sampler}, a valid extractor, in their case, is only given the classical output~$y$, and the quantum output~$\ket{\Phi}$. 

\begin{definition}[Adapted from{\cite{LMZ23}}]\label{def:lmz-extractor}
Let~$\qsampler$ be a unitary algorithm with a pure initial state and some classical string as input. 
Assume that~$\qsampler$ performs a final measurement in the computational basis over part of its register. 
A quantum algorithm~$\extractor$ is said to be an LMZ extractor for~$\qsampler$ if it operates as follows: in the first phase, the sampler runs its circuit and outputs~$\ket{y,\Phi}$ where~$y$ is classical and~$\ket{\Phi}$ is quantum; in the second phase, the circuit of the extractor runs over~$\ket{y,\Phi}$, and possibly extra ancillas, and outputs a classical string. 
\end{definition}
 
We stress that in the definition above, the extractor is not allowed to engage in the first phase: it proceeds \emph{after} the sampler. 
We show in the following lemma that, when restricted to unitary algorithms with a pure initial state,
our definition of perfect extractor is equivalent to the one from~\cite{LMZ23}.

\begin{lemma} \label{lemma:locc}
Let~$\mathsf{Q}$ be a quantum register initialized with some classical string $z$ and with the pure state~$\ket{0}_{\mathsf{Q}}$. Let~$\mathsf{E}$ be a quantum register with pure initial states~$\ket{0}_{\mathsf{E}}$.
Let~$\qsampler$ be a quantum algorithm operating on~$\mathsf{Q}$ by a series of unitary gates followed by a single measurement.
Let~$\extractor$ be a QPT $\varepsilon$-valid extractor operating on two registers~$\mathsf{Q}$ and~$\mathsf{E}$ as per Definition~\ref{def:valid-extraction}. 
Then, there exists an LMZ extractor~$\extractor'$ (as per Definition~\ref{def:lmz-extractor}) that is QPT, and
$$
\dtr\Big( \langle \qsampler, \extractor \rangle \left((z,\ket{0}_{\mathsf{Q}}),\ket{0}_{\mathsf{E}}\right) \ , \ \extractor' \left( \qsampler \left(z,\ket{0}_{\mathsf{Q}} \right) \right) \Big) \leq 2\sqrt{2\varepsilon}\ .
$$
\end{lemma}

\begin{remark}
	In the case of quantum $\lwe$ samplers, the classical string $z$ will be $1^{\lambda}$ and $\vec{A}$. 
\end{remark}

\begin{lemma}\label{lemma:sub-locc}
Using notations of Lemma~\ref{lemma:locc}, for any $(i,j)$, it holds that
	$$
	\dtr\left( \rho_{\mathsf{Q}\otimes \mathsf{E}}(i,j), \sigma_{\mathsf{Q}}(i) \otimes \mathcal{E}'_{(i,j)}\left( \ket{0} \right) \right) \leq 2\sqrt{2\varepsilon}\ ,
	$$
	where $\mathcal{E}'_{(i,j)}$ is a QPT algorithm given as input a quantum state $\ket{0}$, and implicitly the description of $\sampler$ and $\extractor$ as well as the classical string $z$.
\end{lemma}
\begin{proof}

	Let us first prove the statement for the perfect-case, \ie{} $\varepsilon = 0$. We prove it by induction on $i$. Suppose that the statement holds for $(i-1)\geq 0$ (it clearly holds for $i = 0$ as at this step neither $\sampler$ nor $\extractor$ performs any computations). In particular, we have
	$$
	\rho_{\mathsf{Q} \otimes \mathsf{E}}(i-1,k(i-1)) = \sigma_{\mathsf{Q}}(i-1) \otimes \mathcal{E}'_{(i-1,k(i-1))}(\ket{0}) \ .
	$$ 
	The statement holds for~$(i,0)$ by definition, namely, 
	\begin{align*} 
	\rho_{\mathsf{Q} \otimes \mathsf{E}}(i,0)&= \mathcal{S}_{i}\sigma_{\mathsf{Q}}(i-1) \otimes  \mathcal{E}'_{(i-1,k(i-1))}(\ket{0}) \\
	&= \sigma_{\mathsf{Q}}(i) \otimes  \mathcal{E}'_{(i-1,k(i-1))}(\ket{0}) \enspace.
	\end{align*} 

	Note that $\mathcal{E}$ is a perfect extractor, therefore according to Definition \ref{def:valid-extraction}, we have
	$$
	\tr_\mathsf{E}(\rho_{\mathsf{Q} \otimes \mathsf{E}}(i,j)) = \sigma_{\mathsf{Q}}(i) \ .
	$$
	The state~$\sigma_\mathsf{Q}(i)$ is pure since~$\qsampler$ only applies unitaries to~$\mathsf{Q}$ which initially contains the pure quantum state $\ket{0}_{\mathsf{Q}}$. Therefore, according to the above equality, ~$\rho_{\mathsf{Q} \otimes \mathsf{E}}(i,j)$ 
	is a product state. 
	Furthermore, it is necessarily given by 
	$$
	\sigma_{\mathsf{Q}}(i)\otimes \rho_{\mathsf{E}} \ ,
	$$
	where $\rho_{\mathsf{E}}$ is the quantum state obtained after applying $\mathcal{E}_{i,1},\dots, \mathcal{E}_{i,j}$ as in Figure \ref{fig:interaction} on 
	$$
	\rho_{\mathsf{Q} \otimes \mathsf{E}}(i,0) = \sigma_{\mathsf{Q}}(i) \otimes  \mathcal{E}'_{(i-1,k(i-1))}(\ket{0})\ ,
	$$
	and then tracing out~$\mathsf{Q}$. 
	On one hand, $\sigma_{\mathsf{Q}}(i)$ can be computed via $i$ steps of $\sampler$ given the classical string $z$ and the quantum state $\ket{0}_{\mathsf{Q}}$ where no measurement are performed. Therefore, given the description of $\sampler$, we can compute a polynomial time unitary $\vec{U}$ (the sampler $\sampler$ is QPT) such that $\sigma_{\mathsf{Q}}(i) = \vec{U}\ket{0}$.  On the other hand, the quantum state~$\rho_{\mathsf{E}\otimes\mathsf{Q}}(i,j)$ is equal, by definition,~to 
	\begin{align*}
		\rho_{\mathsf{E}\otimes\mathsf{Q}}(i,j)&=\mathcal{E}_{i,j}\dots \mathcal{E}_{i,1} \left( \rho_{\mathsf{Q} \otimes \mathsf{E}}(i,0) \right)\\
		&= 	\mathcal{E}_{i,j}\dots \mathcal{E}_{i,1} \left( \sigma_{\mathsf{Q}}(i) \otimes  \mathcal{E}'_{(i-1,k(i-1))}(\ket{0}) \right)  \\
		&= 	\mathcal{E}_{i,j}\dots \mathcal{E}_{i,1} \left( \vec{U}\ket{0}_{\mathsf{Q}} \otimes  \mathcal{E}'_{(i-1,k(i-1))}(\ket{0}) \right) \ .
	\end{align*}
	The algorithm $\mathcal{E}'_{(i,j)}$ computes the above state and then it traces out its first register $\mathsf{Q}$, namely it keeps only the second register $\mathsf{E}$. It shows that the lemma holds when $\varepsilon = 0$ for any $i$ and any~$0 \leq j \leq k(i)$. It concludes the proof by induction in this case.

	Now suppose that $\varepsilon > 0$. Let us consider the same algorithm $\mathcal{E}'_{(i,j)}$ as above. However, in this case (when keeping the second register) it is not true anymore that,
	$
	\rho_{\mathsf{Q}\otimes \mathsf{E}}(i,j) = \sigma_{\mathsf{Q}}(i) \otimes \mathcal{E}'_{(i,j)}\left( \ket{0} \right)
	$. Indeed, $\tr_\mathsf{E}(\rho_{\mathsf{Q} \otimes \mathsf{E}}(i,j))$ is not a pure state, therefore $	\rho_{\mathsf{Q}\otimes \mathsf{E}}(i,j) $ is not a product state. To handle this case, let us introduce the fidelity $F(\cdot,\cdot)$ between quantum states, \ie{} for all quantum states~$\rho$ and~$\sigma$
	$$
	F(\rho,\sigma) \eqdef \tr\sqrt{\sqrt{\rho}\sigma\sqrt{\rho}} \ .
	$$
	By Uhlmann's theorem \cite[Th. 9.14, Exercise 9.15]{ChuangNielsen}, there exists some purifications $\ket{\varphi}$ and $\ket{\psi}$ of $\tr_\mathsf{E}(\rho_{\mathsf{Q} \otimes \mathsf{E}}(i,j))$ and $\sigma_{\mathsf{Q}}(i)$, respectively, such that 
	\begin{equation}\label{eq:fidelity} 
	F(\tr_\mathsf{E}(\rho_{\mathsf{Q} \otimes \mathsf{E}}(i,j)),\sigma_{\mathsf{Q}}(i)) =  \left|\braket{\varphi}{\psi} \right| \ .
	\end{equation} 
	Note that by definition: $\tr_{\mathsf{E}}(\ketbra{\psi}) = \sigma_{\mathsf{Q}}(i)$ which is a pure state. Therefore $\ket{\psi}$ is a product state, in particular 
	$$
	\ket{\psi} = \sigma_{\mathsf{Q}}(i) \otimes \rho  \enspace.
	$$ 
	By Fuchs-van de Graaf inequalities, \cite[Eq. 9.110]{ChuangNielsen}, it holds that
	\begin{align*} 
		F(\tr_\mathsf{E}(\rho_{\mathsf{Q} \otimes \mathsf{E}}(i,j)),\sigma_{\mathsf{Q}}(i)) &\geq 1 - \dtr(\tr_\mathsf{E}(\rho_{\mathsf{Q} \otimes \mathsf{E}}(i,j)),\sigma_{\mathsf{Q}}(i)) \\
		&\geq 1-\varepsilon \enspace.
	\end{align*} 
	Therefore, by using Equation \eqref{eq:fidelity}, we have
	$$
	\left|\braket{\varphi}{\psi} \right| \geq 1 -\varepsilon
	$$
	which implies that
	$$
	\dtr\left( \ket{\varphi},\ket{\psi} \right)\leq \sqrt{1-(1-\varepsilon)^{2}} \leq \sqrt{2\varepsilon} \enspace.
	$$
	The above inequality holds for any purification $\ket{\varphi}$ of $\tr_{\mathsf{E}}\left( \rho_{\mathsf{Q} \otimes \mathsf{E}}(i,j) \right)$, in particular for $\rho_{\mathsf{Q} \otimes \mathsf{E}}(i,j)$ (without loss of generality we can suppose that it is a pure quantum after some purification), namely, 
	\begin{equation}\label{eq:dtrF} 
	\dtr\left( \rho_{\mathsf{Q} \otimes \mathsf{E}}(i,j),\ket{\psi} \right)\leq \sqrt{1-(1-\varepsilon)^{2}} \leq \sqrt{2\varepsilon} \ .
	\end{equation} 
	Recall now that $\ket{\psi} = \sigma_{\mathsf{Q}}(i) \otimes \rho$. In its last step, algorithm $\mathcal{E}'_{(i,j)}(\ket{0})$ keeps only the second register $\mathsf{E}$ of $\rho_{\mathsf{Q} \otimes \mathsf{E}}(i,j)$ (it traces out its first register $\mathsf{Q}$). Therefore, by properties of the trace distance, 
	\begin{equation}\label{eq:dtrrho} 
	\dtr\left( \rho, \mathcal{E}'_{(i,j)}(\ket{0}) \right) \leq \sqrt{2\varepsilon} \ .
	\end{equation} 
	By using the triangular inequality,
	\begin{align*}
		\dtr\left( \rho_{\mathsf{Q}\otimes \mathsf{E}}(i,j), \sigma_{\mathsf{Q}}(i) \otimes \mathcal{E}'_{(i,j)}\left( \ket{0} \right) \right) &\leq \dtr\left( \rho_{\mathsf{Q}\otimes \mathsf{E}}(i,j), \ket{\psi} \right) + \dtr\left( \ket{\psi}, \sigma_{\mathsf{Q}}(i) \otimes \mathcal{E}'_{(i,j)}\left( \ket{0} \right) \right) \\
		&=  \dtr\left( \rho_{\mathsf{Q}\otimes \mathsf{E}}(i,j), \ket{\psi} \right) + \dtr\left( \sigma_{\mathsf{Q}}(i) \otimes \rho, \sigma_{\mathsf{Q}}(i) \otimes \mathcal{E}'_{(i,j)}\left( \ket{0} \right) \right) \\
		&=  \dtr\left( \rho_{\mathsf{Q}\otimes \mathsf{E}}(i,j), \ket{\psi} \right) + \dtr\left( \rho, \otimes \mathcal{E}'_{(i,j)}\left( \ket{0} \right) \right) \\
		&\leq 2\sqrt{2\varepsilon}
	\end{align*}
	where we used Equations \eqref{eq:dtrF} and \eqref{eq:dtrrho}. This completes the proof. 
\end{proof}

\begin{proof}[Proof of Lemma~\ref{lemma:locc}]

	Let $\ell$ be the number of steps of the sampler $\sampler$. Recall that after the~$\ell$-th step, the sampler $\sampler$ measures part of the register~$\mathsf{Q}$. According to Lemma \ref{lemma:sub-locc}, we have 
	$$
	\dtr \left( \rho_{\mathsf{Q}\otimes\mathsf{E}}(\ell,k(\ell)),\sigma_{\mathsf{Q}}(\ell) \otimes \mathcal{E}'_{\ell,k(\ell)}\left( \ket{0} \right) \right) \leq 2\sqrt{2\varepsilon} \ , 
	$$
	where $\mathcal{E}'_{\ell,k(\ell)}$ is a QPT algorithm that only needs the description of $\sampler$, $\extractor$, and the knowledge of the classical string $z$ to run.    
	Therefore, after performing the measurement of~$\sampler$ on $\sigma_{\mathsf{Q}}(\ell) \otimes \mathcal{E}'_{\ell,k(\ell)}\left( \ket{0} \right) $, the post-measurement state is within trace distance $\leq 2\sqrt{2\varepsilon}$ from the post-measurement state after applying the same measurement on~$\rho_{\mathsf{Q}\otimes\mathsf{E}}(\ell,k(\ell))$. Let the former be
	$$
	 \ket{y,\Phi}\otimes\mathcal{E}'_{\ell,k(\ell)}\left( \ket{0} \right) \ ,
	$$
	where~$y$ is classical.
	To build an LMZ extractor, it suffices to remove the interaction between~$\extractor$ and~$\sampler$. We first run~$\sampler$ once and let it perform its measurement to obtain~$\ket{y,\Phi}$. Then, we let~$\extractor$ perform its last steps over~$ \ket{y,\Phi}\otimes\mathcal{E}'_{\ell,k(\ell)}\left( \ket{0} \right)$. Note that this defines an LMZ extractor since~$\mathcal{E}'_{\ell,k(\ell)}\left( \ket{0} \right)$ can be computed in polynomial time and independent of $(i)$ the execution of $\sampler$ and $(ii)$ its measurement output.
	
	\end{proof}

\subsection{Obliviousness and black-box reductions}\label{subsec:OSred}
All definitions of this subsection can be extended to the general class of \textit{distributional problems} as well. Recall that a distributional problem~$\mathrm{P}$ is a pair~$(\nprel,\mathrm{D})$ where~$\mathrm{R}$ is an~$\NP$ relation and~$\mathrm{D} = \{\mathrm{D}_{\lambda}\}_{\lambda}$ is a polynomially sampleable ensemble over the instances of~$\nprel$. The problem~$\mathrm{P}$ asks for finding a witness for an instance that has been sampled according to~$\mathrm{D}$.
We note that the search~$\lwe$ problem belongs to this class.

\begin{definition}[Quantum Samplers] \label{def:r-qsampler}
	Let~$\lambda$ be the security parameter. Let~$\mathrm{P}=(\nprel,\mathrm{D})$ be a distributional problem.  Let~$\qsampler$ be a QPT algorithm that has the following specification:
	\begin{itemize}
		\item[] $\qsampler\left(1^{\lambda}, {\tilde{x}},\ket{0}^{\poly(\lambda)}\right)$: Given the parameter~$1^{\lambda}$, a string~$\tilde{x}$, and a polynomial number of ancillas initialized to~$\ket{0}$ as inputs, it returns a string~$x$ of size~$\poly(\lambda)$ that has~$\tilde{x}$ as a prefix.
	\end{itemize}
	We say that~$\qsampler$ is a quantum~$\mathrm{P}$ sampler if there exists a probability distribution~$\{\widetilde{D}_{\lambda}\}_\lambda$ such that for~$\tilde{x}$ sampled from~$\widetilde{D}_\lambda$, the distribution of~$\qsampler\left(1^{\lambda}, {\tilde{x}}, \ket{0}^{\poly(\lambda)}\right)$ is within statistical distance~$\negl(\lambda)$ from~$\mathrm{D}_{\lambda}$.
\end{definition}

One can define witness-oblivious samplers by adapting Definition~\ref{def:qr-obl}. We are interested in preservation of witness-obliviousness  under  reductions.
We begin by recalling the definition of reductions with respect to distributional problems.

\begin{definition} \label{def:prob-red}
	A distributional problem~$\mathrm{P_1}=(\nprel_1,\mathrm{D}_1)$ is randomized Karp-reducible to~$\mathrm{P}_2=(\nprel_2,\mathrm{D}_2)$ if there exists:
	\begin{itemize}\setlength{\itemsep}{5pt}
		\item[$\bullet$] a PPT algorithm~$\mathcal{A}$ that maps instances of~$\mathrm{P}_1$ to instances of~$\mathrm{P}_2$ such that~$\mathcal{A}(\mathrm{D}_1)$ is within negligible statistical distance from~$\mathrm{D}_2$ over the randomness of~$\mathcal{A}$, 
		\item[$\bullet$] a PPT or QPT
		algorithm~$\mathcal{B}$ for~$\mathcal{A}$ such that
		$$
		\forall x_1,y_2 : (\mathcal{A}(x_1;r),y_2) \in \nprel_2 \implies (x_1, \mathcal{B}(x_1,y_2,r)) \in \nprel_1 \ ,
		$$
		with non-negligible probability over the randomness of~$\mathcal{B}$. Note that~$\mathcal{B}$ has the randomness~$r$ of~$\mathcal{A}$ as part of its input (and can use extra randomness).
	\end{itemize}
\end{definition}

The following theorem states that witness-obliviousness is preserved under randomized Karp reductions.

\begin{lemma} \label{lemma:bb-reduction}
	Let~$\mathrm{P}_1$ and~$\mathrm{P}_2$ be two distributional problems. Assume that~$\mathrm{P}_1$ is randomized Karp-reducible to~$\mathrm{P}_2$ with the associated algorithms~$\mathcal{A}$ and~$\mathcal{B}$. If~$\qsampler$ is a quantum witness-oblivious~$\mathrm{P}_1$ sampler, then~$ \mathcal{A}(\qsampler)$ is a quantum witness-oblivious~$\mathrm{P}_2$ sampler.\end{lemma}
\begin{proof}
	Let~$x_1 \leftarrow \qsampler$ and~$x_2 \leftarrow \mathcal{A}(x_1;r)$. Suppose that there exists a valid QPT extractor~$\mathcal{E}_2$ that finds a witness $y_2$ for the instance~$x_2$. One can build a new extractor~$\mathcal{E}_1$ for~$\qsampler$ as follows. To find a witness for~$x_1$, the new extractor $(i)$ collects the randomness~$r$ of~$\mathcal{A}$, $(ii)$ finds the witness~$y_2$ for~$x_2$ using~$\mathcal{E}_2$, and then $(iii)$ applies~$\mathcal{B}(x_1,y_2,r)$. The output of~$\mathcal{B}$ is a witness for~$x_1$ according to the definition of the randomized Karp reduction. It suffices to note that~$\mathcal{B}$ is indeed a valid extractor for~$\mathcal{Q}$.
\end{proof}

Note that the~$\mathrm{P}_2$ sampler is witness-oblivious under the hardness assumption of~$\mathrm{P}_1$. 
Many classical reductions in the context of lattice problems fall into the above framework.

\subsection{Reducing oblivious  $\lwe$ sampling to $\qlwe$.}\label{subsec:OS} We complete this section by providing a general approach to design a quantum witness oblivious sampler via a single unitary and a final measurement. We show that producing $\lwe$ samples in an oblivious manner reduces to synthesizing a quantum state that is a superposition of all~$\lwe$ samples, as defined in~\cite{CLZ22}. This  state synthesis problem is called the~$\qlwe$ problem. 
\begin{definition}[$\qlwe$ State]\label{def:m-LWEstates}
	Let $m \geq  n \geq 1$ and~$q \geq 2$ be integers, and~$f$ be an amplitude function whose domain is~$\mathbb{Z}/q\mathbb{Z}$. 
	The parameters~$m,n,q,f$ are functions of some security parameter~$\lambda$.  For~$\vec{A} =\left( \vec{a}_{1}| \dots | \vec{a}_{m}\right)^{\intercal} \in \left( \mathbb{Z}/q\mathbb{Z}\right)^{m \times n}$, the~$\QLWE{m,n,q,f}{\vec{A}}$ state is defined as
	$$
	\QLWE{m,n,q,f}{\vec{A}} \eqdef\frac{1}{\sqrt{Z_f(\vec{A})}} \sum_{\vec{s} \in \left( \mathbb{Z}/q\mathbb{Z} \right)^{n}} \sum\limits_{\vec{e} \in (\mathbb{Z}/q\mathbb{Z})^m}  \Motimes\limits_{i= 1}^{m} f(e_i) \ket{\langle \vec{a}_i,\vec{s} \rangle + e_i \bmod q}\ ,
	$$
	where~$Z_f(\vec{A})$ is the normalization scalar such that~$\QLWE{m,n,q,f}{\vec{A}}$ becomes a unit vector.

	To simplify notation, when it is clear from the context, we will drop the dependency on $m, n, q, f$, and the matrix~$\vec{A}$.
\end{definition}

The normalization term $Z_f(\vec{A})$, which guarantees that $\qlwe$ is a {\em valid} quantum state, will play an important role. 	
In particular, we will require that $Z_f(\vec{A}) \approx q^{n}$. We will discuss this matter in detail in Section~\ref{subsec:Gauss}, 
when instantiating our algorithm to the case where $|f|^{2}$, i.e., the noise distribution of the measured $\lwe$ sample, is a Gaussian distribution.

Constructing this state was studied in~\cite{CLZ22} in order to solve the \textit{Short Integer Solution} (\textsf{SIS}) problem with some specific parameters. We note that~\cite{CLZ22} neglected the normalization factor~$Z_f(\vec{A})$ by assuming that it is always equal to~$q^n$, see for example~\cite[Def.~9 \& Cor.~9]{CLZ22}. 
In the general problem of constructing an~$\qlwe$ state, one should take the normalization into account. For instance, it was shown in~\cite{DRT21}  that this normalization factor should be handled with care when $|f|^{2}$ concentrates the error weight close to the minimum distance of the spanned linear code.

We  also stress that~\cite[Def.~9]{CLZ22} only allows non-negative real-valued amplitude functions, while we allow complex-valued ones. Although we only use real-valued (but not positive) instantiations of the amplitude function in this work since they are sufficient for our purposes, more choices of the function might have further applications.

\begin{definition}[$\qlwe$ Problem]\label{prob:qlwe}
	Let $m,n,q,f,\lambda$ as in Definition~\ref{def:m-LWEstates}. The~$\qlwe_{m,n,q,f}$ problem is as follows: given as input a matrix~$\vec{A} \in \left( \mathbb{Z}/q\mathbb{Z}\right)^{m \times n}$, the goal is to build the $\QLWE{m,n,q,f}{\vec{A}}$ state. More formally, we say that a QPT algorithm~$\mathcal{S}$ solves $\qlwe_{m,n,q,f}$ if there exists~$M \leq \poly(\lambda)$ such that given $1^{\lambda}$, a uniform~$\vec{A}$ and~$\ket{0}^{m\log q}\ket{0}^M$ as inputs, algorithm~$\mathcal{S}$ builds a state within
trace distance~$\negl(\lambda)$ from $\QLWE{m,n,q,f}{\vec{A}}\ket{0}^M$, with probability~$1 - \negl(\lambda)$ over the randomness of~$\vec{A}$ and its  measurements.
\end{definition}

Notice that measuring the $\qlwe$ state gives the following $m$ $\lwe$ samples:
$$
\left( (\vec{a}_{1}, \langle \vec{a}_{1},\vec{s} \rangle + e_1 \bmod q),\dots, (\vec{a}_m,\langle \vec{a}_{m},\vec{s} \rangle + e_{m} \bmod q) \right) \ , 
$$
where the $e_{i}$'s are i.i.d.~with distribution $|f|^{2}$, while $\vec{s}$ is uniform and independent.

	In the following theorem, we show that solving~$\qlwe$  using a unitary algorithm provides a 
	witness-oblivious~$\lwe$ sampler by measuring the final superposition. We stress that the result holds 
	even if the~$\qlwe$ solver only provides a state that is only approximately equal (in trace distance) to~$\QLWE{m,n,q,f}{\vec{A}}\ket{0}^{M}$.

\begin{theorem}\label{theo:obliviousLWEsample}
	Let $m,n,q,f,\lambda$ as in Definition~\ref{def:m-LWEstates}. 
	Assume that there exists a unitary QPT algorithm~$\qsampler$ that solves~$\qlwe_{m,n,q,f}$ for some~$M \leq \poly(\lambda)$ number of auxiliary ancillas as input. 
	Then~$\qsampler$ followed by a measurement in the computational basis is a witness-oblivious quantum~$\lwe_{m,n,q,|f|^{2}}$ sampler, assuming the quantum hardness of~$\lwe_{m,n,q,|f|^2}$.
\end{theorem}

\begin{proof}  
	Let~$\mathsf{Q} = \mathsf{C} \otimes \mathsf{W}$ be the register of~$\qsampler$ such that the final measurement is performed upon~$\mathsf{C}$ to obtain the classical output and~$\mathsf{W}$ is the remaining register. Let~$\ket{\psi}$ be the final state of the algorithm~$\qsampler$ over~$\mathsf{Q}$, right before the measurement.
	With probability~$1-\negl(\lambda)$ over the uniform choice of~$\vec{A}$, we have:
	\begin{equation}\label{eqref:lwe~psi}
	\dtr\left(\ket{\psi}, \ \QLWE{m,n,q,f}{\vec{A}} \ket{0}^M \right) \leq \negl(\lambda) \ .
	\end{equation}
	After applying the measurement, the state~$\ket{\psi}$ becomes a mixed state as follows:
	$$
	\sigma_\qsampler \eqdef \sum\limits_{\vec{b} \in (\mathbb{Z}/q\mathbb{Z})^{m}} p_\vec{b} \op{\vec{b}}{\vec{b}} \otimes \op{\phi_\vec{b}}{\phi_\vec{b}} \ ,
	$$
	where~$p_{\vec{b}}$ is the probability of observing~$\vec{b}$ as the outcome, and~$\ket{\phi_{\vec{b}}}$ is the corresponding state in the working space. After the measurement, the other state becomes:
	$$
	\sigma_{\lwe} \eqdef \sum\limits_{\vec{b} \in (\mathbb{Z}/q\mathbb{Z})^{m}} q_\vec{b} \op{\vec{b}}{\vec{b}} \otimes \op{\vec{0}}{\vec{0}} \ ,
	$$
	where~$\{q_\vec{b}\}_\vec{b}$ is the induced distribution of~$\lwe_{m,n,q,|f|^2}$ over its support, namely 
	\begin{equation}\label{eq:qx}  
	q_{\vec{b}} = \mathbb{P}_{\vec{s},\vec{e}}\left( \vec{A}\vec{s} + \vec{e} = \vec{b} \right) \ ,
	\end{equation} 
	where $\vec{s} \in \left( \mathbb{Z}/q\mathbb{Z} \right)^{n}$ is picked uniformly at random  and the $e_{i}$'s are i.i.d. with distribution $|f|^{2}$. 
Using the properties of trace distance, we obtain for a proportion~$1-\negl(\lambda)$ of matrices $\vec{A}$:
	\begin{eqnarray}
	\dtr\left(\sigma_\qsampler,\sigma_\lwe  \right) & \leq & \dtr\left(\ket{\psi},\ \QLWE{m,n,q,f}{\vec{A}} \otimes \ket{0}^{M} \right) \nonumber \\
	& \leq & \negl(\lambda) \ , \label{eq:sigmaSLWE}
	\end{eqnarray}		 
	where in the last line we used Equation \eqref{eqref:lwe~psi}. Using now Equation~\eqref{eq:stat-leq-tr}, we obtain for a proportion~$1-\negl(\lambda)$ of matrices $\vec{A}$:
	$$
	\Delta\left(\{p_\vec{b}\}_\vec{b},\{q_\vec{b}\}_\vec{b}\right) 
	\leq \negl(\lambda)\ .
	$$
	This proves, by using Equation \eqref{eq:qx}, that the sampler $\qsampler$ is a quantum $\lwe$ sampler as stated in Definition~\ref{def:lwe-qsampler}.

		Let us now show that $\qsampler$ followed be a single measurement is a {\it witness-oblivious} quantum sampler as stated in Definition \ref{def:qr-obl}. By assumption, the sampler~$\qsampler$  is a  unitary algorithm. We first consider the case where the extractor $\extractor$ is perfect, \ie{} $\varepsilon = 0$ in Lemma \ref{lemma:locc}. Therefore, we can suppose that the input of $\extractor$ is 
		$\sigma_{\qsampler}$.
	
		Suppose that $\extractor$ is instead given  $\sigma_{\lwe}$, namely $\ket{\vec{b}}\ket{0}^{M}$ with
		 $\vec{b} \eqdef \vec{A}\vec{s}+\vec{e} \in \left( \mathbb{Z}/q\mathbb{Z} \right)^{m}$ such that it has been picked according to $q_{\vec{b}}$ given in Equation \eqref{eq:qx}. 
		 In that case, for a uniform choice of matrices $\vec{A}$, its probability to output $(\vec{s}',\vec{e}')$ such that $\vec{s}' = \vec{s}$ and $\vec{e}' = \vec{e}$ is $\negl(\lambda)$ as we assumed the quantum hardness of $\lwe_{m,n,q,|f|^2}$. However the extractor is given $\sigma_{\sampler}$. 
	 Using the properties of the trace distance, it holds that for a proportion $1-\negl(\lambda)$ of matrices~$\vec{A}$: 
	 	\begin{align*}
	 	\dtr\Big(  \extractor \left( \sigma_\qsampler \right)  ,   \extractor \left(\sigma_\lwe \right) \Big) &\leq \dtr(\sigma_\qsampler,\sigma_\lwe) \\
	 	&\leq  \negl(\lambda) \ .
	 \end{align*}
	where in the last line we used Equation \eqref{eq:sigmaSLWE}. This completes the proof in the perfect case, \ie{}~$\varepsilon = 0$.

	Now, suppose that~$\extractor$ is~$\negl(\lambda)$-valid extractor. According to Lemma \ref{lemma:locc}, it is given a quantum state a trace distance $2\sqrt{2\negl(\lambda)} = \negl(\lambda)$ from $\sigma_{\qsampler}$.
	 To conclude the proof we proceed as above.
\end{proof}

As a direct application of Theorem~\ref{theo:obliviousLWEsample} and Lemma~\ref{lemma:bb-reduction}, we obtain that a unitary solver\linebreak for~$\qlwe_{m,n,q,f}$ gives an witness-oblivious quantum $\lwe_{m',n,q,|f|^2}$ sampler for any integer~$m' \in [n,m]$. Indeed, throwing away the superfluous coordinates is a Karp reduction.

\section{An algorithm for~$\qlwe$}\label{sec:lweOS}

	In Subsection~\ref{subsec:OS}, we have shown that witness-oblivious sampling reduces to the $\qlwe$ problem (Definition~\ref{prob:qlwe}).
	Solving this problem consists in building 
	the~$\qlwe$ state  (as per Definition~\ref{def:m-LWEstates}). This state is an~$m$-fold tensor product where each element corresponds to a single $\lwe$ sample $\langle \vec{a}_{i},\vec{s} \rangle + e_i \bmod q$. 
	Like in~\cite{CLZ22}, our approach to solve the~$\qlwe$ problem singles out each of these elements, analyzes them independently, and finally recombines the results.

	\begin{definition}[Coordinate States]\label{def:psi-states}
		Let $q \geq 2$ and $f: \mathbb{Z}/q\mathbb{Z}\rightarrow \mathbb{C}$ be an amplitude function. We define the coordinate states as follows: 
		$$
		\forall j \in \mathbb{Z}/q\mathbb{Z}, \quad \ket{\psi_j} \eqdef \  \sum_{e = 0}^{q-1} f\left(e \right)\ket{j+e  \bmod q} \ .
		$$
	\end{definition}
	
\subsection{Description of the algorithm}\label{subsec:oS}

	Before going into the details, we  briefly explain how our algorithm solves the $\qlwe$ problem for some arbitrary amplitude function~$f$. It proceeds in three general phases that would ideally work as follows. \newline

\begin{enumerate}[label=\arabic*)] 
	\item[\bf Phase A.] First, it builds the following entangled state
	\begin{equation}\label{eq:qstatentangled} 
		\frac{1}{\sqrt{q^n}} \sum_{\vec{s} \in \left( \mathbb{Z}/q\mathbb{Z}\right)^{n}} \sum_{\vec{e} \in \left( \mathbb{Z}/q\mathbb{Z}\right)^{m}} f^{\otimes m}(\vec{e}) \ket{\vec{s}}\ket{\vec{A}\vec{s} + \vec{e}} = \frac{1}{\sqrt{q^n}} \sum_{\vec{s} \in \left( \mathbb{Z}/q\mathbb{Z}\right)^{n}}\ket{\vec{s}} \Motimes_{j=1}^{m} \ket{\psi_{\langle \vec{a}_{j},\vec{s} \rangle}} \ .
	\end{equation}  
	The efficiency of this step depends on the specific choice of~$f$. 
	\newline

	\item[\bf Phase B.] For each~$j$ in parallel, it recovers $\langle \vec{a}_{j},\vec{s} \rangle$
	 from~$\ket{\psi_{\langle \vec{a}_{j},\vec{s}\rangle}}$ with some probability~$p$ (independent from~$j$). If it fails, the outcome could be thought as
	 special symbol~$\bot$. This operation is not allowed to ``perturb'' $\ket{\psi_{\langle \vec{a}_{j},\vec{s}\rangle}}$: it 
	 has to be reversible. We handle this  by applying some polynomial-time unitary that  maps $\ket{\psi_{\langle \vec{a}_{j},\vec{s}\rangle}}$ to the state
	\[
	\sqrt{p}\ket{\langle \vec{a}_{j},\vec{s} \rangle}\ket{0} + \sqrt{1-p} \ket{a}\ket{1} \ , 
	\]
	for some~$\ket{a}$ which does not play any role. We interpret any quantum state whose last qubit is~$\ket{1}$ as~$\bot$. The quality 
	of that step is quantified by the success probability~$p$.
	\newline

	\item[\bf Phase C.]  Using the successful coordinates, the algorithm collects 
	some linear equations $\langle \vec{a}_{j},\vec{s}\rangle$ (for known $\vec{a}_j$'s). The next step 
	is to recompute~$\vec{s}$ by Gaussian elimination. This allows to erase it from the content of the first register, i.e., 
 disentangling the state in Equation~\eqref{eq:qstatentangled} and solving the $\qlwe$ problem. However, note that Phase~B 
 only enables to recover each~$\langle \vec{a}_{j},\vec{s} \rangle$ with some probability~$p$. Therefore our approach will work if the number of non-$\bot$ coordinates is no smaller than~$n$ in order to expect to have a non-singular linear system to solve, namely if~$m = (n+\log\log q)/p \cdot  \omega(\log \lambda)$. 
		Therefore, the success probability~$p$ considered at Phase~B has to be sufficiently large for the purpose of efficiency.  
	\newline
\end{enumerate}

Combining the steps above, one obtains Algorithm~\ref{algo:StMd}. Steps~\ref{inst:s} to~\ref{inst:sum} of  Algorithm~\ref{algo:StMd} correspond to Phase~A above, Steps~\ref{inst:emptylabel} and~\ref{inst:U} correspond to Phase~B above, and Steps~\ref{inst:GE} and~\ref{state:fin} correspond to Phase~C above.

\begin{algorithm}[h!]
	\caption{Quantum $\QLWE{m,n,q,f}{\vec{A}}$ Solver.}\label{algo:StMd} 
	\begin{flushleft}
		{\bf Parameters:} $m,n,q$ and $f$. 
	\end{flushleft}
	\begin{algorithmic}[1]
		\Input  $\vec{A} \eqdef \left( \vec{a}_{1} | \dots | \vec{a}_{m}\right)^{\intercal} \in \left( \mathbb{Z}/q\mathbb{Z}\right)^{m \times n}$.
		\Output A quantum state $\ket{\varphi}$. 
		
		\vspace{0.5cm}
		\State\label{inst:s} Build the state $\frac{1}{\sqrt{q^n}}\sum\limits_{\vec{s} \in \left(\mathbb{Z}/q\mathbb{Z}\right)^n} \ket{\vec{s}}$.
		\State\label{inst:error} Build the state $\sum\limits_{\vec{e} \in \left( \mathbb{Z}/q\mathbb{Z}\right)^m} \Motimes\limits_{i=1}^{m} f(e_{i}) \ket{e_{i}}$.
		\State\label{inst:prod} Consider the joint state of Steps 1 and 2 to get 
		$$
		\frac{1}{\sqrt{q^n}} \sum\limits_{\vec{s} \in \left( \mathbb{Z}/q\mathbb{Z}\right)^n} \ket{\vec{s}} \sum\limits_{\vec{e} \in \left( \mathbb{Z}/q\mathbb{Z}\right)^m} \Motimes\limits_{i=1}^{m} f({e}_i)  \ket{e_i} \ .
		$$
		\State\label{inst:sum} Apply the quantum unitary $\ket{\vec{s},\vec{e}} \mapsto \ket{\vec{s}, \langle \vec{a}_{1},\vec{s}\rangle + e_{1},\dots,\langle \vec{a}_{m},\vec{s} \rangle + e_{m}}$ to get
		$$
		\frac{1}{\sqrt{q^n}} \sum\limits_{\vec{s} \in \left( \mathbb{Z}/q\mathbb{Z}\right)^n} \ket{\vec{s}} \sum\limits_{\vec{e} \in \left( \mathbb{Z}/q\mathbb{Z}\right)^m} \Motimes\limits_{i=1}^{m} f({e}_i)\ket{\langle \vec{a}_i,\vec{s} \rangle +e_i} = \frac{1}{\sqrt{q^{n}}}\sum\limits_{\vec{s} \in \left( \mathbb{Z}/q\mathbb{Z}\right)^n} \ket{\vec{s}}\Motimes\limits_{i=1}^{m} \ket{\psi_{\langle \vec{a}_{i},\vec{s}\rangle}} \ .
		$$
		\State\label{inst:emptylabel} Append~$\ket{0}^m$.
		\State\label{inst:U} Apply the unitary $\vec{I}\otimes \vec{V}^{\otimes m}$ with $\vec{V}$ as defined in Equation~\eqref{eq:unitaryV}, to obtain 
		\begin{align*}
			\frac{1}{\sqrt{q^{n}}}\sum\limits_{\vec{s} \in \left( \mathbb{Z}/q\mathbb{Z}\right)^n} \ket{\vec{s}}\Motimes\limits_{i=1}^{m} \vec{V} \left(\ket{\psi_{\langle \vec{a}_{i},\vec{s}\rangle}}\ket{0} \right)\ .
		\end{align*}
		\State \label{inst:GE} Apply the quantum unambiguous Gaussian elimination as given in Equation \eqref{eq:unitaryUGE} to get
		$$
		\vec{U}_{\mathcal{A}_{\textup{GE}}} \left( \frac{1}{\sqrt{q^{n}}}\sum\limits_{\vec{s} \in \left( \mathbb{Z}/q\mathbb{Z}\right)^n} \ket{\vec{s}}\Motimes\limits_{i=1}^{m} \vec{V} \left(\ket{\psi_{\langle \vec{a}_{i},\vec{s}\rangle}} \ket{0}\right) \right)\ .
		$$
		\State\label{state:fin}Apply $\vec{I} \otimes \left( \vec{V}^{\dagger}\right)^{\otimes m}$ and output the resulting quantum state.
	\end{algorithmic} 
\end{algorithm}

More details are required to 
make Steps~\ref{inst:error}, \ref{inst:U} and~\ref{inst:GE} explicit. For Step~\ref{inst:error}, we assume that we can efficiently implement an approximation of the state (see Condition~\ref{cdt:errorSuperposition} of Theorem~\ref{theo:solveQlwe}). A realization for a specific amplitude function~$f$ will be discussed in Lemma~\ref{cor:vartheta-encod}. Step~\ref{inst:U} relies on a  unitary~$\vec{V}$ satisfying: 
\begin{equation}\label{eq:unitaryV} 
	\forall x \in \mathbb{Z}/q\mathbb{Z}: \vec{V} \left(\ket{\chi_x}\ket{0} \right)= \ket{\chi_x} \left( u_x \ket{0} + \sqrt{1-|u_x|^{2}}\ket{1} \right) \ , 
\end{equation} 
where~$u_x$ is an approximation of~$(\min |\widehat{f}|)/\widehat{f}(-x)$  (see Condition~\ref{cdt:eapx} of Theorem~\ref{theo:solveQlwe}). We will explain in Lemma~\ref{lemma:cpxV} how to implement~$\vec{V}$.
Note that up to the numerical inaccuracy, the unitary~$\vec{V}$ can be viewed as an implementation of the unambiguous measurement from~\cite{CB98} (see Appendix~\ref{app:CB98}).
Step~\ref{inst:GE} uses a version of a Gaussian elimination algorithm~$\mathcal{A}_{\textup{GE}}$  that works as 
follows when
given as input a matrix $\vec{A} \eqdef\left( 
		\vec{a}_{1}| \dots | \vec{a}_{m} \right)^{\intercal} \in \left( \mathbb{Z}/q\mathbb{Z}\right)^{m \times n}$
		and $m$ equations $(y_{i})_{1\leq i \leq m}$ where $y_{i} = \langle \vec{a}_{i},\vec{s} \rangle$ or $y_i = \perp$ (some equations may be erased):
		it first tests  whether the input matrix $\vec{A}$ is invertible modulo~$q$ (which is not required to be prime); if it is, it 
		then outputs the unique solution~$\vec{s}$; if it is not, it outputs~$\bot$. Algorithm~$\mathcal{A}_{\textup{GE}}$ is deterministic polynomial-time and  has the following properties that will prove useful in our analysis of Algorithm~\ref{algo:StMd}:
		\begin{itemize}\setlength{\itemsep}{5pt}
		\item[$\bullet$] it is unambiguous, in the sense that it never outputs an incorrect solution, \ie\ it either outputs the 
		valid~$\vec{s}$ or it fails and outputs~$\bot$;
		\item[$\bullet$] if~$\vec{A}$ is sampled uniformly, and the number of non-$\bot$ input $y_i$'s is~$(n+\log \log q) \omega(\log \lambda)$
		and the indices of the non-$\bot$ input $y_i$'s are chosen independently from~$\vec{A}$, then~$\mathcal{A}_{\textup{GE}}$ returns~$\bot$ with probability~$\negl(\lambda)$  (this can be obtained, e.g., by adapting~\cite[Claim~2.13]{BLPRS13});
		\item[$\bullet$] for any fixed~$\vec{A}$, if the indices of the non-$\bot$ $y_i$'s are chosen randomly and independently
		from the rest, then the success probability  is the same for every~$(\vec{s} \in \mathbb{Z}/q\mathbb{Z})^n$.
		\end{itemize}
In Algorithm~\ref{algo:StMd}, we consider a version of the Gaussian elimination~$\mathcal{A}_{\textup{GE}}$ that is quantized as follows. For any $(\vec{s},\vec{x},\vec{b}) \in  \left( \mathbb{Z}/q\mathbb{Z} \right)^{n} \times \left( \mathbb{Z}/q\mathbb{Z}\right)^{m}  \times  \{0,1\}^{m}$,
\begin{equation}\label{eq:unitaryUGE} 
	\vec{U}_{\mathcal{A}_{\textup{GE}}}:  \ket{\vec{s}}  \Motimes_{i=1}^{m} \ket{ x_i,b_i } \longmapsto \ket{ \vec{s}-\mathcal{A}_{\textup{GE}}(\vec{A},\left( y_i \right)_{1\leq i \leq m} )} \Motimes_{i=1}^{m} \ket{ x_i,b_i } \ . 
\end{equation} 
where $y_i = x_i$ if $b_i = 0$, and~$y_i = \bot$ otherwise. To handle the potential output~$\bot$ of~$\mathcal{A}_{\textup{GE}}$, 
we embed the first quantum register in Equation~\eqref{eq:unitaryUGE} into~$\mathbb{C}^{2q}$ where $\left( \ket{x},\ket{\bot}_x\right)_{x \in \mathbb{Z}/q\mathbb{Z}}$ is the computational basis (for some arbitrary symbols~$\bot_x$).

The following theorem gives conditions under which Algorithm~\ref{algo:StMd} solves the $\qlwe$ problem 
in time~$\poly(\lambda)$.

	\begin{theorem}\label{theo:solveQlwe}
	Let $m \geq n \geq 1$ and~$q \geq 2$ be integers, and~$f : \mathbb{Z}/q\mathbb{Z} \rightarrow \mathbb{C}$ be an amplitude function. The parameters~$m,n,q,f$ are functions of some security parameter~$\lambda$ with~$m,\log q \leq \poly(\lambda)$. Assume that the following conditions hold:
	\begin{enumerate}[label=\textcolor{blue}{\arabic*.}, ref=\arabic*]\setlength{\itemsep}{5pt}

		\item\label{cdt:errorSuperposition} there exists a $\poly(\lambda)$-time algorithm that builds a state within~$\negl(\lambda)$ trace distance of the state~$\sum_{e \in \mathbb{Z}/q\mathbb{Z}} f(e)\ket{e}$;

		\item\label{cdt:eapx} there exists a
		$\poly(\lambda)$-time algorithm that, 
		given~$x$ as input, outputs  $u_x \eqdef (\min |\widehat{f}|)/\widehat{f}(-x) + e_{\textup{apx}}(x)$ on~$\poly(\lambda)$ bits with $\max_x |e_{\textup{apx}}(x)| = \negl(\lambda)/\sqrt{q^{n}}$;
		
		\item\label{cdt:mpn} we have that~$m = (n+\log\log q)  /p \cdot \omega(\log \lambda)$, where $p \eqdef q \cdot \min |\widehat{f}|^2$;

		\item\label{cdt:Za} assuming that $\vec{A} \in \left( \mathbb{Z}/q\mathbb{Z} \right)^{m \times n}$ is uniformly distributed, we have
	
		$$
		\mathbb{P}_{\vec{A}}\left(  \left| \frac{Z_{f}(\vec{A})}{q^{n}} - 1 \right|  \geq \negl(\lambda) \right) = \negl(\lambda) \ ,
		$$
		where~$Z_f(\vec{A})$ is the normalization scalar such that~$\QLWE{m,n,q,f}{\vec{A}}$ becomes a unit vector, as per Definition~\ref{def:m-LWEstates}.
	\end{enumerate}

	Then Algorithm~\ref{algo:StMd} runs in time~$\poly\left(\lambda\right)$
	 and, for a proportion $1-\negl(\lambda)$ of matrices~$\vec{A} \in \left( \mathbb{Z}/q\mathbb{Z}\right)^{m \times n}$, it outputs a quantum state $\ket{\varphi}$ such that
	\begin{equation}\label{eq:dtr} 
	\dtr\left( \ket{\varphi}, \ket{0}^{n \log q}\QLWE{m,n,q,f}{\vec{A}} \ket{0}^m \right) =  \negl(\lambda) \ .
\end{equation} 
\end{theorem}

The first two conditions enable an efficient implementation of Algorithm~\ref{algo:StMd}. In Condition~\ref{cdt:mpn}, the 
value~$p$ refers to the success probability in recovering~$j$ from~$\ket{\psi_j}$. This condition ensures that~$m$ is sufficiently large
for the unambiguous Gaussian elimination algorithm to succeed with probability~$1-\negl(\lambda)$. Note that the condition on~$m$ and the fact that~$m \leq \poly(\lambda)$ imply that we must have~$p \geq 1/\poly(\lambda)$. The latter implies that~$\widehat{f}(x)$ is non-zero for all~$x \in \mathbb{Z}/q\mathbb{Z}$, a condition that is necessary to rely on the measurement from~\cite{CB98} and, more concretely, for the unitary~$\vec{V}$ used in Step~\ref{inst:U} of Algorithm~\ref{algo:StMd} to be well-defined (see  Condition~\ref{cdt:mpn}). Still concerning Condition~\ref{cdt:mpn}, the lower bound on~$m$ is to ensure that a uniform~$m \times n$
matrix modulo~$q$ has an image of size~$(\mathbb{Z}/q\mathbb{Z})^n$ with overwhelming probability. If~$q$ is prime, this condition can be simplified to~$m = n/p \cdot  \omega(\log \lambda)$.   Finally, Condition~\ref{cdt:Za} 
intuitively states that the parametrization of $\lwe$ provides a unique solution with overwhelming probability. The last two conditions 
can be simplified if~$q$ is assumed to be prime.

We will first consider the correctness of Algorithm~\ref{algo:StMd} (Lemma~\ref{lemma:dTR}), and then analyze its run-time (Lemma~\ref{lemma:runningTime}).

\subsection{Correctness} 

The purpose of the unitary~$\vec{V}$ (introduced in Equation~\eqref{eq:unitaryV}) is to 
recover~$\langle \vec{a}_{i},\vec{s} \rangle$ 
from~$\ket{\psi_{\langle \vec{a}_{i},\vec{s} \rangle}}$.
 More formally, we have the following lemma.
 \begin{lemma}\label{lemma:unitaryV}
 	Using notations of Theorem \ref{theo:solveQlwe} and with $\vec{V}$ as defined in Equation~\eqref{eq:unitaryV}, we have
 	\begin{equation*} 
 		\forall j \in \mathbb{Z}/q\mathbb{Z}, \; \vec{V} \left(\ket{\psi_j}\ket{0} \right) = \sqrt{p} \ket{j}\ket{0} +  \sqrt{1-p}\ket{\eta_{j}}\ket{1} + \ket{\textup{error}_j} \ ,
 	\end{equation*} 
 	for some quantum states $\ket{\eta_{j}}$ and $\ket{\textup{error}_j}$ with $\max_j \| \ket{\textup{error}_j} \| = \negl(\lambda)/\sqrt{q^{n}}$.
 \end{lemma}
 
 \begin{proof} Let us write the $\ket{\psi_{j}}$'s (Definition \ref{def:psi-states}) in the Fourier basis $(\ket{\chi_x})_{x \in \mathbb{Z}/q\mathbb{Z}}$. We have, for all~$j \in \mathbb{Z}/q\mathbb{Z}$:
 	\begin{align}
 		\ket{\psi_{j}} &= \sum_{e \in \mathbb{Z}/q\mathbb{Z}} f\left(e\right)\ket{j+e\bmod q} \nonumber\\
 		&= \frac{1}{\sqrt{q}}  \sum_{e \in \mathbb{Z}/q\mathbb{Z}} f\left(e\right) \sum_{x \in \mathbb{Z}/q\mathbb{Z}} \omega_{q}^{-(j+e)x}\ket{\chi_{x}} \quad \quad \left(\mbox{by Lemma \ref{lemma:qftm1}} \right) \nonumber \\
 		&= \sum_{x \in \mathbb{Z}/q\mathbb{Z}} \left( \frac{1}{\sqrt{q}} \sum_{e \in \mathbb{Z}/q\mathbb{Z}} f\left(e \right) \omega_{q}^{-xe} \right) \omega_{q}^{-jx} \ket{\chi_{x}} \nonumber \\
 		&= \sum_{x \in \mathbb{Z}/q\mathbb{Z}} \widehat{f}\left(-x\right) \omega_{q}^{-jx} \ket{\chi_{x}} \ . \nonumber
 	\end{align}
 	Therefore, by linearity and definition of~$\vec{V}$, we have:
 	\begin{align*}
 		\vec{V}\left(\ket{\psi_j}\ket{0}\right) &= \sum_{x \in \mathbb{Z}/q\mathbb{Z}} \widehat{f}(-x)  \omega_{q}^{-jx} \; \vec{V} \left(  \ket{\chi_x}\ket{0} \right) \\
 		&= \underbrace{\left( \sum_{x \in  \mathbb{Z}/q\mathbb{Z}} u_{x}\;  \widehat{f}(-x) \; \omega_{q}^{-jx} \ket{\chi_x}  \right)}_{\eqdef \ket{\psi_{j,0}}}\ket{0} + \underbrace{\left( \sum_{x \in  \mathbb{Z}/q\mathbb{Z}} \sqrt{1-|u_{x}|^{2}}\;  \widehat{f}(-x) \; \omega_{q}^{-jx} \ket{\chi_x}  \right)}_{\eqdef \ket{\psi_{j,1}}}\ket{1} \ .
 	\end{align*}
 	Let us consider~$\ket{\psi_{j,0}}$. By definition of~$u_x$ and~$p$, we have:
 	\[
 		\ket{\psi_{j,0}} =  
 		\underbrace{\sqrt{q} \cdot \min |\widehat{f}| \; \left( \frac{1}{\sqrt{q}}\; \sum_{x \in \mathbb{Z}/q\mathbb{Z}} \omega_{q}^{-jx}\ket{\chi_{x}}\right)}_{= \sqrt{p} \ket{j}} 
 		+ \underbrace{\sum_{x \in \mathbb{Z}/q\mathbb{Z}} e_{\textup{apx}}(x) \widehat{f}(-x)\omega_{q}^{-jx}\ket{\chi_{x}}}_{\eqdef \ket{\textup{error}_{j,0}}} \ . 
 	\]
 	Notice that:
 	$$
 	\| \ket{\textup{error}}_{j,0} \|^{2} = \sum_{x \in \mathbb{Z}/q\mathbb{Z}} e_{\textup{apx}}(x)^{2} \; |\widehat{f}(-x)|^{2} \leq \left( \max_{x \in \mathbb{Z}/q\mathbb{Z}} |e_{\textup{apx}}(x)|\right)^2 \ . 
 	$$
 	Hence, so far, we have:
 	\begin{equation}\label{eq:Upsij} 
 		\vec{V} \left(\ket{\psi_j}\ket{0} \right) = \sqrt{p}\ket{j}\ket{0} + \ket{\textup{error}_{j,0}}\ket{0} + \ket{\psi_{j,1}}\ket{1} \ , 
 	\end{equation}
 	where $\| \ket{\textup{error}_{j,0}} \| = \negl(\lambda)/\sqrt{q^{n}}$, by assumption on~$\max_x |e_{\textup{apx}}(x)|$.
 	Notice that $\vec{V} (\ket{\psi_j}\ket{0})$ is a quantum state as~$\vec{V}$ is unitary.  Therefore, we can write
 	$$
 	\ket{\psi_{j,1}} = \sqrt{1-p}\ket{\eta_{j}} + \ket{\textup{error}_{j,1}} \ ,
 	$$
 	for some quantum states $\ket{\eta_{j}}$ and $\ket{\textup{error}_{j,1}}$ such that $\| \ket{\textup{error}_{j,1}} \| = \negl(\lambda)/\sqrt{q^{n}}$. Plugging this into Equation~\eqref{eq:Upsij} gives the result. 
 \end{proof}
 
 As can be seen from Lemma~\ref{lemma:unitaryV}, the transformation~$\vec{V}$ introduces an error term. It basically comes from the fact that we only assume that we can approximate $(\min | \widehat{f}|)/\widehat{f}(-x)$ (as opposed to exactly computing it). This seems necessary for our subsequent choice of~$f$. Ideally, we would analyze the correctness of Algorithm~\ref{algo:StMd} as if we were applying a unitary~$\vec{W}$ (that we do not know how to implement efficiently) such that
 \begin{equation*} 
 		\forall j \in \mathbb{Z}/q\mathbb{Z}, \quad \vec{W} \left(\ket{\psi_j}\ket{0} \right) = \sqrt{p} \ket{j}\ket{0} +  \sqrt{1-p}\ket{\eta_{j}}\ket{1} \ ,
 	\end{equation*}

 The value $\langle \vec{a}_{i},\vec{s}\rangle$ appears (in superposition) in the first 
 register of $\vec{W}(\ket{\psi_{\langle \vec{a}_{i},\vec{s}\rangle}}\ket{0})$, for all~$i \leq m$. Therefore, 
 applying the unitary $\vec{U}_{\mathcal{A}_{\textup{GE}}}$ as in Step~\ref{inst:GE}
 to the quantum state 
 \begin{equation*}\label{eq:VQS} 
 	\frac{1}{\sqrt{q^{n}}}\sum_{\vec{s} \in \left( \mathbb{Z}/q\mathbb{Z}\right)^{n}}\ket{\vec{s}} \Motimes_{i=1}^{m} \vec{W}\left(\ket{\psi_{\langle \vec{a}_{i},\vec{s} \rangle}}\ket{0} \right)
 \end{equation*} 
 will allow us to erase~$\vec{s}$ from the first register. More precisely, we hope that after Step~\ref{inst:GE}, the quantum state
 \begin{equation*}
 	\vec{U}_{\mathcal{A}_{\textup{GE}}}\left( \frac{1}{\sqrt{q^{n}}}\sum_{\vec{s} \in \left( \mathbb{Z}/q\mathbb{Z}\right)^{n}}	\ket{\vec{s}} \Motimes_{i=1}^{m} \vec{W}\ket{\psi_{\langle \vec{a}_{i},\vec{s} \rangle}}\ket{0} \right) 
 \end{equation*} 
will be ``close'' to the disentangled state
 \begin{equation*}
 	\frac{1}{\sqrt{Z_f(\vec{A})}}\sum_{\vec{s}\in \left( \mathbb{Z}/q\mathbb{Z}\right)^{n}}\ket{\vec{0}} \Motimes_{i=1}^{m}\vec{W}\ket{\psi_{\langle \vec{a}_{i},\vec{s} \rangle}}\ket{0}.
 \end{equation*} 
 Notice now that applying~$\vec{I} \otimes \left( \vec{V}^{\dagger}\right)^{\otimes m}$,
 as in Step~\ref{state:fin}, to the state above does not yield the quantum state~$\ket{\vec{0}}\QLWE{m,n,q,f}{\vec{A}}\ket{0}$. Instead, this would hold if 
 we were rather applying  $\vec{I} \otimes \left( \vec{W}^{\dagger} \right)^{\otimes m}$. But we do not know how to 
 implement~$\vec{W}$ efficiently. In the following two lemmas we show that applying~$\vec{V}$ and~$\vec{V}^{\dagger}$ lead to a quantum state that is close with respect to the trace distance to the case where we would instead apply~$\vec{W}$ and~$\vec{W}^{\dagger}$. 
 \begin{lemma}\label{lemma:W1} 
 	Using notations of Theorem \ref{theo:solveQlwe} and  letting
 	\begin{equation}\label{eq:varphiP} 
 		\ket{\varphi'} \eqdef \left(\vec{I} \otimes \left( \vec{V}^{\dagger} \right)^{\otimes m} \right)   \vec{U}_{\mathcal{A}_{\textup{GE}}} \hspace*{-.1cm} \left( \frac{1}{\sqrt{q^{n}}}
 		\hspace*{-.3cm}\sum_{\vec{s} \in \left( \mathbb{Z}/q\mathbb{Z}\right)^{n}} \hspace*{-.3cm}
 		\ket{\vec{s}} \Motimes_{i=1}^{m} \left(\sqrt{p}\ket{\langle \vec{a}_{i},\vec{s} \rangle}\ket{0}+ \sqrt{1-p}\ket{\eta_{\langle \vec{a}_{i},\vec{s} \rangle}}\ket{1} \right)\right)  \ ,  
 	\end{equation} 
 	 we have
 	$$
 	\dtr\left(\ket{\varphi},\ket{\varphi'} \right) = \frac{\negl(\lambda)}{q^{n/4}}\ .
 	$$
 \end{lemma}
 \begin{proof}
 	Recall that $\ket{\varphi}$ is obtained at the end of Step~\ref{state:fin} of Algorithm~\ref{algo:StMd}. In particular, thanks to Lemma~\ref{lemma:unitaryV}, we have
 	\begin{align*} 
 		\ket{\varphi} &= \left(\vec{I} \otimes \left( \vec{V}^{\dagger} \right)^{\otimes m} \right)  \vec{U}_{\mathcal{A}_{\textup{GE}}}\left( \frac{1}{\sqrt{q^{n}}}\sum_{\vec{s} \in \left( \mathbb{Z}/q\mathbb{Z}\right)^{n}}	\ket{\vec{s}} \Motimes_{i=1}^{m} \vec{V} \left(\ket{\psi_{\langle \vec{a}_{i},\vec{s} \rangle}}\ket{0}\right) \right)  \\
 		&=  \left( \vec{I} \otimes \left( \vec{V}^{\dagger} \right)^{\otimes m}  \right)   \vec{U}_{\mathcal{A}_{\textup{GE}}}\left( \frac{1}{\sqrt{q^{n}}}\sum_{\vec{s} \in \left( \mathbb{Z}/q\mathbb{Z}\right)^{n}}	\ket{\vec{s}} \Motimes_{i=1}^{m} \left( \sqrt{p}\ket{\langle \vec{a}_{i},\vec{s} \rangle}\ket{0}+ \sqrt{1-p}\ket{\eta_{\langle \vec{a}_{i},\vec{s} \rangle}}\ket{1} \right. \right. \\ 
 		& \hspace*{9.65cm} + \ket{\textup{error}_{\langle \vec{a}_{i},\vec{s} \rangle}} \Big) \Bigg) \ .
 	\end{align*} 
 		Taking the Hermitian product, we obtain:
 	\begin{align*}
 		\braket{\varphi'}{\varphi} &= \frac{1}{q^{n}} \sum_{\vec{s} \in \left( \mathbb{Z}/q\mathbb{Z} \right)^{n}} \prod_{i=1}^{m} \left( 1+ \sqrt{p} \bra{ \langle \vec{a}_{i},\vec{s} \rangle,0}\ket{\textup{error}_{\langle \vec{a}_i,\vec{s} \rangle}} + \sqrt{1-p}\bra{\eta_{\langle \vec{a}_{i},\vec{s} \rangle},1}\ket{\textup{error}_{\langle \vec{a}_{i},\vec{s} \rangle}} \right) \ . 
 	\end{align*}
 	As $\max_{j} \| \ket{\textup{error}_j} \| = \negl(\lambda)/\sqrt{q^{n}}$, we have that
 	$$
 	\braket{\varphi'}{\varphi} = q^{-n}  \sum_{\vec{s} \in ( \mathbb{Z}/q\mathbb{Z} )^{n}} \prod_{i= 1}^{m} ( 1+ z_{\vec{s},i}) \ ,
 	$$ 
 	for 
 	some~$z_{\vec{s},i} \in \mathbb{C}$ satisfying~$\max_{\vec{s},i} |z_{\vec{s},i}| \leq \negl (\lambda) /  \sqrt{q^{n}}$. Using the fact that $m \leq  \poly(\lambda)$, we obtain that~$\braket{\varphi'}{\varphi} = 1 - \negl(\lambda) / \sqrt{q^n}$.
 \end{proof}
 \begin{lemma}\label{lemma:W2} Using notations of Theorem \ref{theo:solveQlwe} and letting
 		\begin{equation}\label{eq:ketpsi'} 
 		\ket{\psi'} \eqdef \left( \vec{I} \otimes \left( \vec{V}^{\dagger} \right)^{\otimes m} \right)\left( \frac{1}{\sqrt{Z_f(\vec{A})}}\sum_{\vec{s}\in \left( \mathbb{Z}/q\mathbb{Z}\right)^{n}}\ket{\vec{0}} \Motimes_{i=1}^{m} \left(\sqrt{p}\ket{\langle \vec{a}_{i},\vec{s} \rangle}\ket{0}+ \sqrt{1-p}\ket{\eta_{\langle \vec{a}_{i},\vec{s}\rangle}}\ket{1} \right)\right) 
 	\end{equation}
 	we have
 	$$
 	\dtr\left( \ket{\psi'}, \ket{\vec{0}}\QLWE{m,n,q,f}{\vec{A}}  \ket{0} \right) \leq \sqrt{1 - \left( 1 - \sqrt{\frac{q^{n}}{Z_{f}(\vec{A})}}\; \negl(\lambda)\right)^{2}} \ .
 	$$ 
 \end{lemma}
 \begin{proof}
	By Definition \ref{def:m-LWEstates}, we have
		\begin{align*} 
		\ket{\vec{0}}\QLWE{m,n,q,f}{\vec{A}}  \ket{0} 
		&= \frac{1}{\sqrt{Z_f(\vec{A})}}\sum_{\vec{s}\in \left( \mathbb{Z}/q\mathbb{Z}\right)^{n}}\ket{\vec{0}} \Motimes_{i=1}^{m} \ket{\psi_{\langle \vec{a}_{i},\vec{s} \rangle}}\ket{0}  \\ 
		& = 
		 \left( \vec{I} \otimes \left( \vec{V}^{\dagger} \right)^{\otimes m} \right) \left( \frac{1}{\sqrt{Z_f(\vec{A})}}\sum_{\vec{s}\in \left( \mathbb{Z}/q\mathbb{Z}\right)^{n}}\ket{\vec{0}} \Motimes_{i=1}^{m}\vec{V} \left(\ket{\psi_{\langle \vec{a}_{i},\vec{s} \rangle}}\ket{0}\right) \right) \\
		&= \left( \vec{I} \otimes \left( \vec{V}^{\dagger} \right)^{\otimes m} \right) \left( \frac{1}{\sqrt{Z_f(\vec{A})}}\sum_{\vec{s}\in \left( \mathbb{Z}/q\mathbb{Z}\right)^{n}}\ket{\vec{0}} \Motimes_{i=1}^{m}\Big( \sqrt{p}\ket{\langle \vec{a}_{i},\vec{s} \rangle}\ket{0} \right.
		\\ & \hspace*{5.1cm}+ \sqrt{1-p}\ket{\eta_{\langle \vec{a}_{i},\vec{s} \rangle}}\ket{1} + \ket{\textup{error}_{\langle \vec{a}_{i},\vec{s} \rangle}} \Big) \Bigg) \ . 
	\end{align*} 
	Recall that $\max_{j} \| \ket{\textup{error}_{j}} \| = \negl(\lambda)/\sqrt{q^{n}}$ (see Lemma~\ref{lemma:unitaryV}). 	Therefore, we have
	$$
	\ket{\vec{0}}\QLWE{m,n,q,f}{\vec{A}}  \ket{0}^m = \ket{\psi'} + \left(\vec{I} \otimes \left( \vec{V}^{\dagger} \right)^{\otimes m}  \right) \left( \frac{1}{\sqrt{Z_{f}(\vec{A})}} \sum_{\vec{s} \in \left( \mathbb{Z}/q\mathbb{Z} \right)^{n}} \ket{\vec{0}}\ket{\textup{error}_{\vec{s}}} \right)\ ,
	$$
	for some~$\textup{error}_{\vec{s}}$ satisfying $\max_{\vec{s}} \| \ket{\textup{error}_{\vec{s}}} \|  \leq m \; \negl(\lambda)/\sqrt{q^{n}} \leq \negl(\lambda)/\sqrt{q^{n}}$, since  $m \leq \poly(\lambda)$. We hence obtain that
	$$
	\left\| \left( \vec{I} \otimes \left( \vec{V}^{\dagger} \right)^{\otimes m} \right) \left( \frac{1}{\sqrt{Z_{f}(\vec{A})}} \sum_{\vec{s} \in \left( \mathbb{Z}/q\mathbb{Z} \right)^{n}} \ket{\vec{0}}\ket{\textup{error}_{\vec{s}}} \right)  \right\| \leq \frac{q^{n}}{\sqrt{Z_{f}(\vec{A})}}\; \frac{\negl(\lambda)}{\sqrt{q^{n}}} = \sqrt{\frac{q^{n}}{Z_{f}(\vec{A})}}\negl(\lambda) \ , 
	$$
	which completes the proof. 
 \end{proof}
The following lemma will help us in analyzing the effect of the unitary~$\vec{U}_{\mathcal{A}_{\textup{GE}}}$. It considers its application on a state whose second and third registers contain a superposition of solved and undetermined linear equations. It is obtained from~\eqref{eq:unitaryUGE} by linearity
 and the fact that~$y_i$ in Equation~\eqref{eq:unitaryUGE} depends only in the last qubit.
 \begin{lemma}\label{lemma:UGE}
 	Let $\vec{U}_{\mathcal{A}_{\textup{GE}}}$ be defined as in Equation \eqref{eq:unitaryUGE}. Let $x_{1},\dots,x_{m} \in \mathbb{Z}/q\mathbb{Z}$ and $\ket{\eta_1},\dots,\ket{\eta_m}$ be some quantum states. We have
 	$$
 	\vec{U}_{\mathcal{A}_{\textup{GE}}}\left(\ket{\vec{s}} \Motimes_{i=1}^{m} \left( \sqrt{p}\ket{ x_i}\ket{0}+ \sqrt{1-p}\ket{\eta_{i}}\ket{1} \right) \right) \\
 	=  \sum_{\vec{y} \in \left\{ x_i, \bot\right\}^{m}} \ket{\vec{s}-\mathcal{A}_{\textup{GE}}(\vec{A},\vec{y})} \Motimes_{i=1}^{m}\lambda(y_i)\ket{\alpha_{y_i}^{\vec{s}}} \ , 
 	$$
 	where
 	$$
 	\ket{\alpha_{y_{i}}^{\vec{s}}} \eqdef 	\left\{ \begin{array}{ll}
 		\ket{ x_i}\ket{0} & \mbox{if } y_{i} = x_i \\
 		\ket{\eta_{i}}\ket{1}  & \mbox{otherwise}
 	\end{array} 
 	\right.
 	\quad \mbox{and} \quad 			
 	\lambda(y_{i}) \eqdef 	\left\{ \begin{array}{ll}
 		\sqrt{p} & \mbox{if } y_{i} = x_i \\
 		\sqrt{1-p} & \mbox{otherwise}
 	\end{array} 
 	\right. .
 	$$
 \end{lemma}

 We can now show the correctness of Algorithm~\ref{algo:StMd}, i.e., that Equation~\eqref{eq:dtr} holds.

  \begin{lemma}\label{lemma:dTR}
 	Using the notations of Theorem \ref{theo:solveQlwe}, we have, for a proportion $1 - \negl(\lambda)$ of matrices~$\vec{A} \in \left( \mathbb{Z}/q\mathbb{Z} \right)^{m \times n}$: 
 	$$
 	\dtr\left( \ket{\varphi},\ket{0}^{n \log q}\QLWE{m,n,q,f}{\vec{A}} \ket{0}^m \right) = \negl(\lambda) \ . 
 	$$
 \end{lemma}
 \begin{proof}
	 	First, by Condition~\ref{cdt:errorSuperposition} of Theorem~\ref{theo:solveQlwe}, we can build the quantum state 
	$
	\sum_{e \in \mathbb{Z}/q\mathbb{Z}} f(e)\ket{e} 
	$ up to a trace distance $\negl(\lambda)$. Therefore,  when analyzing the trace distance between the output~$\ket{\varphi}$ of Algorithm~\ref{algo:StMd} and~$\ket{\vec{0}}\QLWE{m,n,q,f}{\vec{A}} \ket{\vec{0}}$, we can assume that it is exactly~$\sum_{\vec{e} \in \left( \mathbb{Z}/q\mathbb{Z} \right)^{m}} \Motimes_{j=1}^{m} f(e)\ket{e_i}$ that  is built at Step~\ref{inst:error}. Indeed, this only affects the trace distance by an additive $m \negl(\lambda) = \negl(\lambda)$ term (recall that we have~$m \leq \poly(\lambda)$).

	By Lemmas~\ref{lemma:W1} and~\ref{lemma:W2}, and the triangular inequality over the trace distance, we have
 	\begin{align}
 		\dtr\Big( \ket{\varphi},\ket{0}^{n \log q}&\QLWE{m,n,q,f}{\vec{A}} \ket{0}^m \Big) \nonumber \\
 		&\leq 	\dtr\left(\ket{\varphi},\ket{\varphi'} \right) +  \dtr\left( \ket{\varphi'}, \ket{\psi'} \right)  + 	\dtr\left( \ket{\psi'}, \ket{0}^{n \log q}\QLWE{m,n,q,f}{\vec{A}} \ket{0}^m \right)  \nonumber   \\
 		&\leq  \dtr\left( \ket{\varphi'}, \ket{\psi'} \right) + \frac{\negl(\lambda)}{{q^{n/4}}}+ \sqrt{1- \left( 1- \frac{q^{n}}{Z_{f}(\vec{A})}\negl(\lambda)\right)^{2}} \ ,\label{eq:dtrphiphiPP} 
 	\end{align} 
 	where $\ket{\varphi'}$ and $\ket{\psi'}$ are respectively defined in Equations~\eqref{eq:varphiP} and~\eqref{eq:ketpsi'}. 
 	Applying the unitary~$\vec{I} \otimes \left( \vec{V}^{\dagger}\right)^{\otimes m}$ does not change the trace distance. 
 	Therefore, by using the definitions of~$\ket{\varphi'}$ and~$\ket{\psi'}$, we have
 	\begin{equation}\label{eq:dtrpsipsiIdeal}
 		\dtr\left( \ket{\varphi'},\ket{\psi'} \right)  = \dtr\left( \ket{\psi},\ket{\psi_{\textup{ideal}}} \right)
 	\end{equation}
 	where
 	$$
 	\ket{\psi} \eqdef \vec{U}_{\mathcal{A}_{\textup{GE}}}\left( \frac{1}{\sqrt{q^{n}}}\sum_{\vec{s} \in \left( \mathbb{Z}/q\mathbb{Z}\right)^{n}}	\ket{\vec{s}} \Motimes_{i=1}^{m} \left( \sqrt{p}\ket{\langle \vec{a}_{i},\vec{s} \rangle}\ket{0}+ \sqrt{1-p}\ket{\eta_{\langle \vec{a}_{i},\vec{s} \rangle}}\ket{1} \right) \right) 
 	$$
 	and 
 	$$
 	\ket{\psi_{\textup{ideal}}} \eqdef 	\frac{1}{\sqrt{Z_f(\vec{A})}}\sum_{\vec{s}\in \left( \mathbb{Z}/q\mathbb{Z}\right)^{n}}\ket{\vec{0}} \Motimes_{i=1}^{m} \left(\sqrt{p}\ket{\langle \vec{a}_{i},\vec{s} \rangle}\ket{0}+ \sqrt{1-p}\ket{\eta_{\langle \vec{a}_{i},\vec{s} \rangle}}\ket{1} \right) \ .
 	$$
 	By Lemma~\ref{lemma:UGE}, we have
 	\begin{align*}
 		\ket{\psi} 
 		&=  \frac{1}{\sqrt{q^{n}}}\sum_{\vec{s} \in \left( \mathbb{Z}/q\mathbb{Z}\right)^{n}}  \sum_{\vec{y} \in \left\{ \langle \vec{a}_{i},\vec{s} \rangle, \bot\right\}^{m}}	\ket{\vec{s}-\mathcal{A}_{\textup{GE}}(\vec{A},\vec{y}) } \Motimes_{i=1}^{m} \lambda(y_i)\ket{\alpha_{y_i}^{\vec{s}}} \ 
 	\end{align*}
 	 	where,
 	\begin{equation}\label{eq:lambdah} 
 		\ket{\alpha_{y_{i}}^{\vec{s}}} \eqdef 	\left\{ \begin{array}{ll}
 			\ket{\langle \vec{a}_{i},\vec{s} \rangle}\ket{0} & \mbox{if } y_{i} = \langle \vec{a}_{i},\vec{s}\rangle \\
 			\ket{\eta_{\langle \vec{a}_{i},\vec{s} \rangle}}\ket{1}  & \mbox{otherwise}
 		\end{array} 
 		\right.
 		\quad \mbox{and} \quad 			
 		\lambda(y_{i}) \eqdef 	\left\{ \begin{array}{ll}
 			\sqrt{p} & \mbox{if } y_{i} = \langle \vec{a}_{i},\vec{s}\rangle \\
 			\sqrt{1-p} & \mbox{otherwise}
 		\end{array} 
 		\right. .
 	\end{equation} 
 	Similarly, we have
 	\begin{align*}
 		\ket{\psi_{\textup{ideal}}} 
 		&=  \frac{1}{\sqrt{Z_f(\vec{A})}}\sum_{\vec{s} \in \left( \mathbb{Z}/q\mathbb{Z}\right)^{n}}  \sum_{\vec{y} \in \left\{ \langle \vec{a}_{i},\vec{s} \rangle, \bot\right\}^{m}}	\ket{\vec{0}} \Motimes_{i=1}^{m} \lambda(y_i)\ket{\alpha_{y_i}^{\vec{s}}} \ .
 	\end{align*}
 	We deduce that
 	\begin{multline*}
 		\braket{\psi_{\textup{ideal}}}{\psi} = \frac{1}{\sqrt{q^{n} \; Z_f(\vec{A})}} 
 		\sum_{\vec{s},\vec{s}' \in \left( \mathbb{Z}/q\mathbb{Z} \right)^{n}} \sum_{\vec{y} \in\left\{ \langle \vec{a}_{i},\vec{s} \rangle, \bot\right\}^{m}}\sum_{\vec{y}' \in\left\{ \langle \vec{a}_{i},\vec{s}' \rangle, \bot\right\}^{m}} \\ \underbrace{\braket{\vec{0}}{\vec{s}'- \mathcal{A}_{\textup{GE}}(\vec{A},\vec{y}')}\; \prod_{i=1}^{m}\lambda(y_i)\lambda(y'_i)\braket{\alpha_{y_i}^{\vec{s}}}{\alpha_{y'_i}^{\vec{s}'}}}_{\eqdef P_{\vec{s},\vec{s}',\vec{y},\vec{y}'}}\ .
 	\end{multline*} 
 	Our aim is to show that~$P_{\vec{s},\vec{s}',\vec{y},\vec{y}'}$ is always equal to $0$ except when~$\vec{s} = \vec{s}'$ and~$\vec{y} = \vec{y}'$. 
 	First, notice that~$P_{\vec{s},\vec{s}',\vec{y},\vec{y}'}$ can be non-zero only if the following holds
 	\begin{equation*}
 		\vec{s}' = \mathcal{A}_{\textup{GE}}(\vec{A},\vec{y}') \ .
 	\end{equation*} 
 	At this stage, recall that our Gaussian elimination algorithm~$\mathcal{A}_{\textup{GE}}$ is unambiguous: with the knowledge of $\vec{y}'$, it can only output~$\vec{s}'$ of~$\bot$ (but not output another vector). Further, to have~$P_{\vec{s},\vec{s}',\vec{y},\vec{y}'} \neq 0$, we also need
 	$$
 	\forall i: \braket{\alpha_{y_i}^{\vec{s}}}{\alpha_{y'_i}^{\vec{s}'}} \neq 0 \ . 
 	$$
 	Therefore, by definition of the $\ket{\alpha_{y_i}^{\vec{s}}}$'s in Equation~\eqref{eq:lambdah}, it is necessary that for all $i$, 
 	we have~$y_i = y'_i$. However, the $y_i'$'s uniquely determine $\vec{s}'$, therefore $\vec{s} = \vec{s}'$ in that case.  
 	Overall, we obtain that~$P_{\vec{s},\vec{s}',\vec{y},\vec{y}'} \neq 0$ implies that $\vec{s} = \vec{s}'$ and~$\vec{y}=\vec{y}'$. Therefore, we obtain
 	\begin{align}
 		\braket{\psi_{\textup{ideal}}}{\psi} &= \frac{1}{\sqrt{Z_f(\vec{A})\; q^{n}}}  \sum_{\vec{s} \in \left( \mathbb{Z}/q\mathbb{Z}\right)^{n}} \sum_{\substack{\vec{y} \in \left( \{\langle \vec{a}_j,\vec{s} \rangle, \;\bot \} \right)_{j=1}^{m}:\\ \vec{s} = \mathcal{A}_{\textup{GE}}\left(\vec{A},\vec{y}\right)}} \ \ \prod_{i=1}^{m} \lambda(y_i)^{2} \nonumber \\
 		&= \sqrt{\frac{q^{n}}{Z_f(\vec{A})}} \; p_{\mathcal{A}_{\textup{GE}}}(\vec{A}) \ ,  \label{eq:scalarProduct}
 	\end{align}
 	where $p_{\mathcal{A}_{\textup{GE}}}(\vec{A})$ is the success probability of $\mathcal{A}_{\textup{GE}}$ when each of its~$m$ equations as input is~$\bot$ with probability $1-p$ and $\langle \vec{a}_{i},\vec{s}\rangle$, with probability~$p$ (recall that~$p_{\mathcal{A}_{\textup{GE}}}(\vec{A})$ is independent from~$\vec{s}$). Now, by Condition~\ref{cdt:mpn} of Theorem~\ref{theo:solveQlwe}, we have~$m = (n+\log\log q)  /p \cdot \omega(\log \lambda)$.
 	Therefore, by assumption on algorithm~$\mathcal{A}_{\textup{GE}}$, except for a $\negl(\lambda)$-proportion of matrices~$\vec{A}$,  we have
 	$$
 	p_{\mathcal{A}_{\textup{GE}}}(\vec{A}) = 1 - \negl(\lambda) \ . 
 	$$
 	By using Equations~\eqref{eq:dtrphiphiPP}, \eqref{eq:dtrpsipsiIdeal} and~\eqref{eq:scalarProduct}, we deduce that for 
 	a proportion $1-\negl(\lambda)$ of matrices~$\vec{A}$, we have
 	\begin{multline*} 
 	\dtr\left( \ket{\varphi}, \ket{0}^{n \log q}\QLWE{m,n,q,f}{\vec{A}} \ket{0}^m \right) \leq \sqrt{1 - \frac{q^{n}}{Z_{f}(\vec{A})}(1-\negl(\lambda))^{2}} \\ +  \frac{\negl(\lambda)}{{q^{n/4}}} + \sqrt{1 - \left( 1 - \sqrt{\frac{q^{n}}{Z_{f}(\vec{A})}}\; \negl(\lambda)\right)^{2}} \ . 
 	\end{multline*} 
 	To complete the proof, it suffices to use Condition~\ref{cdt:Za} of Theorem~\ref{theo:solveQlwe}. 
 \end{proof}

\subsection{Run-time} \label{sec:runtime-lwe}
We now focus on the run-time of Algorithm~\ref{algo:StMd}. 
So far,  we did not specify how to compute the unitary~$\vec{V}$. This is the focus of the following lemma.

\begin{lemma}\label{lemma:cpxV}
	Using notations of Theorem \ref{theo:solveQlwe}, we can evaluate a unitary~$\vec{V}$ satisfying Equation \eqref{eq:unitaryV} 
	in time~$\poly(\lambda)$.  
\end{lemma}

\begin{proof} Our objective is to implement~$\vec{V}$ such that
	$$
	\forall x \in \mathbb{Z}/q\mathbb{Z}: \vec{V} \left(\ket{\chi_x}\ket{0} \right)= \ket{\chi_x} \left( u_x \ket{0} + \sqrt{1-|u_x|^{2}}\ket{1} \right) \ . 
	$$
	By Condition~\ref{cdt:eapx} of Theorem~\ref{theo:solveQlwe}, we can efficiently compute~$u_x = (\min |\widehat{f}|) / \widehat{f}(-x) + e_{\textup{apx}} (x) \in \mathbb{C}$ on $\poly(\lambda)$ bits with~$\max_x|e_{\textup{apx}} (x)| = \negl(\lambda)/\sqrt{q^n}$. Without loss of generality, we assume that~$u_x$ is written as its magnitude and phase $(m_x,\theta_x)$ where $m_{x}$ and $\theta_x$ have $b=\poly(\lambda)$ bits. As $x \mapsto u_{x}$ is computable in time~$\poly(\lambda)$, we can evaluate a unitary~$\vec{O}_{u}$ 
	satisfying the following, in quantum-time~$\poly(\lambda)$:
	$$
	\forall x: \vec{O}_{u} \left(\ket{x}\ket{0^{2b}}\right) =  \ket{x}\ket{m_x}\ket{\theta_x} \ . 
	$$
	Now, consider  the following two unitaries:
	\begin{align*}
		\vec{M} &\eqdef \sum\limits_{y \in \{0,1\}^b} \op{y}{y} \otimes \vec{I}_p \otimes \big(\widetilde{y} \ket{0} + \sqrt{1 - \widetilde{y}^2} \ket{1}\big)\bra{0}\ , \\
		\vec{\Theta} &\eqdef \sum\limits_{z \in \{0,1\}^b} \vec{I}_b \otimes \op{z}{z} \otimes \big(\mathrm{e}^{2\pi i \widetilde{z}} \op{0}{0} + \op{1}{1}\big) \ , 
	\end{align*}
	where $\widetilde{y} = \sum_{i=1}^{b} y_i/2^{i}$ and $\widetilde{z} = \sum_{i=1}^{b} z_i/2^{i}$.
	It can be checked that
	$$
	\vec{O}_u^\dagger\vec{\Theta}\vec{M} \vec{O}_u \left(\ket{x}\ket{0^{2b}}\ket{0}\right) = \ket{x} \ket{0^{2b}} (u_x \ket{0} + \sqrt{1 - |u_x|^2} \ket{1}) \ .
	$$
	The unitary~$\vec{M}$ can be implemented with~$O(b) = \poly(\lambda)$ unary and binary gates~\cite[Ch.~9, Exercise~7.a]{dewolf23}. Furthermore, we have
	$$
	\mathrm{e}^{2\pi i \widetilde{z}} = \prod_{k=1}^{b} \mathrm{e}^{2\pi i 2^{-k} z_k} \ .
	$$
	It shows that one only requires $b = \poly(\lambda)$ controlled gates to implement~$\vec{\Theta}$. This completes the proof.
\end{proof}

We are now ready to prove that we can run Algorithm~\ref{algo:StMd} in polynomial time. 
\begin{lemma}\label{lemma:runningTime}
	Using notations of Theorem~\ref{theo:solveQlwe},  Algorithm~\ref{algo:StMd} can be executed in time~$\poly(\lambda)$. 
\end{lemma}

\begin{proof}
	Step~\ref{inst:error} of Algorithm~\ref{algo:StMd} can be executed  in time $\poly(\lambda)$ by Condition~\ref{cdt:errorSuperposition} of Theorem~\ref{theo:solveQlwe}. All steps except Steps~\ref{inst:U}, \ref{inst:GE} and~\ref{state:fin} are readily seen to be computable in time~$\poly(\lambda)$ as~$m,\log q \leq  \poly(\lambda)$. By 
	Lemma~\ref{lemma:cpxV}, Steps~\ref{inst:U} and~\ref{state:fin} can be executed in time~$m \; \poly(\lambda) = \poly(\lambda)$. Finally, Step~\ref{inst:GE}  applies~$\vec{U}_{\textup{GE}}$. This unitary quantizes a $\poly(\lambda)$-time Gaussian elimination
	algorithm. 
\end{proof}

\section{$\qlwe$ for the Gaussian distribution and witness-oblivious $\lwe$ sampling}\label{subsec:Gauss}

Our aim in this section is to construct a witness-oblivious quantum~$\lwe_{m,n,q,|f|^{2}}$ sampler. For this purpose, we use Algorithm~\ref{algo:StMd} with a specific choice of parameter~$f$, to obtain the following theorem. 
The second part of the statement below is obtained by combining the first part and Theorem~\ref{theo:obliviousLWEsample}. 
This proves Theorem~\ref{th:main-intro}.

\begin{theorem}\label{theo:LWEsample}
	Let $m \geq n\geq 1$ and~$q \geq 3$ be integers and $\sigma \geq 2$ be a real number. The parameters~$m,n,q,\sigma$ are functions of the security parameter~$\lambda$ with~$m,\log q \leq \poly(\lambda)$ and~$q$ prime. Assume that the parameters satisfy the following conditions:
	$$
 m \geq n\sigma \cdot \omega (\log\lambda) \quad \mbox{ and } \quad  2 \leq \sigma \leq \frac{q}{ \sqrt{8m\ln q}} \ .
	$$
	Furthermore, let $f : \mathbb{Z}/q\mathbb{Z} \rightarrow \mathbb{C}$ be such that
	\begin{align*}
		f(x) \eqdef 
		\begin{cases}{}
			\sqrt{\vartheta_{\sigma, q}(x)} & \mbox{if } 0 \leq x \leq \frac{q}{2} \\
			- \sqrt{\vartheta_{\sigma, q}(x)} & \mbox{otherwise}
		\end{cases} \ .
	\end{align*}
	Then Algorithm~\ref{algo:StMd} runs in time~$\poly\left(\lambda\right)$
	 and, for a proportion $1-\negl(\lambda)$ of matrices~$\vec{A} \in \left( \mathbb{Z}/q\mathbb{Z}\right)^{m \times n}$, it outputs a quantum state $\ket{\varphi}$ such that~$\dtr( \ket{\varphi}, \ket{\vec{0}}\QLWE{m,n,q,f}{\vec{A}} \ket{0} ) =  \negl(\lambda)$.
	 
	 In particular, if $\lwe_{m,n,q,\sigma}$ is quantumly hard,  then there exists a $\poly(\lambda)$-time 
	 quantum witness-oblivious $\lwe_{m,n,q,\sigma}$ sampler. 
\end{theorem}

Note that Theorem~\ref{th:main-intro} puts some constraints on the arithmetic shape of the modulus~$q$, 
on the number of samples~$m$, and on the standard deviation parameter~$\sigma$. It would be convenient to allow 
smaller values of~$m$, arbitrary arithmetic shapes for~$q$ and superpolynomial values of~$\sigma$. Indeed, these are frequent
parametrizations of~$\lwe$. To reach such values, we can use randomized Karp reductions from $\lwe$ for some parameters 
to $\lwe$ for other parameters. For instance, we can use Theorem~\ref{theo:LWEsample} with many samples, and just throw away 
the superfluous ones. We may also use Theorem~\ref{theo:LWEsample} with some permitted parameters~$n,\sigma,q$ for which 
$\lwe$ is hard, and then perform modulus-switching or modulus-dimension switching~\cite{BLPRS13}. As an example, using 
modulus switching and throwing away superfluous samples, we obtain the following corollary.

\begin{corollary}\label{theo:main}
	Let $m \geq n\geq 1$ and~$q \geq 2$ be integers and $\sigma \geq 2$ be a real number. The parameters~$m,n,q,\sigma$ are functions of the security parameter~$\lambda$ with~$m,\sigma,\log q \leq \poly(\lambda)$. Assume that 
	$\lwe_{m',n,q',\sigma'}$ is hard, where~$q'\leq 2q$ is the smallest prime larger than~$q$, $\sigma' = 
	\sigma / (n+\lambda) \cdot \Omega_\lambda(1)$ and~$m' = \max(m,n\sigma' \cdot \omega(\log \lambda))$. If~$2 \leq \sigma' \leq q'/ \sqrt{8m'\ln q'}$, then there exists a witness-oblivious $\lwe_{m,n,q,\sigma}$ sampler. 
\end{corollary}

To prove Theorem~\ref{theo:LWEsample}, we show that the conditions of  Theorem~\ref{theo:solveQlwe} are fulfilled 
for the amplitude function of Theorem~\ref{theo:LWEsample}. This is the purpose of the rest of this section. 

\subsection{On Conditions~\ref{cdt:errorSuperposition} and~\ref{cdt:eapx} of Theorem~\ref{theo:solveQlwe}}

In the lemmas below, we show that~$f$ and~$\widehat{f}$ can be approximated with sufficient precision for Conditions~\ref{cdt:errorSuperposition} and~\ref{cdt:eapx} to apply.

\begin{lemma}\label{lemma:approximateGaussian}
	Let~$n \geq 1$, $q \geq 3$ integers, $\sigma > 0$  a real number and $f : \mathbb{Z}/q\mathbb{Z} \rightarrow \mathbb{C}$ as in Theorem~\ref{theo:LWEsample}. Assume that~$n, \sigma = \poly(\lambda)$ and $q=2^{\poly(\lambda)}$ is odd, where~$\lambda$ is security parameter. Then we can compute  $u_x = (\min |\widehat{f}|)/ \widehat{f}(-x) + e_{\textup{apx}}(x)$ on~$\poly(\lambda)$ bits  with $\max_{x} |e_{\textup{apx}}(x)| = \negl(\lambda)/\sqrt{q^{n}}$, in classical time $\poly(\lambda)$.
\end{lemma}
\begin{proof}
	We show how to approximate~$\widehat{f}(x)$ for every~$x$ within appropriate accuracy. This also suffices to approximate~$\min |\widehat{f}|$ because, by Lemma~\ref{lemma:odd-phase}, we have~$\min |\widehat{f}| = |\widehat{f}(0)|$. As seen in the proof of  Lemma~\ref{lemma:odd-phase}, we have, for all~$y \in \mathbb{Z}/q\mathbb{Z}$:
	\begin{align*}
		\widehat{f}(y) = \frac{f(0)}{\sqrt{q}}	+ i \; \frac{2}{\sqrt{q}} \sum\limits_{x \in \mathbb{Z} \cap (0,q/2)}f(x)\; \sin \frac{2\pi xy}{q} \ .
	\end{align*}
	First, note that one can efficiently approximate~$f(x)$ on~$\poly(\lambda)$ bits and within an absolute error~$\negl(\lambda)/\sqrt{q^{n}}$,  by relying on the Gaussian tail bound and summing~$\poly(\lambda)$ terms (as~$\sigma = \poly(\lambda)$).
	The quantities~$\sin (2\pi xy/q)$ can be similarly approximated, using the Taylor approximation of~$\sin$ up to 
	degree~$\poly(\lambda)$. To approximate~$\widehat{f}(y)$, we claim that it suffices to compute the summation above for the summands~$x \in \{1,2,\dots,  \poly(\lambda) \}$. We use the tail bound for the Gaussian distribution. Let~$C \eqdef \poly(\lambda) \frac{n}{2} \log q \leq \poly(\lambda)$. We have, for all~$y \in \mathbb{Z}/q\mathbb{Z}$:
	\begin{align*}
		\left| \sum\limits_{x \in \mathbb{Z} \cap (C,q/2)} \sqrt{\vartheta_{\sigma,q}(x)} \ \sin\frac{2\pi xy}{q} \right|
		&\leq \sum\limits_{x \in \mathbb{Z} \cap (C,q/2)} \sqrt{\vartheta_{\sigma,q}(x)}\\
		&= \frac{1}{\sqrt{\rho_{\sigma}(\mathbb{Z})}} \sum\limits_{x \in \mathbb{Z} \cap (C,q/2)} \sqrt{\sum\limits_{k \in \mathbb{Z}} \rho_{\sigma} (x+kq)}\\
		&\leq \frac{1}{\sqrt{\rho_{\sigma}(\mathbb{Z})}} \sum\limits_{x \in \mathbb{Z} \cap (C,q/2)} \sum\limits_{k \in \mathbb{Z}} \rho_{\sqrt{2}\sigma} (x+kq)\\
		&\leq \frac{1}{\sqrt{\rho_{\sigma}(\mathbb{Z})}} \sum\limits_{x \in \mathbb{Z} \setminus [-C,C]} \rho_{\sqrt{2}\sigma} (x)\\
		&\leq \frac{\rho_{\sqrt{2}\sigma}(\mathbb{Z})}{\sqrt{\rho_{\sigma}(\mathbb{Z})}} \ \frac{C}{\sqrt{2}\sigma} \ \sqrt{2\pi \mathrm{e}}   \ \mathrm{e}^{-\pi \frac{C^2}{2\sigma^2}}       \quad \text{(by Lemma~\ref{lemma:Ban-bound})} \\
		&\leq  \frac{1+\sqrt{2}\sigma}{\sqrt{\sigma}} \ \frac{C}{\sqrt{2}\sigma} \ \sqrt{2\pi \mathrm{e}} \ \mathrm{e}^{-\pi \frac{C^2}{2\sigma^2}}   \quad \text{(by Lemma~\ref{lemma:rho-Z-bound})}\\
		&\leq \negl(\lambda)/\sqrt{q^n} \ .
	\end{align*}
	Finally, we observe that the truncated summation can be computed in time~$\poly(\lambda)$.
\end{proof}

\begin{lemma} \label{cor:vartheta-encod}
	Let~$n \geq 1$, $q \geq 3$ integers, $\sigma > 0$  a real number and $f : \mathbb{Z}/q\mathbb{Z} \rightarrow \mathbb{C}$ as in Theorem~\ref{theo:LWEsample}. Assume that~$n,\sigma = \poly(\lambda)$ and $q=2^{\poly(\lambda)}$,
	 where~$\lambda$ is security parameter.
	Then we can build the following state in run-time~$\poly(\lambda)$ and within error~$\negl(\lambda) /\sqrt{q^n}$ in trace distance:
	$$
	\sum\limits_{x \in \mathbb{Z}/q\mathbb{Z}} f(x) \ket{x} \ .
	$$
\end{lemma}
\begin{proof}
	Let~$C=\poly(\lambda) \frac{n}{2} \log q \leq \poly(\lambda)$. First, we build a state proportional to:
	$$\sum_{x \in \mathbb{Z} \cap [-C,C]} \sqrt{\rho_{\sigma}(x)} \ket{x} \ .$$
	Thanks to~\cite{GR02}, such a state can be built in time~$\poly(\lambda)$. This state is within trace distance $\negl(\lambda)/\sqrt{q^n}$ (by using the same reasoning as in the proof of Lemma~\ref{lemma:approximateGaussian}) from 
	$$
	\sum_{x \in \mathbb{Z}/q\mathbb{Z}} \sqrt{\vartheta_{\sigma,q}} \ket{x}  \  . 
	$$
	To complete the proof, it remains to add a $-1$ phase to the states~$\ket{x}$ with~$x < 0$. This can be implemented by 
	using a control gate on the appropriate register of~$\ket{x}$.  	
\end{proof}

\subsection{On  Condition~\ref{cdt:mpn} of Theorem~\ref{theo:solveQlwe}}

We now want to show that~$q \cdot \min | \widehat{f}|^{2}$ is $1/\poly(\lambda)$.  We first observe that, in most cases, 
the direct choice of~$f_0= \sqrt{\vartheta_{\sigma,q}}$ does not satisfy this condition. This motivates the introduction of~$\pm 1$ phases.

 \begin{lemma} \label{lemma:noPhaseDG}
	Let~$q \geq 2$ and integer and~$\sigma\geq 1$ a real number. Let~$f_0 = \sqrt{\vartheta_{\sigma,q}}$. We have:
	$$
	q \cdot \min |\widehat{f_0}|^{2} \leq 32 \sigma \cdot \max\Big( \mathrm{e}^{- \frac{\pi \sigma^2}{4}} , \mathrm{e}^{- \frac{q^2}{4\sigma^2}} \Big) \enspace .
	$$
\end{lemma}
The proof
is deferred to Appendix~\ref{app:GD}.
The result shows that, for Condition~\ref{cdt:mpn} of Theorem~\ref{theo:solveQlwe} to have a chance to hold, one is required to set the standard deviation parameter~$\sigma$ as $O(\sqrt{\log \lambda})$ or such that~$q/\sigma = O(\sqrt{\log \lambda})$. Unfortunately, in the first case, the $\lwe_{m,n,q,\sigma}$ problem can be solved efficiently~\cite{AG11}, whereas the second one is too restrictive to enable cryptographic constructions. 

To circumvent the above difficulty, we consider phases. Note that adding phases to~$f$ does not have any impact on the measurements and, therefore, after measuring the state, one still obtains an~$\lwe$ sample with the same distribution. In the following lemmas, we show that the  phases considered in Theorem~\ref{theo:LWEsample} can sufficiently increase the quantity~$q \cdot \min | \widehat{f}|^{2}$.

\begin{lemma} \label{lemma:odd-phase}
	Let~$q\geq 2$  an odd
integer and  $f : \mathbb{Z}/q\mathbb{Z} \rightarrow \mathbb{R}$   such that~$f(-x) = - f(x)$ for all~$x \in \mathbb{Z}/q\mathbb{Z} \setminus\{0 \}$.
	Then we have
	$$q \cdot \min |\widehat{f}|^{2} = q \cdot |\widehat{f}(0)| = |f(0)|^2 \ .$$
\end{lemma}
\begin{proof}
	The discrete Fourier transform of~$f$ is given by
	\begin{align*}
		\widehat{f}(y) &= \frac{f(0)}{\sqrt{q}}	+ \frac{1}{\sqrt{q}} \sum_{x \in \mathbb{Z} \cap (0,q/2)} f(x) \; \omega_{q}^{xy}  + \frac{1}{\sqrt{q}}\sum_{x \in  \mathbb{Z} \cap (-q/2 ,0)}  f(x) \; \omega_{q}^{xy} \\
		&=   \frac{f(0)}{\sqrt{q}}	+  \frac{1}{\sqrt{q}}\sum_{x \in \mathbb{Z} \cap (0,q/2)} f(x) \left( \omega_{q}^{xy} - \omega_{q}^{-xy} \right) \quad \left(\mbox{as $ \forall x \neq 0: f(-x) = -f(x)$}\right) \\
		&=  \frac{f(0)}{\sqrt{q}}	+ i \; \frac{2}{\sqrt{q}} \sum_{x \in \mathbb{Z} \cap (0,q/2)}f(x)\; \sin \frac{2\pi xy}{q} \enspace ,
	\end{align*}
	for all~$y \in \mathbb{Z}/q\mathbb{Z}$.  Since~$f$ is a real-valued function, the quantity~$|\widehat{f}(y)|$ is no smaller than~$|f(0)/\sqrt{q}|$ and the lower bound is reached at~$y=0$. 
	\end{proof}

We have the following lemma as a special case for the distribution~$\vartheta_{\sigma,q}$.
\begin{lemma} \label{lemma:phase}
	Let~$q\geq 2$  an odd integer, $\sigma >0$ a real number and $f : \mathbb{Z}/q\mathbb{Z} \rightarrow \mathbb{C}$ as in Theorem~\ref{theo:LWEsample}.
	Then we have
	$$
	q \cdot \min |\widehat{f}|^{2} \geq \frac{1}{1+\sigma} \enspace .
	$$
\end{lemma} 
\begin{proof}
	The statement~$f(-x) = - f(x)$ holds for all~$x \neq 0$.  
	Therefore, using the positivity of~$\vartheta_{\sigma,q}$, we obtain
	\[
		q \cdot \min |\widehat{f}|^{2} \geq \vartheta_{\sigma,q}(0) \geq \frac{1}{\rho_{\sigma}(\mathbb{Z})} \ .
	\]
	Lemma~\ref{lemma:rho-Z-bound} then gives the result.
\end{proof}

  Adding $\pm 1$ phases ``exponentially'' increases the success probability $p = q \cdot \min |\widehat{f}|^{2}$, when choosing $|f|^{2} = \vartheta_{\sigma,q}$, which allows to fulfill Condition~\ref{cdt:mpn} of Theorem~\ref{theo:obliviousLWEsample} under the constraint that~$m,\sigma \leq  \poly(\lambda)$. This improvement is crucial as otherwise we could not set~$m$ (which plays a significant role in the run-time of the algorithm) as some~$\poly(\lambda)$.

\subsection{On  Condition~\ref{cdt:Za} of Theorem~\ref{theo:solveQlwe}}
\label{sse:z}

To instantiate Theorem~\ref{theo:solveQlwe}, it now suffices to show that Condition~\ref{cdt:Za} holds. Recall that it involves~$Z_f(\vec{A})$, which is the normalization scalar ensuring that~$\QLWE{m,n,q,f}{\vec{A}}$ is unit vector.

\begin{lemma}\label{lemma:Z-upperbound}
	Let~$m,n\geq 1,q\geq 2$  integers,~$\vec{A} \in (\mathbb{Z}/q\mathbb{Z})^{m \times n}$,~$f$  an amplitude function over~$\mathbb{Z}/q\mathbb{Z}$, and~$Z_f(\vec{A})$  as per Definition~\ref{def:m-LWEstates}. Then we have:
	$$
	\left| \frac{Z_f(\vec{A})}{q^{n}} - 1 \right| \  \leq \ \sum\limits_{\substack{\vec{e\neq e}'\\ \vec{e} - \vec{e}' \in \Im(\vec{A})}} |f|(\vec{e})  \cdot |f|(\vec{e}') \enspace .
	$$
\end{lemma}
\begin{proof}
	For every vector~$\vec{e} \in (\mathbb{Z}/q\mathbb{Z})^m$, let~$\ket{\Im(\vec{A})+\vec{e}}$ denotes the following state:
	$$
	\ket{\Im(\vec{A})+\vec{e}} \eqdef \sum\limits_{\vec{x} \in (\mathbb{Z}/q\mathbb{Z})^n} \ket{\vec{Ax}+\vec{e}}\ .
	$$
	(The state is purposefully not normalized.)
	For two vectors~$\vec{e,e}'$, we have
	\begin{equation}\label{eq:hyperplan-prod-blue}
		\ip{\Im(\vec{A})+\vec{e}'}{\Im(\vec{A})+\vec{e}} =  
		\begin{cases} q^n & \text{if } \vec{e}-\vec{e}' \in \Im(\vec{A}) \\ 
			0 & \text{otherwise} 
			\end{cases} \ .
	\end{equation}
	Then the~$\QLWE{m,n,q,f}{\vec{A}}$ state can be expressed as follows:
	\begin{align*}
		\QLWE{m,n,q,f}{\vec{A}} = \frac{1}{\sqrt{Z_f(\vec{A})}} \sum\limits_{\vec{e} \in (\mathbb{Z}/q\mathbb{Z})^m} f(\vec{e}) \ket{\Im(\vec{A}) + \vec{e}} \ .
	\end{align*}
	Therefore, we have
	\begin{align*}
		Z_f(\vec{A}) &= \Big\| \sum\limits_{\vec{e} \in (\mathbb{Z}/q\mathbb{Z})^m} f(\vec{e}) \ket{\Im(\vec{A}) + \vec{e}} \Big\|^2 \ .
	\end{align*}
	The above term is equal to:
	
	\[
		\sum\limits_{\vec{e,e}'} f(\vec{e})  \overline{f(\vec{e}')} \ip{\Im(\vec{A})+\vec{e}'}{\Im(\vec{A})+\vec{e}}
		=  q^n \hspace*{-.2cm}\sum\limits_{\substack{\vec{e},\vec{e}'\\ \vec{e} - \vec{e}' \in \Im(\vec{A})}} \hspace*{-.2cm} f(\vec{e})  \overline{f(\vec{e}')} 
		= q^{n}+ q^{n}\hspace*{-.2cm} \sum\limits_{\substack{\vec{e}\neq\vec{e}'\\ \vec{e} - \vec{e}' \in \Im(\vec{A})}}\hspace*{-.2cm} f(\vec{e})  \overline{f(\vec{e}')} \ , 
	\]
	where we used Equation~(\ref{eq:hyperplan-prod-blue}). We obtain:	
	\[
		\left|\frac{Z_{f}(\vec{A})}{q^{n}} -1 \right| = \Big| \sum\limits_{\substack{\vec{e}\neq\vec{e}'\\ \vec{e} - \vec{e}' \in \Im(\vec{A})}} f(\vec{e}) \overline{f(\vec{e}')} \; \Big| \ .
		\]
		The result follows from the triangular inequality.
\end{proof}

We now prove the following lemma.
\begin{lemma} \label{lemma:DRT}
	Let~$m,n\geq 1$ and~$q\geq 2$  integers, and~$f$  an amplitude function over~$\mathbb{Z}/q\mathbb{Z}$. Assume that~$q$ is prime.
	Let~$\vec{A}$ be sampled uniformly in~$(\mathbb{Z}/q\mathbb{Z})^{m \times n}$, and let~$Z_f(\vec{A})$  be as per Definition~\ref{def:m-LWEstates}. Then we have, for any~$\delta > 0$:
	$$
	\Pr_{\vec{A}}\left( \left| \frac{Z_f(\vec{A})}{q^{n}} - 1 \right| \geq \delta \right) \leq \frac{\sum_{\vec{e} \neq \vec{e}'}  |f|(\vec{e})\cdot |f|(\vec{e}') }{\delta \cdot q^{m-n}}  \ .
	$$
\end{lemma}

\begin{proof}
	We define:
	$$S \eqdef \sum_{\substack{\vec{e} \neq \vec{e}' \\ \vec{e}-\vec{e}' \in \Im(\vec{A})}} |f|(\vec{e}) \cdot |f|(\vec{e}') \enspace ,$$
	and view it as a random variable over the random choice of~$\vec{A}$.
	By Lemma~\ref{lemma:Z-upperbound}, we have that $| Z_{f}(\vec{A})/q^{n} -1 |   \leq S$ holds for all~$\vec{A}$.
	Further, by Markov's inequality, one obtains
that $\mathbb{P}_{\vec{A}}( S \geq \delta )
		\leq  \Ex_\vec{A}(S)  / \delta$ holds
	for every~$\delta > 0$.
	Using the linearity of the expectation, one obtains:
	\begin{align}
		\Ex_\vec{A}(S) 
		&= \Ex_\vec{A} \Bigg( \sum\limits_{\vec{e} \neq \vec{e}'} \mathbbm{1}_{\Im(\vec{A})}(\vec{e} - \vec{e}') \ |f|(\vec{e}) \cdot |f|(\vec{e}') \Bigg) \nonumber  \\
		&=  \sum\limits_{\vec{e} \neq \vec{e}'} \Pr_\vec{A} \left(\vec{e} - \vec{e}' \in \Im(\vec{A}) \right) \  |f|(\vec{e}) \cdot |f|(\vec{e}') \nonumber \\
		&\leq  \frac{1}{q^{m-n}} \sum\limits_{\vec{e} \neq \vec{e}'}  |f|(\vec{e}) \cdot |f|(\vec{e}')  \ . \label{eq:to_be_changed}
	\end{align}
The last inequality follows from the union bound (over all elements in the image of~$\vec{A}$) and the fact that~$q$ is prime.
\end{proof}

We are particularly interested in the case where~$|f|=\sqrt{\vartheta_{\sigma,q}}$. The following lemma allows us to apply the above result on this particular function.
\begin{lemma}\label{lemma:explicit-ZA-Z-blue}
	Let~$m  \geq 1$ and~$q\geq 2$  integers, and~$\sigma$  a real number such that~$2 \leq \sigma \leq q/ \sqrt{8m\ln q}$.
	Then we have:
	\begin{align*}
		\sum_{\vec{e} \neq \vec{e}'}  \sqrt{\vartheta_{\sigma,q}(\vec{e})}\sqrt{\vartheta_{\sigma,q}(\vec{e}')} \leq q^{\frac{m}{2}}+1 \ .
	\end{align*}
\end{lemma}
\begin{proof}
	First, note that the summation can be rewritten in the following way:
	\begin{align*}
		\sum_{\vec{e} \neq \vec{e}'}  \sqrt{\vartheta_{\sigma,q}(\vec{e})}\sqrt{\vartheta_{\sigma,q}(\vec{e}')} = \Big( \sum_{\vec{e}}  \sqrt{\vartheta_{\sigma,q}(\vec{e})} \ \Big)^2 - \sum_{\vec{e}}  \vartheta_{\sigma,q}(\vec{e}) \ .
	\end{align*}
	By positivity of the second term, it suffices to find an upper bound for the first one. We rely on Lemma~\ref{lemma:vartheta-rho} 
	to approximate~$\vartheta_{\sigma,q}$ with~$D_{\mathbb{Z}^m,\sigma}$. We have
	\begin{align}
		\sum_{\vec{e} \in \mathbb{Z}^m \cap (-\frac{q}{2},\frac{q}{2}]^m}  \sqrt{\vartheta_{\sigma,q}(\vec{e})} &\leq \sum_{\vec{e} \in \mathbb{Z}^m \cap (-\frac{q}{2},\frac{q}{2}]^m} \left( \sqrt{D_{\mathbb{Z}^m,\sigma}(\vec{x})} + \mathrm{e}^{-\frac{q^2}{8\sigma^2}}\right) \quad \text{(by Lemma~\ref{lemma:vartheta-rho})}\nonumber\\
		&= \sum_{\vec{e} \in \mathbb{Z}^m \cap (-\frac{q}{2},\frac{q}{2}]^m} \left( \frac{\rho_{\sqrt{2}\sigma}(\mathbb{Z}^m)}{\sqrt{\rho_{\sigma}(\mathbb{Z}^m)}} \ D_{\mathbb{Z}^m,\sqrt{2}\sigma}(\vec{x})  + \mathrm{e}^{-\frac{q^2}{8\sigma^2}}\right) \nonumber\\
		&\leq \frac{\rho_{\sqrt{2}\sigma}(\mathbb{Z}^m)}{\sqrt{\rho_{\sigma}(\mathbb{Z}^m)}} + q^m \mathrm{e}^{-\frac{q^2}{8\sigma^2}} \nonumber\\
		&\leq \frac{(1+\sqrt{2}\sigma)^m}{\sqrt{\sigma}^m} + q^m \mathrm{e}^{-\frac{q^2}{8\sigma^2}} \quad \text{(by Lemma~\ref{lemma:rho-Z-bound})}\nonumber\\
		&\leq (2\sqrt{\sigma})^m + q^m  \mathrm{e}^{ - \frac{q^2}{8\sigma^2}} \  .\nonumber 
	\end{align}
	Since~$\sigma \leq q/ \sqrt{8m\ln q}$, we have that the last term is~$\leq 1$. Finally, note that the same upper bound on~$\sigma$ also implies that~$2\sqrt{\sigma} \leq \sqrt{q}$. 
\end{proof}

We can now conclude, by combining Lemma~\ref{lemma:DRT} and Lemma~\ref{lemma:explicit-ZA-Z-blue} with~$\delta = q^{-n}$.

\begin{lemma} \label{lemma:Z-for-Zp}
	Let~$m\geq n \geq 1$, $q \geq 2$ integers, $\sigma>0$  a real number and $f : \mathbb{Z}/q\mathbb{Z} \rightarrow \mathbb{C}$ as in Theorem~\ref{theo:LWEsample}. Assume that~$q$ is prime  and~$2 \leq \sigma \leq q / \sqrt{8m\ln q}$.
	Let~$\vec{A}$ be sampled uniformly from~$(\mathbb{Z}/q\mathbb{Z})^{m \times n}$, and let~$Z_f(\vec{A})$ as per Definition~\ref{def:m-LWEstates}. Then we have
	\begin{align*}
		\Pr_{\vec{A}}\left( \left| \frac{Z_f(\vec{A})}{q^{n}} - 1 \right| \geq q^{-n}\right) \ \leq \ q^{2n-m}(q^{\frac{m}{2}}+1) \ .
	\end{align*}
\end{lemma}

\section{On the security of some lattice-based SNARKs} \label{sec:snark}

The purpose of this section is to show that the hardness assumptions used in several standard model lattice-based SNARKs~\cite{GMNO18,NYI20,ISW21,SSEK22,CKKK23,GNS23} 
are invalid in the context of quantum adversaries.

\subsection{Module Learning With Errors}

All the SNARK constructions mentioned above can be framed into an algebraic variant of $\lwe$ called~$\mlwe$, which captures
$\lwe$ and the Ring Learning With Errors problem ($\rlwe$)~\cite{SSTX09,LPR10}. 
To recall the definition of~$\mlwe$ and adapt the results on oblivious $\lwe$ sampling to~$\mlwe$, we first provide some reminders. 

Let~$d \geq 1$ be a power-of-2 integer. The cyclotomic ring~$R$ of degree~$d$ is~$\mathbb{Z}[x]/\langle x^d+1 \rangle$.
 Each element of~$R$ is a polynomial of degree at most~$d-1$ with integer coefficients. We let~$\phi:R \rightarrow \mathbb{Z}^d$ 
 denote the map that sends each element~$\sum_{i<d} a_ix^i \in R$ to the vector~$(a_0,\dots,a_{d-1})^{\intercal} \in \mathbb{Z}^d$.  For every element~$a \in R$, we define~$\rot(a)$ as the matrix whose~$i$-th column is~$\phi(x^{i-1} a \bmod{x^d+1})$, for all~$1 \leq i \leq d$. Then we have~$\phi(a \cdot b)=\rot(a) \phi(b)$ for all~$a,b \in R$. Let~$q \geq 2$ be an integer. 
 Both~$\phi$ and~$\rot$ are extended to the quotient ring~$R/qR$. Similarly, we extend~$\phi$ to~$(R/qR)^m$ and~$\rot$ to~$(R/qR)^{m\times n}$ for any integers~$m,n\geq 1$.  

For a distribution~$\chi$ over~$\mathbb{Z}/q\mathbb{Z}$, we define~$\chi^{\otimes d}$ as the distribution over~$R/qR$ 
obtained by independently sampling each coefficient from~$\chi$. The notation is extended to distributions over~$(R/qR)^m$ for any~$m \geq 1$. 

Module Learning With Errors ($\mlwe$)  is a variant of~$\lwe$ introduced and studied in~\cite{BGV12,LS15}. It is  defined 
by replacing~$\mathbb{Z}/q\mathbb{Z}$ by~$R/qR$ in the $\lwe$ definition. 
	 
\begin{definition}[$\mlwe$] \label{def:prelim-mlwe}
	 \label{def:mlwe} Let~$m \geq n \geq 1, q\geq 2$ be integers,  $R$ be a cyclotomic ring of degree a power-of-2 integer~$d$ and~$\chi$ be a distribution over~$\mathbb{Z}/q\mathbb{Z}$.
The parameters~$m,n,d,q$ and~$\chi$ are functions of some security parameter~$\lambda$.  Let~$\vec{A} \in (R/qR)^{m \times n}$, $\vec{s} \in (R/qR)^{n}$ be sampled uniformly and~$\vec{e} \in (R/qR)^{m}$ be sampled from from~$\chi^{\otimes dm}$. 
	 The search~$\mlwe_{m,n,d,q,\chi}$ problem is to find~$\vec{s}$ and~$\vec{e}$ given the pair~$(\vec{A},\vec{A}\vec{s}+\vec{e})$. The vectors~$\vec{s}$ and~$\vec{e}$ are respectively called the secret and the noise.

	  	 Whenever~$\chi$ is equal to the folded discrete Gaussian distribution~$\vartheta_{\sigma,q}$ for some~$\sigma>0$, we overwrite the notations as~$\mlwe_{m,n,d,q,\sigma}$ .
\end{definition}

We now show how Theorem~\ref{theo:LWEsample} can be extended to~$\mlwe$. 
The~$\mlwe$ problem can be viewed as a special case of~$\lwe$. Concretely, an $\mlwe_{m,n,d,q,\chi}$ instance~$(\vec{A},\vec{b} = \vec{A}\vec{s}+\vec{e}) \in (R/qR)^{m\times n} \times (R/qR)^m$ is mapped to the $\lwe_{md,nd,q,\chi}$ 
instance
\[ 
(\rot(\vec{A}),\phi(\vec{b}) = 
\rot(\vec{A})\phi(\vec{s})+\phi(\vec{e})) \ \in \ (\mathbb{Z}/q\mathbb{Z})^{md\times nd} \times (\mathbb{Z}/q\mathbb{Z})^{md} \ .
\] 
Our goal is to use Theorem~\ref{theo:solveQlwe} with these specific matrices. By the identity above, one can observe that Conditions~\ref{cdt:errorSuperposition} and~\ref{cdt:eapx} are not impacted by the change from~$\lwe$ to~$\mlwe$. Condition~\ref{cdt:mpn} is related to the  Gaussian elimination subroutine of Algorithm~\ref{algo:StMd}.
We note that there is no~$q$ such that~$R/qR$ 
is a field  for~$d > 2$ 
(as opposed to $\mathbb{Z}/q\mathbb{Z}$ with~$q$). Instead, we choose~$q$ prime such that~$q = 3 \bmod 8$. 
In that case, the ring~$R/qR$ is isomorphic to~$\mathbb{F}_{q^{d/2}} \times \mathbb{F}_{q^{d/2}}$. 
For~$m \geq n \cdot \omega(\log \lambda)$, a uniform~$\vec{A} \in (R/qR)^{m \times n}$ has a set of~$n$ rows 
that form an invertible matrix, with 
probability~$1- \negl(\lambda)$. This allows us to adapt the Gaussian elimination subroutine of Algorithm~\ref{algo:StMd} to the 
module setting  when~$d > 2$. Overall, for such a modulus~$q$, Condition~\ref{cdt:mpn} is also  not impacted by the change from~$\lwe$ to~$\mlwe$.

We now focus on Condition~\ref{cdt:Za}, which was proved in Subsection~\ref{sse:z} to be fulfilled in the~$\lwe$ case for a specific choice of amplitude function (defined in Theorem~\ref{theo:LWEsample}). We keep the same amplitude function, and adapt 
Lemma~\ref{lemma:Z-for-Zp} to the module setting. 

\begin{lemma} \label{lemma:Z-for-Rp}
Let~$m \geq n \geq 1, q\geq 2$ integers,   $\sigma>0$  a real number, $f : \mathbb{Z}/q\mathbb{Z} \rightarrow \mathbb{C}$ as in Theorem~\ref{theo:LWEsample} and~$R$ a cyclotomic ring of degree a power-of-2 integer~$d$. Assume that~$d>2$,~$q$ is prime and satisfies~$q=3 \bmod 8$, and~$2 \leq \sigma \leq \sqrt{q/(8m\ln q)}$. Let~$\vec{A}$ be sampled uniformly from~$(R/qR)^{m \times n}$, and let~$Z_f(\rot(\vec{A}))$ as per Definition~\ref{def:m-LWEstates}. Then we have
	\begin{align*}
		\Pr_{\vec{A}}\left( \left| \frac{Z_f(\rot(\vec{A}))}{q^{nd}} - 1 \right| \geq q^{-nd}\right) \ \leq \ q^{(2n-\frac{m}{2})d}(q^{\frac{md}{4}}+1) \ .
	\end{align*}
\end{lemma}

\begin{proof}
We follow the proof of Lemma~\ref{lemma:Z-for-Zp} in Subsection~\ref{sse:z}. Lemma~\ref{lemma:Z-upperbound}  applies without any change. For Lemma~\ref{lemma:DRT}, the only step that needs to be adapted is Equation~\eqref{eq:to_be_changed}. We have, for~$\vec{e} \neq \vec{e}' \in (R/qR)^m$:
\[
\Pr_\vec{A} \left(\vec{e} - \vec{e}' \in \Im(\vec{A})\right) 
\leq q^{dn}  \max_{\vec{s} \in (R/qR)^n} \Pr_\vec{A} \left(\vec{e} - \vec{e}' = \vec{A} \vec{s} \right) \  \leq  q^{dn} \cdot q^{-\frac{dm}{2}} \ . 
\]
where we used the union bound in the first inequality and considered only one of the components of
$R/qR \simeq \mathbb{F}_{q^{d/2}} \times \mathbb{F}_{q^{d/2}}$ in the second inequality. As a result, the term~``$q^{m-n}$'' 
in statement of Lemma~\ref{lemma:DRT} is replaced by~$q^{(m/2-n)d}$. 
The proof of Lemma~\ref{lemma:explicit-ZA-Z-blue} is unchanged, but we strengthen the upper bound 
on~$\sigma$ to~$\sigma \leq \sqrt{q/(8m\ln q)}$ to be able to replace the term``$q^{m/2}$'' 
in statement of Lemma~\ref{lemma:explicit-ZA-Z-blue}  by~$q^{md/4}$. This completes the proof of Lemma~\ref{lemma:Z-for-Rp}.
\end{proof}

Using the above, we obtain the following adaptation of  Theorem~\ref{theo:LWEsample}.

\begin{theorem}\label{theo:MLWEsample}
 Let~$m,n,d,q,R,\sigma,\lambda$ as in Definition~\ref{def:prelim-mlwe}. \color{black}
Assume that~$m,\log q \leq \poly(\lambda)$, $d>2$, and~$q$ is prime with~$q = 3 \bmod 8$. Assume further that the parameters satisfy the following conditions:
	$$
 m \geq n\sigma \cdot \omega (\log\lambda) \quad \mbox{ and } \quad  2 \leq \sigma \leq \sqrt{\frac{q}{8m\ln q}} \ .
	$$
	Then Algorithm~\ref{algo:StMd} runs in time~$\poly\left(\lambda\right)$
	 and, for a proportion $1-\negl(\lambda)$ of matrices~$\vec{A} \in \left( R/qR\right)^{m \times n}$, it outputs a quantum state $\ket{\varphi}$ such that~$\dtr( \ket{\varphi}, \ket{\vec{0}}\QLWE{q,f}{\rot(\vec{A})} \ket{0} ) =  \negl(\lambda)$.
	 
	 In particular, if $\mlwe_{m,n,d,q,\sigma}$ is hard,  then there exists a $\poly(\lambda)$-time 
	 quantum witness-oblivious $\mlwe_{m,n,d,q,\sigma}$ sampler. 
\end{theorem}

As in the $\lwe$ context, we could use Karp reductions from $\mlwe$ with one parametrization to $\mlwe$ with another parametrization
to significantly extend the range of allowed $\mlwe$ parametrizations in Theorem~\ref{theo:MLWEsample}. We could notably 
throw away superfluous samples, switch from one modulus to another~\cite{LS15} or trade modulus for dimension~\cite{AD17}.

\subsection{Knapsack~$\mlwe$}
\label{sse:kmlwe}
We generalize the  knapsack variant of~$\lwe$ from~\cite{MM11} to modules. 
\begin{definition}[$\kmlwe$]
\label{def:knapsackmlwe}
Let~$m,n,d,q,R,\chi$  as in Definition~\ref{def:prelim-mlwe}.  Let~$\vec{B} \in (R/qR)^{ n \times m}$ and~$\vec{e} \in (R/qR)^{m}$ be sampled from~$\chi^{\otimes dm}$. 
The search~$\kmlwe_{m,n,d,q,\chi}$ problem is to find~$\vec{e}$ from~$(\vec{B}, \vec{B}\vec{e})$.

   Whenever~$\chi$ is equal to the folded discrete Gaussian distribution~$\vartheta_{\sigma,q}$, we overwrite the notation as~$\kmlwe_{m,n,d,q,\sigma}$.
\end{definition}

A Karp reduction from~$\lwe$ to its knapsack form was given in~\cite[Le.~4.8]{MM11}. To extend it to modules, one needs to 
be able to perform linear algebra efficiently and that uniform matrices over~$(R/qR)^{m \times (m-n)}$ contain a subset of~$m-n$ rows that 
is invertible with sufficiently high probability. These conditions were already required to obtain Theorem~\ref{theo:MLWEsample}, so we can keep the same parameter constraints here. Using  Theorem~\ref{theo:MLWEsample} and Lemma~\ref{lemma:bb-reduction}, we obtain the 
following result.

\begin{theorem}\label{th:prime-kmlwe}
Let~$m,n,d,q,R,\sigma,\lambda$  as in Definition~\ref{def:prelim-mlwe}. \color{black}
Assume that~$m,\log q \leq \poly(\lambda)$, $d>2$, and~$q$ is prime with~$q = 3 \bmod 8$. Assume further that the parameters satisfy the following conditions:
	$$
 m \geq (m-n)\sigma \cdot \omega (\log\lambda) \quad \mbox{ and } \quad  2 \leq \sigma \leq \sqrt{\frac{q}{8m\ln q}} \ .
	$$
Assume that~$\mlwe_{m,m-n,d,q,\sigma}$
 is hard. 
Then there exists a $\poly(\lambda)$-time algorithm that produces samples~$(\vec{B},\vec{B}\vec{e})$ that 
are within statistical distance~$\negl(\lambda)$ from those obtained by sampling~$\vec{B}$ uniformly and~$\vec{e}$ 
from~$\vartheta_{\sigma,q}^{\otimes dm}$,
and for which there exists no efficient extractor algorithm that would recover the witness~$\vec{e}$. 
\end{theorem}

We discuss some restrictions of Theorem~\ref{th:prime-kmlwe}. For a large number of columns~$m$, the matrix~$\vec{B}$ must 
be almost  square for the condition~$m \geq (m-n) \cdot \omega(\log \lambda)$ to be satisfied. If we are interested in much fewer rows than columns (which is the case in our applications), one may use Theorem~\ref{th:prime-kmlwe} with a near-square matrix and then throw away the superfluous rows. This preserves obliviousness. 

Another restriction of Theorem~\ref{th:prime-kmlwe} is that the modulus~$q$ is required to be prime and 
to satisfy~$q=3\bmod 8$. However, in most applications, the modulus 
is not of that form: for example, in~\cite{ISW21}, the considered moduli are powers~of~$2$. We want to 
use modulus switching for~$\kmlwe$, but there are two difficulties. First, modulus switching introduces a small rounding error. We make it part of the weight vector~$\vec{e}$ by putting the matrix~$\vec{B}$ in canonical form. Second, in our application, the matrix~$\vec{B}$ 
is given as input to the instance sampler rather than generated by the sampler itself. For this aspect, we note that the sampler of Theorem~\ref{th:prime-kmlwe} satisfies this property: given as input a uniform matrix~$\vec{B}$, with 
probability~$1-\negl(\lambda)$, it outputs~$\vec{B}\vec{e}$ such that~$\vec{e}$ is within~$\negl(\lambda)$ statistical distance 
from~$\vartheta_{\sigma,q}^{\otimes dm}$.

Algorithm~\ref{algo:kmlwe} is designed to handle those aspects. Step~\ref{step:2-1} puts the input matrix in canonical form. Step~\ref{step:2-4} performs a modulus switch for the non-trivial component~$\overline{\vec{B}}$ of the canonical form. The new modulus~$q'$ is prime and satisfies~$q' = 3 \bmod 8$ (note that such primes are frequent). By choice of~$\vec{E}$ and~$\tau$, the resulting matrix~$\overline{\vec{B}}'$ is within negligible statistical distance from uniform (this may be proved using standard facts on discrete Gaussian distributions, such as done for example in~\cite{BLPRS13}). Step~\ref{step:2-6} randomizes to hide the canonical form to obtain a uniform matrix~$\overline{\vec{B}}$. Step~\ref{step:2-7} calls
the algorithm  from Theorem~\ref{th:prime-kmlwe} to obtain~$\overline{\vec{b}} = \overline{\vec{B}} \vec{e}$ for some 
unknown~$\vec{e}$ (as discussed above, the algorithm from Theorem~\ref{th:prime-kmlwe} satisfies the property that it outputs a vector for a given matrix, rather than sampling them together). Finally, Step~\ref{step:2-8} sends~$\overline{\vec{b}}$ back to~$R/qR$. 

We can see that the output~$\vec{b}$ is of the correct form. 
First, note that we have~$\overline{\vec{T}}^{-1} \overline{\vec{b}} = ( \vec{I} \ | \ \overline{\vec{B}}' ) \vec{e}$. Rounding from modulus~$q'$ to modulus~$q$ gives~$\frac{q}{q'}( \vec{I} \ | \ \overline{\vec{B}}' ) \vec{e} + \vec{f}$  for some small-magnitude vector~$\vec{f}$. Letting~$\vec{e}_1$ denote the first~$n$ entries of~$\vec{e}$ and~$\vec{e}_2$ the remaining~$m-n$, and using the definition of~$\overline{\vec{B}}'$, we see that $\frac{q}{q'}( \vec{I} \ | \ \overline{\vec{B}}' ) \vec{e} + \vec{f}$ is of the form~$\overline{\vec{B}} \vec{e}_2 + \vec{g}$ for some small magnitude vector~$\vec{g}$. This can be rewritten 
as~$( \vec{I} \ | \ \overline{\vec{B}}) (\vec{g}_{\phantom{2}}^\intercal | \vec{e}_2^\intercal)^\intercal$. Multiplying by~$\vec{T}$ gives that~$\vec{b}$ is indeed of the correct form. Further, the transformation preserves obliviousness. Assume by contradiction that an extractor can recover~$(\vec{g}_{\phantom{2}}^\intercal | \vec{e}_2^\intercal)^\intercal$. Then it can in particular recover~$\vec{e}_2$. From~$\vec{e}_2$, it can recover~$\vec{e}_1$ as~$\vec{e}_1 = \overline{\vec{T}}^{-1} \overline{\vec{b}} - \overline{\vec{B}}' \vec{e}_2$. This contradicts the fact that the algorithm from Theorem~\ref{th:prime-kmlwe} is oblivious.

Finally, let us comment on the failure probability of Step~\ref{step:2-1} (as~$q'$ is prime and satisfies~$q' = 3 \bmod 8$, the failure probability of 
Step~\ref{step:2-5} is very low). Depending on the arithmetic shape of~$q$ and the values of $m$ and~$n$, this value could possibly be 
non-negligible. Fortunately, in all the applications, the number of columns~$m$ is orders of magnitude higher than the number 
of rows~$n$, so that, with overwhelming probability, we can find a subset of $n$ columns that is invertible. It then suffices to apply Algorithm~\ref{algo:kmlwe} after an appropriate reordering of the columns.

\begin{algorithm}[h!]
	\caption{Witness-oblivious $\kmlwe$ sampler for arbitrary~$q$}\label{algo:kmlwe} 
	\begin{flushleft}
		{\bf Parameters:} $m,n,q,d,\sigma$ and~$\lambda$ as in Definition~\ref{def:prelim-mlwe}. 
         \end{flushleft}
	\begin{algorithmic}[1]
		\Input  $\vec{B} \in \left( R/qR\right)^{n \times m}$.
		\Output A vector~$\vec{b} \in  (R/qR)^{n}$. 
		\vspace{0.2cm}
		\State\label{step:2-1} Compute a matrix~$\vec{T}$ such that~$\vec{T}\vec{B} = ( \vec{I} \ | \ \overline{\vec{B}})$. If~$\vec{T}$ does not exist, then abort.	
		\State \label{inst:set-q}Set~$q'$ as the smallest prime larger than~$q$ such that~$q' = 3 \bmod{8}$. 
		\State Set~$\tau \eqdef q'/q \cdot \sqrt{\lambda}$. 
		\State\label{step:2-4} Set~$\overline{\vec{B}}'\eqdef \frac{q'}{q} \overline{\vec{B}} + \vec{E}$ with each entry of $\rot(\vec{E})$ sampled  from~$D_{\mathbb{Z}^d-\frac{q'}{q} \rot(\overline{\vec{B}}_{ij}),\tau}$  for all~$i,j$. \State\label{step:2-5} Sample~$\overline{\vec{T}} \in (R/q' R)^{n \times n}$. If it is not invertible, then abort. 
		\State\label{step:2-6} Set~$\overline{\vec{B}} \eqdef  \overline{\vec{T}} ( \vec{I} \ | \ \overline{\vec{B}}' )$.
		\State\label{step:2-7} Apply the sampler from Theorem~\ref{th:prime-kmlwe} with parameters $m,n,q',d,\sigma$ on~$\overline{\vec{B}}$  
		to obtain~$\overline{\vec{b}} = \overline{\vec{B}} \vec{e}$.
		\State\label{step:2-8} Compute~$\vec{b} \eqdef  \vec{T}^{-1} \lfloor \frac{q}{q'} (\overline{\vec{T}}^{-1} \overline{\vec{b}}  \bmod q') \rceil \bmod q$.
		\State Return~$\vec{b}$.
	\end{algorithmic} 
\end{algorithm}

\subsection{SNARKs from linear-only vector encryption}  \label{subsec:love}

 For constructing SNARKs, the authors of~\cite{ISW21,SSEK22,CKKK23} adapt the approaches of~\cite{BCI13} and~\cite{BISW17}
to the $\lwe$ setting (the possibility of adaptation to~$\lwe$ was actually suggested in~\cite{BCI13}, see Remark~5.19 therein). 
They use secret-key vector encryption schemes that are linear-only homomorphic. The plaintexts  belong to an~$R/pR$-module, whereas the ciphertexts belong to an~$R/qR$-module for some integers~$q > p \geq 2$ where~$R = \mathbb{Z}[x]/\langle x^d + 1\rangle$ for some power-of-2 degree~$d$. Such schemes allow the players to compute~$R/pR$-linear functions of the ciphertexts but no other function than those ones. This is called to the linear-only property.

\begin{definition} [Vector Encryption over Cyclotomic Fields] \label{def:VE}
Let~$\ell,m,n \geq 1$ be integers,~$R=\mathbb{Z}[x]/\langle x^d+1 \rangle$ with a power-of-2 degree $d$, $ q > p \geq 2$ be  integers, and~$S$ be a subset of~$(R/pR)^m$. All these are functions of the security parameter~$\lambda$. 
A secret-key linearly-homomorphic vector encryption scheme with the message space~$(R/pR)^\ell$ and the ciphertext space~$(R/qR)^n$ is a tuple of algorithms~$\Pi_\mathsf{Enc}=(\mathsf{Gen}, \mathsf{Enc}, \mathsf{Dec}, \mathsf{Add})$ with the following specifications.
\begin{itemize}\setlength{\itemsep}{5pt}
\item $\mathsf{Gen}(1^\lambda) \mapsto (\mathsf{pp},\mathsf{sk})$: Given the security parameter~$\lambda$, it outputs public parameters~$\mathsf{pp}$ and a secret key~$\mathsf{sk}$;
\item $\mathsf{Enc}(\mathsf{sk}, \vec{v}) \mapsto \mathsf{ct}$: Given the secret key~$\mathsf{sk}$ and a vector~$\vec{v} \in (R/pR)^\ell$, it outputs a ciphertext~$\vec{ct} \in (R/qR)^n$;
\item $\mathsf{Dec}(\mathsf{sk}, \vec{ct}) \mapsto \vec{v}\slash \bot$: Given the secret key~$\mathsf{sk}$ and a ciphertext~$\vec{ct}$, it outputs a vector~$\vec{v} \in (R/pR)^\ell$ or a special symbol~$\bot$;
\item $\mathsf{Add}(\mathsf{pp}, \{\vec{ct}_i\}_i, \{y_i\}_i) \mapsto \vec{ct}^*$: Given the public parameters~$\mathsf{pp}$, a collection of ciphertexts~$\{\vec{ct}_i\}_i$, and a collection of scalars~$\{y_i\}_i$ from~$R/pR$, it outputs a ciphertext~$\vec{ct}^*$.
\end{itemize}

Moreover, Algorithm~$\mathsf{Add}$ satisfies the following property:
\begin{itemize}
	\item Additive homomorphism with respect to the set~$S$: For all security parameters~$\lambda$, all vectors~$\{\vec{v}_1,\dots,\vec{v}_m\}$ from~$(R/pR)^\ell$, and~$(y_1,\dots,y_m) \in S$, it holds that
	\begin{align*}
\Pr\left(\mathsf{Dec}(\mathsf{sk}, \vec{ct}^*) = \sum_{i=1}^m y_i\vec{v}_i \ \Bigg| \ 
\begin{array}{l}
(\mathsf{pp},\mathsf{sk}) \leftarrow \mathsf{Gen}(1^\lambda)\\
\vec{ct}_i \leftarrow \mathsf{Enc}(sk, \vec{v}_i)\\
\vec{ct}^* \leftarrow \mathsf{Add}(\mathsf{pp}, \{\vec{ct}_i\}_i, \{y_i\}_i) 
\end{array} 
\right)
= 1- \negl(\lambda) \ .
\end{align*}
\end{itemize}
\end{definition}

The set~$S$ controls the level of homomorphic operations that are allowed. In~\cite[Th.~3.12]{ISW21}, it is  showed that the proposed vector encryption scheme allows homomorphic operations with respect to the whole set~$(R/pR)^m$ when~$q$ is chosen sufficiently large. In~\cite[Th.~2]{SSEK22}, this set is more restricted.

When using lattice problems, the functionality of Definition~\ref{def:VE} is obtained as follows. One typically relies on an 
$\lwe/\mlwe$ encryption scheme with plaintexts defined modulo~$p$. Given ciphertexts~$\vec{ct}_i$'s, which are vectors modulo~$q$, 
and scalars~$y_i$, the $\mathsf{Add}$ algorithm first computes the linear combination~$\sum_i y_i \vec{ct}_i$ and then possibly adds 
some large amount of noise (a technique typically referred to as noise flooding or noise smudging) or rounds. These operations can 
be publicly implemented. In terms of security, the ciphertexts~$\vec{ct}_i$ are designed to be computationally indistinguishable
 from uniform, under an appropriate $\lwe/\mlwe$ parametrization, to ensure the IND-CPA security of the vector encryption scheme.  
 We note that all the schemes we consider follow this blueprint.

In this work, we are particularly interested in the linear-only security property. Note that the adversary is allowed to be a quantum algorithm in the context of post-quantum cryptography. This is taken into account in the following definition.

\begin{definition}[Linear-Only Against Quantum Adversaries] \label{def:linear-only}
A vector encryption scheme~$\Pi_\mathsf{Enc}=(\mathsf{Gen}, \mathsf{Enc}, \mathsf{Dec}, \mathsf{Add})$ is linear-only if for all QPT algorithms~$\mathcal{A}$, there exists a valid QPT extractor~$\extractor$ such that for all security parameters~$\lambda$, auxiliary mixed states~$\rho$ over~$\mathbb{C}^{2^{\poly(\lambda)}}$, and any QPT plaintext generator~$\mathcal{M}$, it holds that
$$\Pr\left( \mathsf{ExptLinearExt}_{\Pi_{\mathsf{Enc}},\mathcal{A},\mathcal{M},\extractor,\rho}(1^\lambda) = 1\right) = 
\negl(\lambda) \ ,$$
where the experiment~$\mathsf{ExptLinearExt}_{\Pi_{\mathsf{Enc}},\mathcal{A},\mathcal{M},\extractor,\rho}(1^\lambda)$ is defined as follows.
\begin{enumerate}[label=\textcolor{blue}{\arabic*.}, ref=\arabic*]\setlength{\itemsep}{5pt}
\item The challenger samples the public parameters and the secret key~$(\mathsf{pp},\mathsf{sk}) \leftarrow \mathsf{Gen}(1^\lambda)$, together with~$m$ vectors~$(\vec{v}_1,\dots,\vec{v}_m) \leftarrow \mathcal{M}(1^\lambda, \mathsf{pp})$. It computes the ciphertexts~$\vec{ct}_i \leftarrow \mathsf{Enc}(\mathsf{sk}, \vec{v}_i)$ for all~$i$. 
\item \label{def:linear-only-cond-2} Then it runs the extraction process with the outputs as follows:
$$
\big((\vec{ct}_1',\dots,\vec{ct}_k'), \vec{\Pi}\big) \leftarrow \langle \mathcal{A}, \extractor \rangle(1^\lambda,\ket{\mathsf{pp},\vec{ct}_1,\dots,\vec{ct}_m} \otimes \rho,\ket{0}) \ .
$$
Let~$\vec{V}' = (\vec{v}_1 |  \dots  |  \vec{v}_m) \vec{\Pi}$. The output of the experiment is~$1$ if there exists an~$i \leq m$ such that~$\mathsf{Dec}(\mathsf{sk},\vec{ct}_i) \neq \bot$ and~$\mathsf{Dec}(\mathsf{sk},\vec{ct}_i) \neq \vec{v}_i'$ where~$\vec{v}_i'$ is the~$i$-th column of~$\vec{V}'$. Otherwise, the experiment outputs~$0$.
\end{enumerate}
\end{definition}

As discussed in~\cite[Rem.~3.6]{ISW21}, the requirement that the extractor must succeed for all auxiliary inputs~$\rho$ is too strong. In particular, no polynomial-time extractor exists if~$\rho$ is the output of a one-way function that the extractor must invert in order to analyze the behaviour of the sampler. In all cases that we consider, in the classical setting, the auxiliary inputs are sampled as uniform strings. In the quantum setting, such a string can be simulated by Hadamard gates and projective measurements. Therefore, in our applications, we choose~$\rho$ to be null.

In Definition~\ref{def:linear-only}, the adversary is given~$m$ ciphertexts $\vec{C} \eqdef (\vec{ct}_1 | \dots | \vec{ct}_m) \in (R/qR)^{n \times m}$ and is supposed to output~$k$ distinct small linear combinations of these, namely~$\vec{C}\pi_1,\dots,\vec{C}\pi_k$ where~$\pi_i \in R^{m}$ is the~$i$-th column of~$\vec{\Pi}$, with each entry in~$(-p/2,p/2]$. It asks the extractor to find the exact value of the matrix~$\vec{\Pi}$. 

We observe that~$(\vec{C},\vec{C}\pi_i)$ is a~$\kmlwe$ instance, for all~$i$. 
In~\cite{ISW21}, the authors use~$\mlwe$ with~$d=2$, whereas much larger degrees are considered in~\cite{SSEK22,CKKK23}. In all cases, the $\kmlwe$ number of columns is very large, of the order of~$2^{20}$, whereas the number of rows corresponds to $\mlwe$-based ciphertexts is of the order of~$2^{12}$. The plaintext modulus~$p$ has a bit-size that is  much smaller than the one of the ciphertext modulus~$q$. The latter may have up to~$100$ bits in~\cite{ISW21}. 

We attack the linear-only property as follows. Let~$\vec{C} = (\vec{ct}_1 | \dots | \vec{ct}_m) \in (R/qR)^{n \times m}$ with~$\vec{ct}_i = \mathsf{Enc}(\mathsf{sk},\vec{v}_i)$ for all~$i$. Consider 
a quantum $\kmlwe$ sampler as in Subsection~\ref{sse:kmlwe}. Note that~$\vec{C}$ is not statistically uniform but only computationally indistinguishable from uniform. This assumption holds for all secret-key encryption schemes used in the considered SNARK constructions. We claim that the sampler is still oblivious in this situation: if an extractor exists for $\vec{C}$'s 
of this form, then we can distinguish~$\vec{C}$ from uniform (note that one can efficiently verify the validity of the extracted witness). 
Now, let~$\vec{Ce}$ be the output of the sampler, with~$\vec{e}=(e_1,\dots,e_m)^\intercal$ small. It then holds that
$e_1\vec{ct}_1 + \dots +e_m \vec{ct}_m$ decrypts to
\[e_1\vec{v}_1 + \dots +e_m \vec{v}_m =  \vec{Ve} \bmod p \ , 
\] 
as~$\Pi_\mathsf{Enc}$ is additively-homomorphic modulo~$p$. Since~$\vec{C}\vec{e}$ is a hard instance sampled obliviously, extracting~$\vec{e}$ out of~$\vec{Ce}$ is not possible for QPT extractors, except with negligible probability. This contradicts Condition~\ref{def:linear-only-cond-2} of Definition~\ref{def:linear-only}.

\subsection{SNARKs from encoding schemes}
\label{sse:othersnarks}

In~\cite{GGPR13}, specific encoding schemes were introduced to build SNARKs from assumptions related to the discrete logarithm 
problem. Later, the framework was applied to lattices for constructing presumably post-quantum SNARKs~\cite{GMNO18,NYI20,GNS23}. The constructions in~\cite{GMNO18,NYI20} consider encodings for finite fields, while encodings for rings of the form~$R/pR$ are designed. 
Concretely, the message space is of the form~$R/pR$ for some integer~$p$ and ring of integers~$R$ of a number field, and the codeword space is~$(R/qR)^n$ for some integers~$n,q > p$. The ring~$R$ is usually chosen to be the ring of integers of a power-of-2 cyclotomic field. We recall the definition of encoding schemes, keeping only the parameters and properties that are relevant for our purposes.

\begin{definition} [Encoding Schemes Over Cyclotomic Rings]\label{def:encod-synx}
Let~$m,n \geq 1$ be integers,~$R=\mathbb{Z}[x]/\langle x^d+1 \rangle$ with a power-of-2 degree $d$, and~$ q > p \geq 2$ be integers. All these are functions of the security parameter~$\lambda$. 
An $m$-linearly-homomorphic encoding scheme with the message space~$R/pR$ and the codeword space~$C \subseteq (R/qR)^{n}$ is a tuple of algorithms~$\Pi_{\mathsf{Ecd}}=(\mathsf{Gen},\mathsf{Encode}, \mathsf{Eval})$ with the following specifications.
\begin{itemize}\setlength{\itemsep}{5pt}
\item $\mathsf{Gen}(1^\lambda) \mapsto (\mathsf{pp},\mathsf{sk})$: Given the security parameter~$\lambda$, it outputs  public parameters~$\mathsf{pp}$ and a secret key~$\mathsf{sk}$.
\item $\mathsf{Encode}(\mathsf{sk}, a) \mapsto \vec{cw}$: Given the secret key~$\mathsf{sk}$ and a ring element~$a \in R/pR$, it outputs a codeword~$\vec{cw} \in C$ with the following property: the subsets~$\{C_a \ | \ a \in R/pR\}$ partition~$C$ where~$C_a$ is the set of all possible encodings of~$a$.
\item $\mathsf{Eval}(\mathsf{pp}, \{\vec{cw}_1,\dots, \vec{cw}_m\}, \{c_1,\dots,c_m\}) \mapsto \vec{cw}^*$: Given the public parameters~$\mathsf{pp}$,~$m$ codewords~$\{\vec{cw}_1,\dots, \vec{cw}_m\}$, and~$m$ scalars~$\{c_1,\dots,c_m\}$ in~$R/pR$, it outputs a codeword~$\vec{cw}^*$.
\end{itemize}

Moreover, Algorithm~$\mathsf{Eval}$ satisfies the following property:
\begin{itemize}
	\item $m$-linearly homomorphism: For all~$a_1,\dots,a_m, c_1,\dots,c_m \in (R/pR)^m$, it holds that
	\begin{align*}
\Pr\left(\vec{cw}^* \in C_{\langle \vec{a},\vec{c}\rangle} \ \Bigg| \ 
\begin{array}{l}
(\mathsf{pp},\mathsf{sk}) \leftarrow \mathsf{Gen}(1^\lambda)\\
\vec{cw}_i \leftarrow \mathsf{Encode}(sk, a_i)\\
\vec{cw}^* \leftarrow \mathsf{Eval}(\mathsf{pp}, \{\vec{cw}_i\}_i, \{c_i\}_i) 
\end{array} 
\right)
= 1- \negl(\lambda) \ .
\end{align*}
\end{itemize}
\end{definition}

Algorithm~$\mathsf{Eval}$ operates within the same framework as Algorithm~$\mathsf{Add}$ of Definition~\ref{def:VE}. 
Moreover, the encodings are such that the codewords are computationally indistinguishable from random elements in the codeword space.

The~$m$-power knowledge of exponent assumption ($m$-PKE) is a generalization of the knowledge of exponent assumption by~\cite{Dam91} to  encoding schemes. We adapt this assumption to the quantum setting.

\begin{definition}[$m$-PKE Against Quantum Adversaries] \label{def:m-PKE}
An encoding scheme~$\Pi_{\mathsf{Ecd}}=(\mathsf{Gen}, \allowbreak \mathsf{Encode}, \allowbreak \mathsf{Eval})$ satisfies~$m$-PKE assumption for the auxiliary input generator~$\mathcal{Z}$ if for all QPT algorithms~$\mathcal{A}$, there exists a valid QPT extractor~$\extractor$ such that
$$\Pr\left( \mathsf{ExptKnowledgeExt}_{\Pi_{\mathsf{Ecd}},\mathcal{A},\mathcal{Z},\extractor,k}(1^\lambda) = 1\right) = \negl(\lambda) \ ,$$
where the experiment~$\mathsf{ExptLinearExt}_{\Pi_{\mathsf{Ecd}},\mathcal{A},\mathcal{Z},\extractor,k}(1^\lambda)$ is defined as follows.
\begin{enumerate}[label=\textcolor{blue}{\arabic*.}, ref=\arabic*]\setlength{\itemsep}{5pt}
\item The challenger samples the public parameters and the secret key~$(\mathsf{pp},\mathsf{sk}) \leftarrow \mathsf{Gen}(1^\lambda)$, together with~$\alpha$ and~$s$ sampled uniformly from~$(R/pR)^{\times}$ and a fixed subset of~$(R/pR)^{\times}$, respectively. It computes~$\sigma$ as follows:
\begin{align*} \sigma \eqdef \big(\mathsf{pk},&\mathsf{Encode}(\mathsf{sk},1),\mathsf{Encode}(\mathsf{sk},s),\dots, \mathsf{Encode}(\mathsf{sk},s^m), \\ 
& \mathsf{Encode}(\mathsf{sk},\alpha),\mathsf{Encode}(\mathsf{sk},\alpha s),\dots, \mathsf{Encode}(\mathsf{sk},\alpha s^m)\big) \ .
\end{align*}
It also computes~$z \leftarrow \mathcal{Z}(\sigma)$.
\item \label{eq:m-PKE}Then it runs the extraction process with the outputs, as follows:
$$\big((\vec{cw},\vec{cw}'), (a_0,\dots,a_m)\big) \leftarrow \langle \mathcal{A}, \extractor \rangle(1^\lambda,\ket{\sigma,z},\ket{0}) \ .$$
The output of the experiment is~$1$ if~$\vec{cw}' - \alpha \vec{cw} \in C_{0}$ and~$\vec{cw} \not\in C_S$ where~$S = \sum_{i=0}^m a_i s^i$. Otherwise, the output of the experiment is~$0$.
\end{enumerate}
\end{definition}

In~\cite{GMNO18,NYI20,GNS23}, it is assumed that~$\mathcal{Z}$ is ``benign'', in the sense that the auxiliary information~$z$ is generated with a dependency on~$\mathsf{sk}$, $s$ and~$\alpha$ that is limited to the extent that it can be generated efficiently from~$\sigma$. The extractor is also given the randomness of the adversary. In the quantum setting, we allow the extractor to have auxiliary inputs of the above type, while we omit the randomness of the adversary since it can be simulated by Hadamard gates and projective measurements. 

In~\cite{GMNO18,NYI20}, the authors use~$\lwe$ symmetric encryption (\ie{} with~$d=1$) for the encoding scheme. 
The value of~$m$ is of order~$2^{15}$, which is significantly larger than the rank of the ciphertext space~$n$ (chosen around~$2^{10}$). The ciphertext modulus~$q$ can be as large as~$736$ bits, whereas the plaintext modulus~$p$ has 32 bits. The authors of~\cite{GNS23} rely on high-degree~$\mlwe$, whereas the (module) rank of their ciphertext space is constant.

In Definition~\ref{def:m-PKE}, the adversary is given~$2(m+1)$ encodings of the powers of~$s$. We use them to 
define the following matrix:$$ \vec{C} \eqdef 
\left(
\begin{array}{c} 
\mathsf{Encode}(\mathsf{sk},1) \\
\mathsf{Encode}(\mathsf{sk},\alpha)
\end{array} 
\ \Big| \ 
\begin{array}{c} 
\mathsf{Encode}(\mathsf{sk},s) \\
\mathsf{Encode}(\mathsf{sk},\alpha s)
\end{array}
\ \Big| 
\ \cdots 
\ \Big| \ 
\begin{array}{c}
\mathsf{Encode}(\mathsf{sk},s^m) \\ 
\mathsf{Encode}(\mathsf{sk},\alpha s^m) \\ 
\end{array}
\right)  \in (R/qR)^{2n \times (m+1)}
\ .$$
A small combination~$\vec{Ce}$ of the columns gives a pair of ciphertext~$(\vec{cw}, \vec{cw}')$ that 
satisfies~$\vec{cw} - \alpha \vec{cw}' \in C_0$, by the $(m+1)$-linear homomorphism property of the scheme.
The vectors~$\vec{cw}$ and~$\vec{cw}'$ respectively correspond to the first and second halves of~$\vec{C e}$.  Note that the auxiliary input~$z$ does not help to recover~$\vec{e}$. It contains codewords of the form~$\mathsf{Encode}(\mathsf{sk}, \beta v(s))$ where~$v$ is a publicly known polynomial and~$\beta$ is a uniformly sampled element from~$R/pR$ that is independent from all other parameters. This does not help the adversary to extract information about the matrix~$\vec{C}$, as can be shown using a hybrid argument in which one replaces the plaintext~$\beta v(s)$ with a garbage plaintext (using the fact that the codewords are indistinguishable from uniform). 
By adapting the arguments from Subsection~\ref{subsec:love}, it can be seen that sampling~$\vec{C e}$ obliviously allows to break the security assumption of 
Definition~\ref{def:m-PKE}.

\section*{Acknowledgments.}
The authors are grateful to Dan Boneh,  Andr\'e Chailloux, Omar Fawzi, Alex Grilo, Yuval Ishai, Amit Sahai, Jean-Pierre Tillich and David Wu for insightful discussions. The authors 
particularly thank Jean-Pierre Tillich and André Chailloux for the reference to the POVM from~\cite{CB98} and discussions pertaining to its implementation.  The work of Thomas Debris-Alazard was funded by the French Agence Nationale de la Recherche through ANR JCJC COLA (ANR-21-CE39-0011). The authors were supported by  the PEPR
quantique France 2030 programme (ANR-22-PETQ-0008).

\newpage

\bibliographystyle{alpha}

\begin{thebibliography}{GMNO18}

\bibitem[ACL{\etalchar{+}}22]{ACLMT22}
Martin~R. Albrecht, Valerio Cini, Russell W.~F. Lai, Giulio Malavolta, and Sri
  Aravinda~Krishnan Thyagarajan.
\newblock Lattice-based {SNARKs}: Publicly verifiable, preprocessing, and
  recursively composable.
\newblock In {\em {CRYPTO}}, 2022.

\bibitem[AD17]{AD17}
Martin~R. Albrecht and Amit Deo.
\newblock Large modulus ring-{LWE} {\(\geq\)} module-{LWE}.
\newblock In {\em {ASIACRYPT}}, 2017.

\bibitem[AFLN23]{AFLN23}
Martin~R. Albrecht, Giacomo Fenzi, Oleksandra Lapiha, and Ngoc~Khanh Nguyen.
\newblock {SLAP}: Succinct lattice-based polynomial commitments from standard
  assumptions.
\newblock 2023.
\newblock Available at \url{https://eprint.iacr.org/2023/1469}.

\bibitem[AG11]{AG11}
Sanjeev Arora and Rong Ge.
\newblock New algorithms for learning in presence of errors.
\newblock In {\em {ICALP}}, 2011.

\bibitem[Ban93]{Ban93}
Wojciech Banaszczyk.
\newblock New bounds in some transference theorems in the geometry of numbers.
\newblock {\em Math. Ann.}, 1993.

\bibitem[BCI{\etalchar{+}}13]{BCI13}
Nir Bitansky, Alessandro Chiesa, Yuval Ishai, Rafail Ostrovsky, and Omer
  Paneth.
\newblock Succinct non-interactive arguments via linear interactive proofs.
\newblock In {\em TCC}, 2013.

\bibitem[BCPR16]{BCPR16}
Nir Bitansky, Ran Canetti, Omer Paneth, and Alon Rosen.
\newblock On the existence of extractable one-way functions.
\newblock {\em {SIAM} J. Comput.}, 2016.

\bibitem[BD20]{BD20}
Zvika Brakerski and Nico D{\"{o}}ttling.
\newblock Hardness of {LWE} on general entropic distributions.
\newblock In {\em {EUROCRYPT}}, 2020.

\bibitem[BGV12]{BGV12}
Zvika Brakerski, Craig Gentry, and Vinod Vaikuntanathan.
\newblock ({L}eveled) fully homomorphic encryption without bootstrapping.
\newblock In {\em {ITCS}}, 2012.

\bibitem[BISW17]{BISW17}
Dan Boneh, Yuval Ishai, Amit Sahai, and David~J. Wu.
\newblock Lattice-based {SNARG}s and their application to more efficient
  obfuscation.
\newblock In {\em EUROCRYPT}, 2017.

\bibitem[BLP{\etalchar{+}}13]{BLPRS13}
Zvika Brakerski, Adeline Langlois, Chris Peikert, Oded Regev, and Damien
  Stehl{\'{e}}.
\newblock Classical hardness of learning with errors.
\newblock In {\em {STOC}}, 2013.

\bibitem[BSCS16]{BCS16}
Eli Ben-Sasson, Alessandro Chiesa, and Nicholas Spooner.
\newblock Interactive oracle proofs.
\newblock In {\em {TCC-B}}, 2016.

\bibitem[CB98]{CB98}
Anthony Chefles and Stephen~M. Barnett.
\newblock Optimum unambiguous discrimination between linearly independent
  symmetric states.
\newblock {\em Phys. Lett. A}, 1998.

\bibitem[CDPR16]{CDPR16}
Ronald Cramer, L{\'{e}}o Ducas, Chris Peikert, and Oded Regev.
\newblock Recovering short generators of principal ideals in cyclotomic rings.
\newblock In {\em {EUROCRYPT}}, 2016.

\bibitem[CDW21]{CDW21}
Ronald Cramer, L{\'{e}}o Ducas, and Benjamin Wesolowski.
\newblock Mildly short vectors in cyclotomic ideal lattices in quantum
  polynomial time.
\newblock {\em J. {ACM}}, 2021.

\bibitem[CKKK23]{CKKK23}
Heewon Chung, Dongwoo Kim, Jeong~Han Kim, and Jiseung Kim.
\newblock Amortized efficient {zk-SNARK} from linear-only {RLWE} encodings.
\newblock {\em J. Comm. Netw.}, 2023.

\bibitem[CLZ22]{CLZ22}
Yilei Chen, Qipeng Liu, and Mark Zhandry.
\newblock Quantum algorithms for variants of average-case lattice problems via
  filtering.
\newblock In {\em EUROCRYPT}, 2022.

\bibitem[CT23]{CT23}
Andr{\'{e}} Chailloux and Jean{-}Pierre Tillich.
\newblock The quantum decoding problem.
\newblock 2023.
\newblock Available at \url{https://eprint.iacr.org/2023/1686}.

\bibitem[Dam91]{Dam91}
Ivan Damg{\aa}rd.
\newblock Towards practical public key systems secure against chosen ciphertext
  attacks.
\newblock In {\em {CRYPTO}}, 1991.

\bibitem[DRT23]{DRT21}
Thomas {Debris-Alazard}, Maxime Remaud, and Jean{-}Pierre Tillich.
\newblock Quantum reduction of finding short code vectors to the decoding
  problem.
\newblock June 2023.
\newblock arXiv:2106.02747.

\bibitem[dW23]{dewolf23}
Ronald de~Wolf.
\newblock Quantum computing: Lecture notes, 2023.
\newblock Available at \url{https://arxiv.org/abs/1907.09415}.

\bibitem[GGPR13]{GGPR13}
Rosario Gennaro, Craig Gentry, Bryan Parno, and Mariana Raykova.
\newblock Quadratic span programs and succinct {NIZK}s without {PCP}s.
\newblock In {\em EUROCRYPT}, 2013.

\bibitem[GKPV10]{GKPV10}
Shafi Goldwasser, Yael~Tauman Kalai, Chris Peikert, and Vinod Vaikuntanathan.
\newblock Robustness of the learning with errors assumption.
\newblock In {\em {ICS}}, 2010.

\bibitem[GMNO18]{GMNO18}
Rosario Gennaro, Michele Minelli, Anca Nitulescu, and Michele Orr\`{u}.
\newblock Lattice-based {ZK-SNARKs} from square span programs.
\newblock In {\em {CCS}}, 2018.

\bibitem[GNSV23]{GNS23}
Chaya Ganesh, Anca Nitulescu, and Eduardo Soria-Vazquez.
\newblock Rinocchio: {SNARKs} for ring arithmetic.
\newblock {\em J. Cryptol.}, 2023.

\bibitem[GPV08]{GPV08}
Craig Gentry, Chris Peikert, and Vinod Vaikuntanathan.
\newblock Trapdoors for hard lattices and new cryptographic constructions.
\newblock STOC '08. Association for Computing Machinery, 2008.

\bibitem[GR02]{GR02}
Lov Grover and Terry Rudolph.
\newblock Creating superpositions that correspond to efficiently integrable
  probability distributions, 2002.
\newblock Available at \url{https://arxiv.org/abs/quant-ph/0208112}.

\bibitem[GW11]{GW11}
Craig Gentry and Daniel Wichs.
\newblock Separating succinct non-interactive arguments from all falsifiable
  assumptions.
\newblock In {\em {STOC}}, 2011.

\bibitem[HKM18]{HKM18}
Gottfried Herold, Elena Kirshanova, and Alexander May.
\newblock On the asymptotic complexity of solving {LWE}.
\newblock {\em Des. Codes and Cryptogr.}, 2018.

\bibitem[ISW21]{ISW21}
Yuval Ishai, Hang Su, and David~J. Wu.
\newblock Shorter and faster post-quantum designated-verifier {zkSNARKs} from
  lattices.
\newblock In {\em {CCS}}, 2021.

\bibitem[LMSV12]{LMSV12}
Jake Loftus, Alexander May, Nigel~P. Smart, and Frederik Vercauteren.
\newblock On {CCA}-secure somewhat homomorphic encryption.
\newblock In {\em {SAC}}, 2012.

\bibitem[LMZ23]{LMZ23}
Jiahui Liu, Hart Montgomery, and Mark Zhandry.
\newblock Another round of breaking and making quantum money: How to not build
  it from lattices, and more.
\newblock In {\em {EUROCRYPT}}, 2023.

\bibitem[LPR10]{LPR10}
Vadim Lyubashevsky, Chris Peikert, and Oded Regev.
\newblock On ideal lattices and learning with errors over rings.
\newblock In {\em EUROCRYPT}, 2010.

\bibitem[LS15]{LS15}
Adeline Langlois and Damien Stehl{\'e}.
\newblock Worst-case to average-case reductions for module lattices.
\newblock {\em Des. Codes Cryptogr.}, 2015.

\bibitem[MM11]{MM11}
Daniele Micciancio and Petros Mol.
\newblock Pseudorandom knapsacks and the sample complexity of {LWE}
  search-to-decision reductions.
\newblock In {\em CRYPTO}, 2011.

\bibitem[Nao03]{Naor03}
Moni Naor.
\newblock On cryptographic assumptions and challenges.
\newblock In {\em {CRYPTO}}, 2003.

\bibitem[NC11]{ChuangNielsen}
Michael~A. Nielsen and Isaac~L. Chuang.
\newblock {\em Quantum Computation and Quantum Information: 10th Anniversary
  Edition}.
\newblock Cambridge University Press, 2011.

\bibitem[NYI{\etalchar{+}}20]{NYI20}
Ken Naganuma, Masayuki Yoshino, Atsuo Inoue, Yukinori Matsuoka, Mineaki
  Okazaki, and Noboru Kunihiro.
\newblock Post-quantum zk-{SNARK} for arithmetic circuits using {QAP}s.
\newblock In {\em {AsiaJCIS}}, 2020.

\bibitem[Pei09]{Peikert09}
Chris Peikert.
\newblock Public-key cryptosystems from the worst-case shortest vector problem.
\newblock In {\em {STOC}}, 2009.

\bibitem[Reg09]{Regev09}
Oded Regev.
\newblock On lattices, learning with errors, random linear codes, and
  cryptography.
\newblock {\em J. ACM}, 2009.

\bibitem[SSEK22]{SSEK22}
Ron Steinfeld, Amin Sakzad, Muhammed~F. Esgin, and Veronika Kuchta.
\newblock Private re-randomization for module {LWE} and applications to
  quasi-optimal {ZK-SNARKs}, 2022.
\newblock Available at \url{https://eprint.iacr.org/2022/1690}.

\bibitem[SSTX09]{SSTX09}
Damien Stehl{\'e}, Ron Steinfeld, Keisuke Tanaka, and Keita Xagawa.
\newblock Efficient public key encryption based on ideal lattices.
\newblock In {\em ASIACRYPT}, 2009.

\bibitem[Wat18]{watrous18}
John Watrous.
\newblock {\em The Theory of Quantum Information}.
\newblock Cambridge University Press, 2018.

\bibitem[WW23]{WW23}
Hoeteck Wee and David~J. Wu.
\newblock Lattice-based functional commitments: Fast verification and
  cryptanalysis.
\newblock In {\em ASIACRYPT}, 2023.

\end{thebibliography}
\newcommand{\etalchar}[1]{$^{#1}$}

\appendix
\newpage

\section{Quantum Unambiguous Measurement from \cite{CB98}}\label{app:CB98}

\subsection{Positive operator-valued measures}

Positive Operator-Valued Measures (POVM) are defined as follows. They are the most general measurements allowed within quantum information theory.

\begin{definition}[POVM measurements] A \textup{POVM} is a set $\{\vec{E}_{i}\}_{i \in \mathcal{I}}$ of positive operators where $\mathcal{I}$ is the set of measurement outcomes and the operators satisfy $\sum_{i}\vec{E}_{i} = \vec{Id}$. A measurement upon a quantum state $\ket{\psi}$  outputs~$i$ with probability~$\bra{\psi}\vec{E}_{i}\ket{\psi}$. 
\end{definition}

POVMs are sometimes considered in the following situation: given a set of quantum states $\ket{\psi_1},\ldots,\ket{\psi_N}$, devise a POVM that when applied over $\ket{\psi_j}$, it either outputs the correct index $j$ or some special symbol $\bot$ representing the ``unknown'' answer. In other words, the measurement never makes an error when it succeeds to identify the prepared state and we say that it {\em unambiguously} distinguishes the states $\ket{\psi_j}$'s. The probability of error is defined as the probability that the measurement outputs $\bot$, when it is maximized over all possible input states: 
$$
p_{\bot} \eqdef \max\limits_{k} \bra{\psi_k}\vec{E}_\bot\ket{\psi_k}
$$
where $\vec{E}_{\bot}$ corresponds to the outcome~$\bot$.

 \subsection{Discrimination of coordinate states} \label{subsec:unambDiscr}

 We now describe the POVM from~\cite{CB98}, which is known to be optimal to unambiguously distinguish  the~$\ket{\psi_{k}}$'s (as given in Definition~\ref{def:psi-states}).  Namely, it minimizes the error parameter~$p_{\bot}$ over all possible choice of POVMs. This optimality is enabled by the fact that the $\ket{\psi_k}$'s verify the following ``symmetry'' condition:
 $$
\forall k \in \mathbf{Z}/q\mathbb{Z}, \quad \vec{T}\ket{\psi_{k}} = \ket{\psi_{k+1 \bmod q}} \ , 
 $$
 where $\vec{T}$ denotes the translation operator, \ie{} $\vec{T}\ket{a} =\ket{a+1 \bmod q}$ for all~$a$, and from the fact 
 that they are linearly independent (which is ensured by $\widehat{f}(x) \neq 0$ for all~$x$, as is the case for  our instantiation with the folded Gaussian distribution). 
 
\begin{theorem}[Adapted from~\cite{CB98}]
\label{theo:POVMLWE}
	Let $q $ be an integer and~$f: \mathbb{Z}/q\mathbb{Z}\rightarrow \mathbb{C}$ be an amplitude function such that~$\widehat{f}\left( y\right) \neq 0$ for every~$y \in \mathbb{Z}/q\mathbb{Z}$. 
	Let
	\begin{equation}\label{eq:psiperp} 
		\ket{\psi_{j}^{\perp}} \eqdef \frac{1}{\sqrt{N}} \sum_{y \in \mathbb{Z}/q\mathbb{Z}}  \overline{\widehat{f}(-y)^{-1}} \; \omega_{q}^{-jy} \ket{\chi_{y}}, \mbox{ where } \; N \eqdef \sum_{y \in \mathbb{Z}/q\mathbb{Z}}^{q-1} |\widehat{f}(y)|^{-2}\ ,
	\end{equation} 
	and
	$$
	\forall j \in \mathbb{Z}/q\mathbb{Z}, \quad \vec{E}_{j} \eqdef \frac{1}{\lambda_{+}} \ketbra{\psi_{j}^{\perp}}{\psi_{j}^{\perp}}, \text{ and  }\vec{E}_{\bot} \eqdef \vec{I} - \sum_{j \in \mathbb{Z}/q\mathbb{Z}} \vec{E}_{j}\ ,
	$$
	where $\lambda_{+}$ is the maximum eigenvalue of $\sum_{j \in \mathbb{Z}/q\mathbb{Z}}\ketbra{\psi_{j}^{\perp}}{\psi_{j}^{\perp}}$.
	Then the set $\{\vec{E}_{j}\}_{j \in (\mathbb{Z}/q\mathbb{Z}) \cup \{ \bot \}}$ is a \textup{POVM} that unambiguously distinguishes the coordinate states with success probability~$p$ as follows (it is independent  of~$j$):
	$$
	p = \bra{\psi_j}\vec{E}_{j} \ket{\psi_j} = q \cdot \min\limits_{y \in \mathbb{Z}/q\mathbb{Z}}  \left| \widehat{f}\left( y\right)\right|^{2} \ .
	$$ 
\end{theorem}

Representating the coordinate states in the Fourier basis is helpful to approach the problem. The first lemma shows that~$\ket{\psi_{i}^{\perp}}$ defined as in Equation~\eqref{eq:psiperp} is a quantum state orthogonal to all~$\ket{\psi_{j}}$'s where~$i\neq j$. 

\begin{lemma}\label{lemma:psijPerp}
	Using the notations of Theorem \ref{theo:POVMLWE}, we have:
	\begin{equation*}
		\forall i,j \in  \mathbb{Z}/q\mathbb{Z}, \quad \braket{\psi_{i}^{\perp}}{\psi_{j}} = \left\{
		\begin{array}{ll}
			\frac{q}{\sqrt{N}} & \mbox{if } j = i \\
			0 & \mbox{otherwise}
		\end{array}	
		\right.  \ .
	\end{equation*} 
\end{lemma}

\begin{proof}	
	Let us write the $\ket{\psi_{j}}$'s in the Fourier basis. We have for all $j \in \mathbb{Z}/q\mathbb{Z}$:
	\begin{align}
		\ket{\psi_{j}} &= \sum_{e \in \mathbb{Z}/q\mathbb{Z}} f\left(e\right)\ket{j+e\bmod q} \nonumber\\
		&= \frac{1}{\sqrt{q}}  \sum_{e \in \mathbb{Z}/q\mathbb{Z}} f\left(e\right) \sum_{x \in \mathbb{Z}/q\mathbb{Z}} \omega_{q}^{-(j+e)x}\ket{\chi_{x}} \quad \left(\mbox{by Lemma \ref{lemma:qftm1}} \right) \nonumber \\
		&= \sum_{x \in \mathbb{Z}/q\mathbb{Z}} \left( \frac{1}{\sqrt{q}} \sum_{e \in \mathbb{Z}/q\mathbb{Z}} f\left(e \right) \omega_{q}^{-xe} \right) \omega_{q}^{-jx} \ket{\chi_{x}} \nonumber \\
		&= \sum_{x \in \mathbb{Z}/q\mathbb{Z}} \widehat{f}\left(-x\right) \omega_{q}^{-jx} \ket{\chi_{x}} \ . \nonumber
	\end{align}
	We thus have, for all $i \in \mathbb{Z}/q\mathbb{Z}$:
	\begin{equation*}
		\braket{\psi_{i}^{\perp}}{\psi_{j}} = \frac{1}{\sqrt{N}} \sum_{x\in \mathbb{Z}/q\mathbb{Z}} \omega_{q}^{x(i-j)} = \left\{
		\begin{array}{cc}
			\frac{q}{\sqrt{N}} & \mbox{if } j = i \\
			0 & \mbox{otherwise}
		\end{array}	
		\right. \ .
	\end{equation*} 
	This completes the proof. 
\end{proof}

We now consider the maximum eigenvalue $\lambda_{+}$ of $\sum_{j \in \mathbb{Z}/q\mathbb{Z}} \ketbra{\psi_{j}^{\perp}}{\psi_{j}^{\perp}}$. 

\begin{lemma}\label{lemma:lambda+}
	Using notations of Theorem \ref{theo:POVMLWE}, we have: 
	$$
	\lambda_{+} = \frac{q}{N} \; \frac{1}{\min\limits_{x \in  \mathbb{Z}/q\mathbb{Z}}  \left| \widehat{f}\left( x\right)\right|^{2}} \enspace .
	$$
\end{lemma}

\begin{proof}
	We have the following equalities:
	\begin{align*}
		\sum_{j \in \mathbb{Z}/q\mathbb{Z}} \ketbra{\psi_{j}^{\perp}}{\psi_{j}^{\perp}} &= \frac{1}{N} \sum_{j \in \mathbb{Z}/q\mathbb{Z}} \left( \sum_{x \in \mathbb{Z}/q\mathbb{Z}}\; \overline{\widehat{f}(-x)^{-1}}\; \omega_{q}^{-jx} \ket{\chi_{x}} \right) \left( \sum_{y \in \mathbb{Z}/q\mathbb{Z}}\; \widehat{f}(-y)^{-1} \; \omega_{q}^{jy} \bra{\chi_{y}} \right) \\
		&= \frac{1}{N} \sum_{x,y \in \mathbb{Z}/q\mathbb{Z}} \left( \sum_{j \in \mathbb{Z}/q\mathbb{Z}} w_{q}^{j(y-x)} \right) \; \overline{\widehat{f}(-x)^{-1}} \; \widehat{f}(-y)^{-1} \; \ketbra{\chi_{x}}{\chi_{y}} \\
		&= \frac{q}{N} \sum_{x \in \mathbb{Z}/q\mathbb{Z}} |\widehat{f}(-x)|^{-2} \ketbra{\chi_{x}}{\chi_{x}} \ .
	\end{align*}
	Therefore, as the $\ket{\chi_{x}}$'s define an orthonormal basis of the underlying Hilbert space, we obtain
	$$
	\lambda_{+} = \frac{q}{N} \; \frac{1}{\min\limits_{x \in \mathbb{Z}/q\mathbb{Z}}  \left| \widehat{f}\left(  x\right)\right|^{2}} \enspace .
	$$
	This completes the proof. 
\end{proof}

\begin{proof}[Proof of Theorem \ref{theo:POVMLWE}] The fact that $\{\vec{E}_{j}\}_{j \in  (\mathbb{Z}/q\mathbb{Z})\cup \{ \bot\}}$ defines a POVM follows from the definition of $\lambda_{+}$: they are positive operators and sum to the identity.

	By Lemma \ref{lemma:psijPerp}, the state $\ket{\psi_{i}^{\perp}}$ is orthogonal to $\ket{\psi_{j}}$ for all $j \neq i$. Therefore, given $\ket{\psi_{j}}$, the probability to successfully measure $j$ with the POVM $\{\vec{E}_{i}\}_{i \in  (\mathbb{Z}/q\mathbb{Z})\cup \{ \bot\}}$
	is given by
	\begin{equation*}
		p = \bra{\psi_{j}}\vec{E}_{j}\ket{\psi_{j}} = \frac{1}{\lambda_{+}} \left| \braket{\psi_{j}^{\perp}}{\psi_{j}}\right|^{2} = \frac{q^{2}}{\lambda_{+}\; N} =  q \cdot \min\limits_{y \in \mathbb{Z}/q\mathbb{Z}}  \left| \widehat{f}\left( y\right)\right|^{2},
	\end{equation*}
	where the two last equalities follow from Lemmas \ref{lemma:psijPerp} and \ref{lemma:lambda+}. 
\end{proof}

\section{Proof of Lemma~\ref{lemma:noPhaseDG}} \label{app:GD}

Recall 
that
\[
\forall e \in \mathbb{Z}: \vartheta_{\sigma, q} (e) = \frac{1}{\rho_\sigma(\mathbb{Z})}\sum_{k \in \mathbb{Z}} \exp\left(-\frac{|e+qk|^2}{\sigma^2}\right) \ .
\]
Let~$A,B:\mathbb{Z}/q\mathbb{Z} \rightarrow \mathbb{C}$ be defined as follows:
 \begin{align*}
\forall y \in \mathbb{Z}/q\mathbb{Z}: A(y) &\eqdef \sum\limits_{x \in \mathbb{Z} \cap (-q/2, q/2]} \omega_q^{xy} \left(\sqrt{\vartheta_{\sigma,q}(x)} - \sqrt{D_{\mathbb{Z},\sigma}(x)} \right), \\
\forall y \in \mathbb{Z}/q\mathbb{Z}: B(y) &\eqdef \sum\limits_{x \in \mathbb{Z} \cap (-q/2, q/2]} \omega_q^{xy} \left(\frac{\sum_{k \in \mathbb{Z}}\rho_{\sqrt{2}\sigma}(x + kq)}{\sqrt{\rho_{\sigma}(\mathbb{Z})}} - \sqrt{D_{\mathbb{Z},\sigma}(x)}\right). 
\end{align*}
Then, for all~$y \in \mathbb{Z}/q\mathbb{Z}$, it holds that
 \begin{align*}
 \widehat{f_0}(y) &= \frac{1}{\sqrt{q}} \sum_{x  \in \mathbb{Z}/q\mathbb{Z}} \omega_{q}^{xy} \ \sqrt{\vartheta_{\sigma,q} (x)}\\
 &=  \frac{1}{\sqrt{q}} \left(A(y) - B(y) + \sum_{x  \in \mathbb{Z}/q\mathbb{Z}} \omega_q^{xy} \ \frac{\sum_{k \in \mathbb{Z}}\rho_{\sqrt{2}\sigma}(x + kq)}{\sqrt{\rho_{\sigma}(\mathbb{Z})}}  \right).
 \end{align*}
By the Poisson summation formula, the above term is equal to:
 
 \begin{equation}
 \frac{1}{\sqrt{q}} \bigg(A(y) - B(y) + \sum_{\ell  \in \mathbb{Z}} \omega_q^{\ell y} \ \frac{\rho_{\sqrt{2}\sigma}(\ell)}{\sqrt{\rho_{\sigma}(\mathbb{Z})}}   \bigg) 
 = \frac{1}{\sqrt{q}} \bigg(A(y) - B(y) + \frac{\sqrt{2}\sigma}{\sqrt{\rho_{\sigma}(\mathbb{Z})}} \sum_{\ell  \in \mathbb{Z}}  \rho_{\frac{1}{\sqrt{2}\sigma}}\Big(\ell + \frac{y}{q}\Big)  \bigg) \ . \label{eq:psf-vartheta}
 \end{equation}
 
 We now find upper bounds for the terms~$A(y)$ and~$B(y)$ and a lower bound for the remaining term of Equation~\eqref{eq:psf-vartheta}.
 Using the fact that~$\sqrt{\rho_{\sigma}} = \rho_{\sqrt{2}\sigma}$, we have, for all~$y \in \mathbb{Z}/q\mathbb{Z}$: 
 \begin{align*}
 B(y) &= \sum\limits_{x \in \mathbb{Z} \cap (-q/2, q/2]} \omega_q^{xy} \left(\frac{\sum_{k \in \mathbb{Z}}\rho_{\sqrt{2}\sigma}(x + kq)}{\sqrt{\rho_{\sigma}(\mathbb{Z})}} - \sqrt{D_{\mathbb{Z},\sigma}(x)}\right) \\
 &= \frac{\rho_{\sqrt{2}\sigma}(\mathbb{Z})}{\sqrt{\rho_{\sigma}(\mathbb{Z})}}  \sum\limits_{x \in \mathbb{Z} \cap (-q/2, q/2]} \omega_q^{xy} \left(\vartheta_{\sqrt{2}\sigma,q}(x) - D_{\mathbb{Z},\sqrt{2}\sigma}(x)\right).
 \end{align*}
By the triangular inequality, it follows that
 \begin{align*}
 |B(y)| &\leq  \frac{\rho_{\sqrt{2}\sigma}(\mathbb{Z})}{\sqrt{\rho_{\sigma}(\mathbb{Z})}} \sum\limits_{x \in \mathbb{Z} \cap (-q/2, q/2]} \left(\vartheta_{\sqrt{2}\sigma,q}(x) - D_{\mathbb{Z},\sqrt{2}\sigma}(x)\right) \\
 &\leq \frac{\rho_{\sqrt{2}\sigma}(\mathbb{Z})}{\sqrt{\rho_{\sigma}(\mathbb{Z})}} \ q \ \mathrm{e}^{-\frac{q^2}{8\sigma^2}} \quad \text{(by Lemma~\ref{lemma:vartheta-rho})}  \ .
 \end{align*}
The use of Lemma~\ref{lemma:vartheta-rho} requires that~$\sigma \leq q/2$, which is implied by our assumptions. 
 
We also have, for all~$y \in \mathbb{Z}/q\mathbb{Z}$: 
 \begin{align*}
 |A(y)| 
 & \leq \sum\limits_{x \in \mathbb{Z} \cap (-q/2, q/2]} \left(\sqrt{\vartheta_{\sigma,q}(x)} - \sqrt{D_{\mathbb{Z},\sigma}(x)} \right) \quad \text{(by the triangular inequality)}\\
 & \leq \sum\limits_{x \in \mathbb{Z} \cap (-q/2, q/2]} \mathrm{e}^{-\frac{q^2}{8\sigma^2}} \quad \text{(by Lemma~\ref{lemma:vartheta-rho})}\\
 & = q \ \mathrm{e}^{-\frac{q^2}{8\sigma^2}} \ .
 \end{align*}

Further, for every~$y \in \mathbb{Z}$, it holds that
\[
\sum_{\ell  \in \mathbb{Z}} \ \rho_{\frac{1}{\sqrt{2}\sigma}}\left(\ell + \frac{y}{q}\right) \geq \mathrm{e}^{- \pi \frac{\sigma^{2}}{8}} \enspace .
\]
To see this, note that the sum contains at least one term~$\ell + y/q$ that has absolute value~$\leq 1/2$. 

Going back to Equation~\eqref{eq:psf-vartheta} and using the triangular inequality, we see that, for all~$y \in \mathbb{Z}/q\mathbb{Z}$:
 \begin{eqnarray*}
 |\widehat{f_0}(y)| & \leq \frac{1}{\sqrt{q}} \left(\frac{\sqrt{2}\sigma}{\sqrt{\rho_{\sigma}(\mathbb{Z})}} \mathrm{e}^{- \pi \frac{\sigma^{2}}{8}}  
 + q \ \mathrm{e}^{-\frac{q^2}{8\sigma^2}} + \frac{\rho_{\sqrt{2}\sigma}(\mathbb{Z})}{\sqrt{\rho_{\sigma}(\mathbb{Z})}} \ q \ \mathrm{e}^{-\frac{q^2}{8\sigma^2}} \right) \\
 & \leq \frac{1}{\sqrt{q}} \left(
\sqrt{2\sigma}    \mathrm{e}^{- \pi \frac{\sigma^{2}}{8}}  
 + q \ \mathrm{e}^{-\frac{q^2}{8\sigma^2}} + \frac{\sqrt{2}\sigma+1}{\sqrt{\sigma}} \ q \ \mathrm{e}^{-\frac{q^2}{8\sigma^2}} 
  \right) \\
  & \leq \frac{4\sqrt{\sigma}}{\sqrt{q}} \left(
    \mathrm{e}^{- \pi \frac{\sigma^{2}}{8}}  
 +  q \ \mathrm{e}^{-\frac{q^2}{8\sigma^2}} 
  \right) 
  \ ,
 \end{eqnarray*}
 where the second inequality follows from Lemma~\ref{lemma:rho-Z-bound} and the third one from~$\sigma\geq 1$. 
 \qed

\end{document}